%% file: main.tex
\documentclass{llncs}
\usepackage[hmarginratio=1:1,marginpar=2cm]{geometry}

\newcommand{\norm}[1]{}
\usepackage[utf8]{inputenc}
\usepackage{graphicx}
\usepackage{listings}
\usepackage{amssymb}
\usepackage{dsfont}
\usepackage{mathtools}
\usepackage{algorithm}
\usepackage{algpseudocode}
\usepackage{xcolor}
\usepackage{multirow}
\usepackage{verbatim}
\usepackage{amsmath}
\usepackage[english]{babel}
\usepackage{wasysym}
\usepackage{qip}
\usepackage{qm}
\usepackage{enumitem}
\usepackage{environ}
\usepackage{booktabs}
\usepackage{pifont}
\usepackage[
	n,
	operators,
	advantage, 
	sets,
	adversary,
	landau,
	probability,
	notions,
	logic,
	ff,
	mm,
    primitives,
    events,
    complexity,
    asymptotics, 
    keys]{cryptocode}

\usepackage{xcolor}
\usepackage[normalem]{ulem}
\usepackage[colorlinks=true, linkcolor=black, citecolor=green]{hyperref}   
\usepackage[nameinlink,capitalise]{cleveref}

\spnewtheorem{Claim}{Claim}{\bfseries}{\itshape}
\usepackage[capitalise]{cleveref}

\input{BQSM_commands}

\input{commands}

\begin{document}
\title{On the Feasibility of Simulation-Secure Quantum Functional Encryption}
\author{Mohammed Barhoush\inst{1} \and Arthur Mehta\inst{2} \and Anne Müller\inst{3,4} \and Louis Salvail\inst{1}}
\institute{Universit\'e de Montr\'eal (DIRO), Montr\'eal, Canada\\  \email{mohammed.barhoush@umontreal.ca}, \email{salvail@iro.unmontreal.ca}  \and University of Ottawa, Canada\\ \email{amehta2@uottawa.ca} \and
CISPA Helmholtz Center for Information Security, Saarbrücken, Germany
\and
Saarland University, Saarbrücken, Germany\\
\email{anne.mueller@cispa.de}}
\maketitle

\begin{abstract}
Functional encryption is a powerful cryptographic primitive that enables fine-grained access to encrypted data and underlies numerous applications. Although the ideal security notion for $\FE$---simulation security---has been shown to be impossible in the plain classical setting, those impossibility results rely on inherently classical arguments. This raises a fundamental question: can simulation-secure \textsf{FE} be realized in the quantum regime or in alternative models?

We resolve both directions of this question. On the negative side, we show that the classical impossibility results largely extend to the quantum setting. When the adversary can submit an unbounded number of challenge messages, we establish an unconditional impossibility, matching the classical barrier. In the case where the adversary may  obtain many functional keys, classical arguments only yield impossibility under the assumption of pseudorandom
functions; we strengthen this by proving impossibility under the potentially weaker assumption of pseudorandom quantum states. We additionally give an independent impossibility based solely on public-key encryption with classical public-keys. Since public-key encryption is not known to imply pseudorandom states, this yields complementary evidence of the barrier and, in particular, establishes unconditional impossibility of public-key functional encryption with classical public-keys. A key technical ingredient in our proofs is a novel incompressibility property for pseudorandom states, which may be of independent interest.

On the positive side, we show how to circumvent these fundamental barriers in memory-restricted models, where the adversary's quantum or classical memory is bounded. In this setting, we obtain new feasibility results for simulation-secure \textsf{FE} that allow many challenge ciphertexts or many functional key queries.
\end{abstract}

    \setcounter{secnumdepth}{3}
    \setcounter{tocdepth}{3}
    \newpage
   \tableofcontents 
 \newpage
    
\input{intro_reviewed}

\input{tech_overview_last_reviewed}

\input{preliminaries_reviewed}

\input{first_imposs_reviewed}

\input{section5_reviewed}

\input{section6_reviewed}

\input{section7_reviewed}

\input{section8_audit_reviewed}

\subsection*{Author contributions}
  The authors are listed in alphabetical order. All authors contributed to the development of the results and the writing of the manuscript.
\bibliographystyle{splncs04}
\bibliography{mybib}

\appendix

\input{appendix_audit_reviewed}

\end{document}

%% file: BQSM_commands.tex
\newcommand{\OWF}{\ensuremath{\textsf{OWF}}}

\newcommand{\s}{\ensuremath{\texttt{s}}}

\newcommand{\xx}{\ensuremath{\texttt{m}}}

\newcommand{\PRF}{\ensuremath{\mathsf{PRF}}}
\newcommand{\ct}{\ensuremath{\mathsf{ct}}}
\newcommand{\GB}{\ensuremath{\mathsf{GBO}}}
\newcommand{\FE}{\ensuremath{\mathsf{FE}}}
\newcommand{\BCSM}{\ensuremath{\mathsf{BCSM}}}
\newcommand{\PKE}{\ensuremath{\mathsf{PKE}}}
\newcommand{\BQSM}{\ensuremath{\mathsf{BQSM}}}

\newcommand{\SIM}{\textsf{DF}\text{-}\textsf{SIM}}

\newcommand{\QFE}{\textsf{QFE}}

\newcommand{\Hy}{\ensuremath{\textsf{H}}}

\newcommand{\ek}{\ensuremath{\textsf{ek}}}
\newcommand{\mk}{\ensuremath{\textsf{mk}}}
\newcommand{\nk}{\ensuremath{\textsf{k}}}

\newcommand{\Sim}{\ensuremath{\mathcal{S}}}

%% file: commands.tex
\newcommand{\anne}[1]{{\color{green}[\textbf{Anne}: #1]}}

\spnewtheorem{fact}{Fact}{\bfseries}{\itshape}
\crefname{fact}{Fact}{Facts}
\spnewtheorem{construct}{Construction}{\bfseries}{\itshape}
\spnewtheorem*{thm}{Theorem}{\bfseries}{\itshape}
\spnewtheorem*{lem}{Lemma}{\bfseries}{\itshape}

\newcommand{\N}{\mathbb{N}}
\newcommand{\C}{\mathbb{C}}
\newcommand{\Hcal}{\mathcal{H}}
\newcommand{\Tr}{\ensuremath{\mathsf{Tr}}}
\newcommand{\T}{\ensuremath{\mathsf{T}}}
\newcommand{\Enc}{\ensuremath{\mathsf{Enc}}}
\newcommand{\Dec}{\ensuremath{\mathsf{Dec}}}
\newcommand{\Setup}{\ensuremath{\mathsf{Setup}}}
\newcommand{\Keygen}{\ensuremath{\mathsf{KeyGen}}}

\newcommand{\PRS}{\ensuremath{\mathsf{PRS}}}

\newcommand{\Comp}{\mathsf{Comp}}
\newcommand{\Decomp}{\mathsf{Decomp}}

\newcommand{\Favg}{\ensuremath{\mathsf{F}_\mathsf{avg}}}
\newcommand{\Fe}{\ensuremath{\mathsf{F}_\mathsf{E}}}
\newcommand{\Id}{\ensuremath{\mathcal{I}}}
\newcommand{\Expect}{\ensuremath{\mathbb{E}}}

\newcommand{\Gen}{\ensuremath{\mathsf{Gen}}}
\newcommand{\PKEnc}{\ensuremath{\mathsf{PKE.Enc}}}
\newcommand{\PKDec}{\ensuremath{\mathsf{PKE.Dec}}}
\newcommand{\PKKeygen}{\ensuremath{\mathsf{PKE.Gen}}}

\newcommand{\Exp}{\mathsf{Exp}}
\newcommand{\Adv}{\ensuremath{\mathcal{A}}}
\newcommand{\ch}{\ensuremath{\mathcal{C}}}
\newcommand{\D}{\ensuremath{\mathcal{D}}}

\newcommand{\msg}{\ensuremath{\mathsf{m}}}

\newcommand{\cmark}{\ding{51}} 
\newcommand{\xmark}{\ding{55}} 

\renewcommand{\qed}{\ensuremath{\blacksquare}}

%% file: intro_reviewed.tex
\providecommand{\QFEreplace}[2]{\sout{#1}\ {\color{red}#2}}
\providecommand{\QFEinsert}[1]{{\color{red}\relax#1}}
\providecommand{\QFEmathreplace}[2]{\text{\sout{\ensuremath{#1}}}\;{\color{red}#2}}

\section{Introduction}
Traditional encryption offers an all-or-nothing guarantee: a ciphertext can either be fully decrypted or not at all. While this suffices for basic confidentiality, it falls short in modern applications—such as cloud services, data outsourcing, or fine-grained access control—where different parties should be able to access only specific aspects of the data. Functional encryption (FE) was proposed as a solution to this limitation, first informally introduced by Sahai and Waters \cite{SW05} and later developed into a formal framework  by Boneh, Sahai and Waters \cite{BSW11}.

In an $\FE$ scheme, a master secret key $\mk$ enables the generation of functional keys ${\fk}_F$ associated with functions $F$ from a given family. Possession of ${\fk}_F$ allows a user to compute $F(m)$ from an encryption of $m$, but nothing beyond this prescribed output. The precise interpretation of ``nothing beyond'' has been the subject of extensive study \cite{BSW11,AGV13}, and plays a central role in shaping both the feasibility and the strength of $\FE$ constructions.

Two central paradigms have emerged for formalizing security in this context: simulation-based security (\textsf{SIM}--security) and the weaker but more permissive indistinguishability-based security (\textsf{IND}--security). Each of these can be combined with additional dimensions: adaptive versus non-adaptive queries to the functional-key derivation oracle, single versus multiple challenge ciphertexts, and single versus multiple functional key queries. Together, these refinements yield 16 distinct security notions, denoted by $q\text{-}x\text{-}y\text{-}z$, where $q \in \{1, \textsf{M}\}$ specifies a single or many ciphertext queries, $x \in \{1, \textsf{M}\}$ a single or multiple functional key queries, $y \in \{\textsf{NA}, \textsf{AD}\}$ the adaptivity, and $z \in \{\textsf{IND}, \textsf{SIM}\}$ the underlying framework.\footnote{Here $\textsf{M}$ refers to many or, more specifically, polynomially many queries where the exact number is not known at setup time. The notion of an a-priori bounded number of queries has also been explored in classical FE constructions~\cite{GVW12}.}

Among these variants, $\textsf{SIM}$--security is the stronger and more desirable goal, since $\textsf{IND}$--security becomes vacuous for certain classes of functionalities \cite{BSW11}. Unfortunately, a series of impossibility results \cite{BSW11,AGV13} show that $\textsf{SIM}$--secure (classical) $\FE$ cannot be achieved in full generality. In particular, the notions of $\textsf{M}\text{-}1\text{-}\textsf{AD}$ \textsf{SIM}--security and $1\text{-}\textsf{M}\text{-}\textsf{NA}$ \textsf{SIM}--security are provably unattainable for general polynomial circuits. In essence, this means that security breaks down once the adversary is allowed either numerous non-adaptive functional key queries or, alternatively, multiple ciphertext queries together with a single adaptive functional key query.

Crucially, these impossibility results apply only in the \textbf{classical} setting. Their proofs rely on inherently classical arguments that do not directly carry over to the quantum world. This raises a central question: 
\begin{center}
\textit{Do the classical impossibilities of $\textsf{SIM}$--secure $\FE$ generalize to the quantum regime?}     
\end{center}

The history of quantum cryptography offers strong evidence that classical impossibility does not always imply quantum impossibility. A range of primitives---such as unclonable encryption, certified deletion, unforgeable money---are provably impossible classically, yet attainable quantumly. A particularly notable example is the recent result showing that \emph{quantum} public keys enable public key encryption from one-way functions~\cite{BGH23}, in contrast to the classical setting where a black-box separation is known~\cite{IR89}. These examples suggest that quantum resources can overcome barriers that appear absolute in the classical world. By the same analogy, it is natural to ask whether quantum public keys, ciphertexts, or functional keys might provide the extra layer of security needed to realize simulation-based notions of functional encryption.


Next, to generalize our exploration further, we consider alternative models where the impossibility results do not necessarily  hold. Specifically, we look into simulation-secure \textsf{FE} in the bounded storage models. In these models, the adversary's classical or quantum memory is assumed to be bounded. These models have been extensively studied to build other cryptographic applications prior to this work \cite{m92,GZ19,DQW21,R18,DQW22,dfss05,BS23}. Notably, the classical impossibility result implicitly assumes that the adversary can store either many ciphertexts or many functional keys. Hence, the feasibility of \textsf{SIM}--security in memory bounded models is not clear, motivating the following question:
\begin{center}
\textit{Do the impossibilities of $\textsf{SIM}$--secure $\FE$ generalize to the memory bounded models?}     
\end{center}

\subsection{Our Contribution}

In our work, we show that the classical barriers for simulation-secure $\FE$ largely extend to the quantum realm. Yet, we show that they do not extend to the memory bounded models, where we provide actual constructions. 

\paragraph{Impossibilities in the Plain Quantum Model.}

We first discuss the plain quantum regime, in which the encrypted data may be quantum states and the evaluated functionalities may be quantum circuits.
We provide three impossibility results on \textsf{SIM}--secure Quantum Functional Encryption ($\QFE$), showing:

\begin{enumerate}
    \item Unconditional impossibility of $\textsf{M}\text{-}1\text{-}\textsf{AD}$ \textsf{SIM}--secure $\QFE$.
    \item Impossibility of \textit{succinct} $1\text{-}1\text{-}\textsf{NA}$ \textsf{SIM}--secure $\QFE$ assuming (sufficiently-large) pseudorandom quantum states. 
\item Impossibility of $1\text{-}\textsf{M}\text{-}\textsf{NA}$  \textsf{SIM}--secure $\QFE$ assuming public key encryption with classical keys ($\PKE$).
\end{enumerate}

{All the results are shown in the secret-key setting with classical master secret keys, which implies impossibility for public-key $\QFE$ within this model.} We now compare these results to the classical impossibility landscape for $\FE$. Boneh et al.~\cite{BSW11} proved that $\textsf{M}\text{-}1\text{-}\textsf{AD}$ \textsf{SIM}--secure $\FE$ is unconditionally impossible. Our first result establishes the same barrier in the quantum setting, though the proof requires a new basic lemma demonstrating the limitations of compressing large amounts of classical information into a small quantum state (\cref{lem:incompress}). 

Next, Agrawal et al.~\cite{AGV13} showed that $1\text{-}\textsf{M}\text{-}\textsf{NA}$ \textsf{SIM}--secure $\FE$ is impossible under the assumption of pseudorandom functions ($\PRF$). Since $\PRF$s can be constructed from any $\FE$ scheme, this immediately yields an unconditional impossibility result. However, this reasoning breaks down in the quantum setting: $\QFE$ is not known to imply $\PRF$, as the classical construction via derandomization of the $\Keygen$ algorithm cannot be applied to a quantum algorithm. {Thus, adapting the classical proof using our quantum incompressibility lemma only gives a conditional impossibility for $\QFE$ based on quantum-secure $\PRF$s.}

{We strengthen this by replacing $\PRF$ with pseudorandom quantum states ($\PRS$) for succinct schemes.} This constitutes a potentially weaker assumption than $\PRF$s, since a quantum-secure $\PRF$ implies $\PRS$ \cite{JLS18}, but there is an oracle separation in the other direction \cite{K21}. Moreover, our impossibility applies to succinct $1\text{-}1\text{-}\textsf{NA}$ \textsf{SIM}--secure $\QFE$, where the adversary requires only a single non-adaptive functional key query, rather than the weaker $1\text{-}\textsf{M}\text{-}\textsf{NA}$ \textsf{SIM}--security case ruled out classically. On the downside, our second impossibility only applies to succinct schemes, meaning that the ciphertext sizes should not depend on the maximum size of the admissible circuits.

Our proof introduces new technical tools, including an incompressibility property of $\PRS$ that may be of independent interest. In particular, we show that compressing an $n$-qubit $\PRS$ into $n-k$ qubits {irreversibly loses almost all structure: any efficient decompression procedure recovers a state with expected fidelity at most $2^{-k}$ with the original}.

Finally, our third result shows that  $1\text{-}\textsf{M}\text{-}\textsf{NA}$ \textsf{SIM}--secure $\QFE$ is impossible under the existence of $\PKE$. Unlike our previous result, this result holds even for non-succinct schemes. As we mentioned, Agrawal et al. \cite{AGV13} already proved an unconditional impossibility for this class of $\FE$ schemes in the classical setting, whereas in the quantum setting our $\PRS$--based proof only applied to succinct schemes. This motivated us to explore alternative assumptions leading to impossibility, and we found that any $\PKE$ scheme implies impossibility of such a $\QFE$ scheme. {Importantly, $\PKE$ and $\PRS$ are incomparable: there does not exist a black-box construction of perfect-complete $\PKE$ from one-way functions (and thus from $\PRS$) \cite{LLLL25}, while no implication or separation in the opposite direction is known.} The two results use different assumptions and apply to different
regimes: the $\PRS$-based result requires succinctness, whereas
the $\PKE$-based result does not.  Hence, our second and third impossibility results complement one another, together providing strong evidence against the existence of \textsf{SIM}--secure $\QFE$. Moreover, our $\PKE$--based result allows us to show unconditionally that $1\text{-}\textsf{M}\text{-}\textsf{NA}$ \textsf{SIM}--secure public key $\QFE$ with classical public keys is impossible, given that this implies $\PKE$. Notably, our $\PKE$--based impossibility proof is entirely new, with no classical analogue.

\begin{table}[]
    \centering
\begin{tabular}{|c|c|c|c|c|}
\hline
& \begin{tabular}{@{}c@{}}one-message\\ one-query\end{tabular}
& \begin{tabular}{@{}c@{}}many-message\\ one-query\end{tabular}
& \begin{tabular}{@{}c@{}}one-message\\ many-query\end{tabular}
& \begin{tabular}{@{}c@{}}many-message\\ many-query\end{tabular} \\
\hline
\textbf{non-adaptive} 
& \cmark$^{a}$  
& \cmark$^{a}$ 
&  \xmark\ \cref{thm:imp1MNA} 
& \xmark $^{b}$ \\
\hline
\textbf{adaptive} 
& \cmark$^{a}$ 
&  \xmark\ \cref{thm:impM1AD}
& \xmark $^{b}$ 
& \xmark $^{b}$ \\

\hline
\hline
\textbf{succinct (NA)} 
& \xmark \ \cref{thm:imp11NA}
&  \xmark$^{b}$
& \xmark  $^{b}$
& \xmark$^{b}$  \\
\hline

\end{tabular}
    \caption{Possibilities and impossibilities for \textsf{SIM}--secure $\QFE$. (a) \cite{MM24} construct single-message, single-query adaptive $\QFE$ and they show that non-adaptive single-query single-message $\QFE$ implies non-adaptive single-query many-message $\QFE$ for $\QFE$ schemes that have a classical function key. (b) Implied by one of our results.}
    \label{tab:overviewsim} 
\end{table}

\paragraph{Impossibility in Plain Model for Classical Functionalities.}
There remains a notable question concerning the impossibility of \textsf{SIM}--security for functional encryption of \emph{classical} circuits in the presence of quantum capabilities. Although we do not provide a formal treatment of this setting, we give an informal overview explaining how our results imply an impossibility or can be readily adapted to this setting. The formal proofs can easily be derived from our work and the existing literature. 

For our first impossibility result, the adversary only makes use of classical messages and the identity functionality, which is a classical function. Therefore, the impossibility of many-message single-query adaptive \textsf{SIM}--secure $\QFE$ for classical functionalities but with quantum capabilities immediately follows. In the second impossibility result, we rely on the incompressibility of $\PRS$, an inherently quantum primitive. This could be replaced by a $\PRF$ and a similar argument regarding the incompressibility of classical information as in \cref{lem:incompress}. Therefore, the impossibility of succinct single-message, single-query non-adaptive $\textsf{SIM}$--secure $\FE$ for classical circuits in a quantum world also follows. The reason we opted to use $\PRS$ instead of $\PRF$ is to get an impossibility result from a weaker assumption. The third result assumes the existence of a $\PKE$ scheme with classical keys. {If we further assume a quantum-secure classical $\PKE$ scheme, an impossibility result for single-message many-query $\FE$ for classical functionalities follows from the $\PRF$--based argument, since such a $\PKE$ scheme implies a quantum-secure $\PRF$.} 

\paragraph{Constructions in the Bounded Storage Models.}

Given the difficulty of achieving simulation-security in both the plain classical and quantum models, we turn our attention to alternative, more restricted models. 

Specifically, we investigate $\textsf{SIM}$--security for $\textsf{FE}$ within both the Bounded Quantum Storage Model ($\textsf{BQSM}$) and the Bounded Classical Storage Model ($\textsf{BCSM}$). In these models, the adversary's quantum or classical memory is limited. 

Our $\textsf{FE}$ constructions are identical in both models; the security of these schemes relies fundamentally on a primitive termed \emph{disappearing grey-box obfuscation}. A standard grey-box obfuscation ($\textsf{GBO}$) provides a program encoding that permits evaluation of the encoded program while strictly limiting further information leakage; formally, it is required that any interaction with the obfuscation can be simulated by a computationally unbounded algorithm that is given only polynomial query access to the program. A \emph{disappearing} obfuscation adds a temporal constraint: {the encoding can no longer be used for evaluation \emph{after} it has been received}. Any model that supports the construction of this primitive, in conjunction with {one-way functions ($\textsf{OWF}$s), is sufficient} to realize our \textsf{FE} schemes. We highlight that while disappearing $\GB$ exists unconditionally in the $\textsf{BQSM}$, only candidate constructions exist in the $\textsf{BCSM}$~\cite{GZ21}.

Critically, the impossibility results for $\textsf{FE}$ in this work and previous works \cite{BSW11,AGV13} do not apply to the bounded storage models as functional keys and ciphertexts may not be storable for future use. This property is implicitly assumed in these impossibility results. When incorporating disappearing $\GB$ into either the ciphertext or functional keys, the adversary cannot store these transmissions, preventing the attacks against simulation security.




Informally, our contributions are as follows:
\begin{itemize}
    \item We first introduce new \textsf{FE} security notions specific to the bounded storage models. For instance, \emph{disappearing ciphertexts} requires that a ciphertext becomes unusable for further computation after it has been received, similar to notions introduced in \cite{GZ21,BS23}, but adapted to \textsf{FE}. Meanwhile, the notion of \emph{disappearing functional keys} requires that an adversary cannot utilize functional keys obtained prior to the challenge phase to decrypt or analyze the challenge ciphertext itself. These notions are formalized in the security experiments in \cref{sec:sec1,sec:sec2}.
    
    \item We build a $1\text{-}\textsf{M}\text{-}\textsf{AD}$ \textsf{SIM}--secure $\FE$ scheme assuming $\OWF$s and disappearing $\GB$. Specifically, our scheme satisfies simulation security against adversaries that are given many (adaptive) disappearing functional key queries and a single challenge ciphertext. 
    
    \item {We build a $\textsf{M}\text{-}1\text{-}\textsf{AD}$ \textsf{SIM}--secure $\FE$ scheme assuming disappearing $\GB$. Specifically, our scheme satisfies simulation security against adversaries that are given a single non-adaptive functional key query, many disappearing challenge ciphertexts, and the adaptive functional key of the identity function.}
\end{itemize}

Regarding the memory requirements of our scheme, they are inherited from the memory requirements of the disappearing $\GB$. In the \textsf{BQSM}, there exists an unconditional $\GB$ \cite{BS23} where the honest user requires no quantum memory, while the adversary requires memory proportional to the length of the obfuscation transmission. Therefore, this scheme can be made secure against adversaries with arbitrarily large quantum memory requirements by making the transmission long enough. Meanwhile, in the \textsf{BCSM}, the candidate constructions of $\GB$ \cite{GZ21} have a quadratic gap $\xx \ \text{vs.} \ O(\xx^2)$ between the memory requirements to run the scheme honestly and the memory required to break security. 

\subsection{Related Works}
We briefly discuss related works on \textsf{FE} in various models.

\paragraph{Plain Quantum Model.}
Several earlier works explored the feasibility of $\QFE$ schemes in different settings. Hiroka et al.~\cite{HKM+23} construct certified everlasting \textsf{FE} by enhancing $\FE$ for classical circuits with quantum ciphertexts. A receiver holding a quantum ciphertext can later produce a certificate proving that the ciphertext has been irreversibly destroyed, along with any embedded information such as the plaintext.
Beyond deletion, \cite{KN23} introduce $\FE$ with secure key leasing from any secret-key $\FE$. The work of Çakan and Goyal~\cite{CG24} improves upon this result by building $\FE$ with copy-protected secret keys that withstand unbounded collusions. This result does not contradict our impossibility result regarding many key-queries since they achieve \textsf{IND}--security and we only rule out that such a construction could achieve \textsf{SIM}--security. 


Mehta and Müller~\cite{MM24} initiated a broader formal investigation of $\QFE$. They introduce definitions for simulation and indistinguishability based notions of $\QFE$ and provide a construction for single-query, single-message adaptive \textsf{SIM}--secure $\QFE$. They show that, equivalently to the classical case, the simulation-based notion is stronger and implies the indistinguishability-based notion. Their construction is not succinct.
Our work builds on this emerging line of research, focusing on simulation-based security. 

Note that our positive results in the bounded storage models cannot be realized in the plain model because disappearing $\GB$ is impossible without memory restrictions; any obfuscation can be stored for future use. However, {traditional $\GB$ has been constructed from standard assumptions \cite{BCK17}} and shown to be equivalent to a strong form of \emph{indistinguishability obfuscation (\textsf{iO})}. Notably, \textsf{iO} is more studied compared to $\GB$, and there is a close relationship between \textsf{iO} and \textsf{IND}--secure \textsf{FE} \cite{BV18,GGH+16}. In fact, even \textsf{SIM}--secure \textsf{FE} was built from \textsf{iO} in a limited setting which evades the impossibility results \cite{DIJ13}. Specifically, the scheme in \cite{DIJ13} allows for a single non-adaptive functional key query, a single challenge ciphertext, and many adaptive functional key queries. We note that full-fledged \textsf{M}\text{-}\textsf{M}\text{-}\textsf{AD}\text{-}\textsf{SIM} \textsf{FE} has been achieved in the random oracle model \cite{DIJ13}, as the impossibility results do not apply to that model either.

\paragraph{Bounded Classical Storage Model.}

Maurer \cite{m92} was the first to consider restricting an adversary with respect to its memory (classical memory), introducing the \textsf{BCSM}. Following this, a series of works showed how to build interactive oblivious transfer, key agreement, symmetric key encryption, and signatures in this model \cite{GZ19,DQW21,R18,DQW22}. Most of these constructions offer information-theoretic security against adversaries with $O(\xx^2)$ memory where $\xx$ is the required memory to run the schemes honestly. In fact, for many primitives, the $\xx \ \text{vs.} \ \xx^2$ memory gap is optimal \cite{DQW22}. 

One drawback to the information-theoretic signature scheme constructed in the \textsf{BCSM} \cite{DQW22} is that the verification key is a long stream that can only be announced at the start. It is not difficult to show that if a computationally unbounded adversary can receive many verification keys in sequence, then it can eventually break security. This is unfortunate since it requires all users to be online at the time of key announcement. As such, it does not capture the accessibility of standard public key encryption or signature schemes. 

However, when considering polynomial-time adversaries, there exist schemes in the $\BCSM$ without this issue. Specifically, Guan, Wichs and Zhandry \cite{GWZ22} provided constructions for public key encryption and signatures with keys that can be distributed at any point, albeit at the cost of assuming public key encryption and one-way functions, respectively. Guan and Zhandry \cite{GZ21} also provided a scheme for upgrading a standard \textsf{AD-IND} secure \textsf{FE} scheme into one with disappearing ciphertexts, at the cost of assuming the existence of disappearing $\GB$, non-interactive zero knowledge proofs and \textsf{OWF}s. In comparison, our Construction \ref{con:fe} also allows key distribution at any point and {achieves \textsf{SIM}--security for disappearing functional keys at the cost of assuming} $\OWF$ and disappearing $\GB$.

We chose to use $\GB$ for our constructions instead of \textsf{iO} because it is not only unclear how to define a disappearing property for \textsf{iO}, but also because our proofs explicitly use the stronger simulation security provided by $\GB$. 

Furthermore, Goyal, Koppula, Rajasree and Verma \cite{goyal2025incomp} investigate the notion of incompressible functional encryption and {give three constructions of incompressible \textsf{FE}}  in the \textsf{BCSM}. Incompressibility is similar to disappearing security, but a major difference between our notion of disappearing ciphertexts and their work is that incompressible functional encryption is defined as an \textsf{IND}--security game, while we work in the \textsf{SIM}--security setting. For functional encryption in the plain model it is known that \textsf{SIM}--security implies \textsf{IND}--security and by the same proof idea this implication also holds true in the bounded storage model. While {our construction provides \textsf{SIM}--security}, they focus additionally on short ciphertexts and key length and achieve optimal trade-offs~\cite{bhushan24}. Both works leave open the question of constructing a scheme that fulfills a security definition which allows to reveal the master secret key. 

\paragraph{Bounded Quantum Storage Model.}

In 2005, Damg{\aa}rd, Fehr, Salvail, and Schaffner \cite{dfss05} considered what would happen if adversaries were instead restricted with respect to their quantum memories. Intuitively, the memory bound is used to force adversaries to measure quantum states received, leading to an inevitable loss of information that enables secure applications. In practice, the memory bound is applied by pausing and delaying further transmissions, justified by the difficulty of maintaining quantum states for an extended period. 

Quantum memory limitation turns out to be quite powerful in cryptography. Damg{\aa}rd, Fehr, Salvail, and Schaffner \cite{dfss05} found that non-interactive oblivious transfer and bit-commitment can be achieved information-theoretically in this model. Subsequently, Barhoush and Salvail \cite{BS23} attained one-time programs and program broadcasts in the \textsf{BQSM}, which implies $\GB$. They used these primitives to construct information-theoretic \textsf{CCA1} secure public key encryption, signatures, and token schemes. The public key encryption and signature schemes require the public/verification key to be announced only at the start. Our Construction \ref{con:fe} gives a public key encryption scheme with keys that can be distributed at any point. Furthermore, prior to this work, there did not exist any constructions for \textsf{FE} in the \textsf{BQSM}.

\subsection{Future Directions}
While this work establishes several impossibility results for simulation-secure $\textsf{QFE}$, several intriguing avenues remain for future exploration:
\begin{itemize}
    \item Our impossibility results target $\textsf{QFE}$ for general circuit classes. However, restricted classes may still permit simulation security. For instance, {O’Neill \cite{O10} showed that non-adaptive $\textsf{IND}$--security implies $\textsf{SIM}$--security for preimage sampleable $\FE$ schemes.} Future research should investigate which specific subclasses of quantum circuits bypass our impossibility barriers.
    \item Moreover, the robustness of some of our results currently hinges on specific cryptographic assumptions. Strengthening these findings by weakening or entirely removing these assumptions would provide a more fundamental understanding of the boundaries of $\QFE$.
    \item Finally, our results leave open the specific regime where the adversary receives a single ciphertext and a single non-adaptive functional key query (prior to the challenge), followed by multiple adaptive queries. Given that De Caro et al. \cite{DIJ13} demonstrated the feasibility of this setting in the classical model, a primary direction for future work is to determine if their construction (or a quantum-resistant variant) maintains security in the quantum realm.
    \item Our positive results in the \textsf{BCSM} assume the existence of disappearing $\GB$, which is not a typical assumption, with only candidate constructions given in \cite{GZ21}. So an interesting direction is to base a scheme on more standard assumptions. 
\end{itemize}

\subsection{Acknowledgments} We used ChatGPT 6 Astra to improve presentation and perform minor improvements. The ideas and proofs were mainly generated by the authors. The authors take responsibility for any mistakes or issues. Louis Salvail and Mohammed Barhoush also acknowledge the support of Natural Sciences and Engineering Research Council of Canada (NSERC) under the QUORUM project [ALLRP-578455-2022]
and NSERC Alliance Quantum Consortium grant for the QEYSSat mission under the  QUINT project [ALLRP 578460-22].

\section{Technical Overview}
\label{sec:tech}

We describe our results more concretely. {For our impossibility results, we work with \emph{secret-key} $\FE$ in the simulation-based security setting.} We allow the ciphertexts and functional keys to be quantum while the master secret-key is classical. 

Note that the impossibilities for secret-key $\QFE$ immediately imply the equivalent impossibilities for {public key $\QFE$ with classical master secret keys, including quantum public keys that can be generated from the master secret key.}

\subsection{Impossibility of \textsf{SIM}--Secure $\QFE$ with Many Challenge Ciphertexts}

Our first result demonstrates the impossibility of \textsf{SIM}--secure $\QFE$ in the setting where the adversary receives many ciphertext queries and a single adaptive functional key query. This is a quantum version of the classical impossibility result given in \cite{BSW11}. More formally, we show the following:
\begin{theorem}[informal]
    There does not exist a $\textsf{M}\text{-}1\text{-}\textsf{AD}$ \textsf{SIM}--secure $\QFE$ scheme. 
\end{theorem}
The proof proceeds via the following straightforward attack strategy. The adversary begins by requesting the encryption of $n$ {random bits $(m_i)_{i \in [n]}$}, where $n$ is some sufficiently large polynomial in the security parameter. After receiving the corresponding ciphertexts from the challenger, the adversary then requests the
functional key for the identity function $I$. {By applying decryption coherently and using the gentle measurement lemma to reuse the key, the adversary recovers all messages except with negligible probability.} Thus, the adversary only makes a \emph{single} adaptive functional-key query, and receives \emph{many} challenge ciphertexts.

Assuming the reader's familiarity with the $\textsf{SIM}$--security game (\cref{def:simsecurityad}), we now describe why this interaction is hard to simulate. Observe that since the
functional key is queried \emph{after} the challenge messages, the simulator must
first produce $n$ ciphertexts without any knowledge of the underlying messages.
Suppose the simulator outputs ciphertexts $(\ct_i)_{i \in [n]}$. When the adversary
subsequently requests the functional key for $I$, the simulator receives the
messages $(m_i)_{i \in [n]}$ and must construct a state $\mathsf{fk}_I$ such that
\[
\textsf{Dec}(\mathsf{fk}_I, \ct_i) = m_i \quad \text{for all } i \in [n].
\]

However, two fundamental issues arise:
\begin{enumerate}
  \item The size of $\mathsf{fk}_I$ is bounded by a fixed polynomial. Hence, as long as $n$ is chosen large enough, $\mathsf{fk}_I$ is of much smaller size.
  \item The ciphertexts $(\ct_i)_{i\in [n]}$ were
  generated independently of the messages $(m_i)_{i\in [n]}$.
\end{enumerate}

Together, these imply that $\mathsf{fk}_I$ cannot contain enough information to
describe a mapping from each ciphertext to its corresponding message.
Classically, this impossibility is straightforward. In the quantum setting,
however, one must contend with the fact that both ciphertexts and functional {keys may be quantum states, which in principle can encode much more information
than classical strings.}

To address this issue, we found that it is useful to reinterpret the simulator’s task as an \emph{information
compression problem}. Upon receiving the $n$ messages, the simulator would need
to compress them into a quantum state $\fk_I$ of size {strictly less than $n/2$ qubits}, in
such a way that this compressed state, when combined with independently chosen
ciphertexts, could still be used to {recover all $n$ messages except with negligible probability}.

We show that such a compression procedure is impossible, even in the quantum
setting, by using a related result from \cite{NS06}. This establishes the impossibility of $\textsf{M}\text{-}1\text{-}\textsf{AD}$ \textsf{SIM}--secure $\QFE$. 

Interestingly, we obtain an impossibility result for another variant of $\QFE$ using our incompressibility result along with \cite{AGV13}. Recall that Agrawal et al.~\cite{AGV13} showed that in the classical setting $1\text{-}\textsf{M}\text{-}\textsf{NA}$ \textsf{SIM}--secure $\FE$ is impossible assuming the existence of $\PRF$\footnote{Their result is actually unconditional, given that $\FE$ implies $\PRF$ in the classical setting. }. This result only applies to classical schemes as it relies on a classical incompressibility result. With our quantum incompressibility result, we can adapt their result to the quantum setting with minimal modifications obtaining impossibility of $1\text{-}\textsf{M}\text{-}\textsf{NA}$ \textsf{SIM}--secure $\QFE$, assuming the existence of $\PRF$s.

\subsection{Impossibility of \textsf{SIM}--Secure  $\QFE$ for Succinct Schemes}
In the second result, we show the impossibility of succinct $\textsf{SIM}$--secure $\QFE$ in the simplest setting where only one message and one non-adaptive function query are allowed. The high-level idea is inspired by the classical impossibility result of \cite{AGV13}: 
\begin{enumerate}
    \item Show that there exists a class of incompressible circuits.
    \item Show that a succinct $\QFE$ scheme for such a class of circuits implies a compression and decompression algorithm for this class of circuits.
\end{enumerate} 

We show that a $\PRS$ generator is an incompressible circuit, therefore, we obtain the following theorem.

\begin{theorem}[informal]
    Assuming the existence of $\PRS$ with sufficiently long output size, there does not exist a succinct $1\text{-}1\text{-}\textsf{NA}$ \textsf{SIM}--secure $\QFE$ scheme. 
\end{theorem}

Recall that the security of a $\PRS$ is defined with respect to Haar random states, which is the uniform distribution over pure quantum states of a given dimension. A $\PRS$ is secure if an adversary cannot distinguish whether he is given a polynomial number of copies of a Haar random state or a $\PRS$. We use this relationship to work with Haar random states instead of $\PRS$ in the first step of the proof. In particular, we show that there do not exist CPTP maps $\Comp,\Decomp$ such that a Haar random state can undergo the compression $\Comp$ and then be recovered by the decompression $\Decomp$  with high fidelity. 

\begin{lemma}[informal]
    Let $m<n$. There do not exist CPTP channels $\Comp: D(\Hcal^{n})\rightarrow D(\Hcal^{m}),$ $ \Decomp: D(\Hcal^{m}) \rightarrow D(\Hcal^n)$ such that 

    $$ \Pr \left[ F(|\psi\rangle \langle \psi|, \Decomp \circ \Comp (|\psi \rangle \langle \psi |)) \geq 1 - \negl \right] \geq \frac{1}{\poly}$$
    where the probability is taken over choosing a Haar random state.
\end{lemma}

This Lemma can be seen as a quantum analog of the classical fact that uniformly random strings cannot be compressed. Note that in this Theorem the compression does not depend on the number of qubits by which we compress the state. Even a compression of just a single qubit makes it nearly impossible to recover the state. This is in contrast to the classical setting where removing a single bit from a random string leaves an adversary with a high chance of recovering the string by guessing the missing bit. 

The proof of the incompressibility of Haar random states proceeds as follows. First, we compute the expected fidelity $\Favg$ of the channel $\Phi = \Decomp \circ \Comp$ applied to a Haar random state by relating it to the entanglement fidelity $\Fe$. This allows us to analyze the effect of the map $\Phi$ on one half of the maximally entangled state instead of a Haar random state. Formally the expected fidelity and the entanglement fidelity are related as $\Favg = \frac{2^n \Fe + 1}{2^n + 1}$. A key step in the proof is the decomposition of the CPTP maps $\Comp$ and $\Decomp$ in their Kraus representations $\{A_i\}$ and $\{B_i\}$. Due to the required dimension of the maps we know that $A_i$ outputs states in a $2^m$-dimensional space and $B_i$, acting on this space, also has outputs contained in a $2^m$-dimensional subspace. We conclude that {each Kraus operator $B_jA_i$ of the combined map $\Phi$ has rank at most $2^m$} which we can use to show that the entanglement fidelity is $\Fe \leq \frac{2^m}{2^n}$. Therefore the average fidelity is $\Favg(\Phi) \leq \frac{1}{2^{n-m}} + O(\frac{1}{2^n})$. Finally, we can apply Levy's Lemma, which states that if a function $f$ does not change too abruptly (a Lipschitz function), then its value on a Haar random state $f(|\psi \rangle)$ does not deviate much from the expected value of the function. The probability of a deviation by $\epsilon$ is exponentially small in $\epsilon^2$ and the dimension $2^n$. For a successful compression and decompression algorithm we expect the fidelity to be close to 1 but we have just computed that the expected fidelity $\Favg(\Phi)$ is much lower, i.e. even for only a single qubit of compression $m = n-1$ we get $\Favg(\Phi) \leq \frac{1}{2} + O(\frac{1}{2^n})$. Due to Levy's Lemma we see that the fidelity on a random state cannot deviate much and we conclude that $\Comp, \Decomp$ with the required properties cannot exist. 

To finish the first step in our proof, we need to define a class of incompressible circuits. From the previous Lemma, it easily follows that $\PRS$ are incompressible. {If there were efficient compression and decompression algorithms for pseudo-random states}, then we could distinguish $\PRS$ from Haar random states by applying the algorithms and checking if we got the correct state back. Note that to check if we get the correct state then we need at least one additional copy of the $\PRS$ or Haar random state to apply a Swap test. Then the adversary can guess not with certainty but with high enough probability if he was given a $\PRS$ or a truly random state. 

What remains is the second step in the proof: Showing that a $\QFE$ scheme implies a compression and decompression algorithm for $\PRS$. Assuming the reader's familiarity with the $\textsf{SIM}$--security game (\cref{def:simsecurityna}) we give a brief description of the procedures $\Comp$ and $\Decomp$ based on the simulator of the security game. In the non-adaptive security game the adversary first requests a function secret key for the circuit $\PRS.\Gen(\cdot): \bin^\lambda \rightarrow D(\Hcal^{s})$ which is the generator for a pseudorandom state. Then the adversary samples a random key $ k \in \bin^\lambda$ and requests an encryption of this key. The simulator is now given the output of the circuit applied to the key $\PRS.\Gen(k) = |\psi_k\rangle$ and needs to produce a ciphertext of size $t$. Now, if $s>t$ then this is exactly the task of compressing $|\psi_k\rangle$ with knowledge of the state and the circuit but not the key $k$, which we showed to be impossible. Decompression of the state is the same as applying the decryption algorithm of the $\QFE$ scheme given the ciphertext and the function secret key.

\subsection{Impossibility of \textsf{SIM}--Secure  $\QFE$ with Many Functional Keys}

Our final result demonstrates the impossibility of \textsf{SIM}--secure $\QFE$ in the setting where the adversary receives one ciphertext query and many non-adaptive functional key queries, assuming $\PKE$. More formally, we show the following:

\begin{theorem}[informal]
    {Assuming the existence of $\PKE$ with classical keys,} there does not exist a $1\text{-}\textsf{M}\text{-}\textsf{NA}$ \textsf{SIM}--secure $\QFE$ scheme. 
\end{theorem}

\paragraph{High-level intuition.} In our reduction, we generate functional keys that compute \emph{public key encryptions} for a \textsf{MK-CPA} secure $\PKE$ scheme under many independent public keys. The correctness of a $\QFE$ scheme ensures that a ciphertext of a message $M$ (under the $\QFE$ scheme) can be used along with the functional keys to obtain encryptions of $M$ instead under the $\PKE$ scheme with the various public keys. We move from a case where all the ciphertexts encrypt a \emph{single} message $M$ (under different public keys) to {a hybrid where each ciphertext encrypts an \emph{independent random} message}. Indistinguishability of the $\PKE$ scheme then forces the $\QFE$ simulator to produce one short $\QFE$ ciphertext that, when evaluated via the functional keys, {allows recovery of the independent messages with overwhelming probability}. This turns the simulator into an \emph{information compressor}: from a single object it must enable recovery of $n$ independent encrypted messages. This violates basic information-theoretic limits, yielding a contradiction.

\paragraph{From \textsf{IND-CPA} to \textsf{IND-MK-CPA}.} We first formalize a convenient strengthening of $\PKE$ security, which we call \textsf{IND-MK-CPA}: roughly, security must hold even when the adversary simultaneously receives challenge ciphertexts under \emph{many} independently generated public keys. A standard hybrid argument shows that ordinary \textsf{IND-CPA} implies \textsf{IND-MK-CPA}.

\paragraph{Attack against \textsf{IND-MK-CPA}.}
Assume for contradiction that there exists a $1\text{-}\textsf{M}\text{-}\textsf{NA}$ \textsf{SIM}--secure $\QFE$ scheme given by the algorithms $(\Setup, $ $\Keygen,$ $ \Enc, \Dec)$ (see \cref{def:qfe} of $\QFE$).  \textsf{SIM}--security of the scheme guarantees the existence of a simulator $\textsf{Sim}$. 

We describe an attack $\adv$ against \textsf{MK-CPA} security of a $\PKE$ scheme given by the algorithms $(\textsf{PKE.Gen},\textsf{PKE.Enc},\textsf{PKE.Dec})$. The adversary $\adv$ in the \textsf{IND-MK-CPA} security experiment picks a random message $M$ and $n$ independent random messages $m_{i}$, and outputs $(M,m_{i})_{i\in [n]}$, where $n$ is a large enough polynomial in the security parameter. The challenger samples a bit $b\gets \{0,1\}$ and returns $(\ct_i)_{i\in[n]}$ where $\ct_i=\PKEnc(pk_i,M)$ if $b=0$ and $\ct_i=\PKEnc(pk_i,m_{i})$ if $b=1$. Next, $\adv$ internally samples a key pair {$(\tilde{sk},\tilde{pk}) \leftarrow \textsf{PKE.Gen}(1^\lambda)$} and encrypts $\tilde{\ct}\gets \PKEnc(\tilde{pk},M)$. $\adv$ then picks $j\gets [n]$ and replaces the ciphertext $\ct_j$ with $\tilde{\ct}$ and the public key $pk_j$ with $\tilde{pk}$.{ For the analysis, we also set $m_j=M$ and $sk_j=\tilde{sk}$.} Intuitively, this public key and ciphertext act as a ``trap'' that is used by the adversary to check whether the simulator is actually generating a valid encryption of $M$. 

$\adv$ now uses the $\QFE$ scheme to attack the $\PKE$ scheme as follows. $\adv$ samples a master key $\mk\gets \Setup(1^\lambda)$ and, for each $i \in [n]$,
$\adv$ generates a functional key  $\mathsf{fk}_{C_i}$ corresponding to the circuit
\[
C_i := \PKEnc(\mathsf{pk}_i, \cdot),
\]

Now $\adv$ runs the simulator $\Sim$ {on the function list and evaluations $(C_i,\ct_i)_{i\in[n]}$}. $\Sim$ outputs a $\QFE$ ciphertext $\widetilde{\mathsf{CT}}$.\footnote{Crucially, the size of a real $\QFE$ ciphertext is fixed by the scheme and does \emph{not} depend on $n$.}

{Using the functional key $\mathsf{fk}_{C_j}$, the adversary evaluates
$\widetilde{\mathsf{CT}}$ to obtain $\ct'_j$,} and then, using his secret key, verifies that
\[
{\textsf{PKE.Dec}(\tilde{{sk}}, \ct'_j) = M .}
\]

Notice that in the case $b=0$, the ciphertexts $(\ct_i)_{i\in [n]}$ correspond to the evaluations of the same message $M$ on the circuits $(C_i)_{i\in [n]}$. This means that the simulator, given $(C_i,\ct_i)_{i\in[n]}$, should be able to produce a single ciphertext that is indistinguishable from an encryption of $M$ in the eyes of the adversary. In other words, {the verification step above should pass by the correctness of $\QFE$, when $b=0$.}

{Indistinguishability ensures that this simulation is also valid in a hybrid world $(b=1)$
where the ciphertexts $(\ct_i)_{i\in [n]}$ encrypt $n$ \emph{independent} messages $(m_{i})_{i\in [n]}$. In that case,
the simulator must still produce a compressed ciphertext $\widetilde{\mathsf{CT}}$
that, together with the functional keys, expands back into $n$ ciphertexts $(\ct'_i)_{i\in [n]}$ such that $\mathsf{Dec}(\tilde{{sk}}, \ct'_j) = M$. Otherwise, we could use the simulator to break \textsf{IND-MK-CPA} and distinguish $b$ from random.}

{Since $j$ is uniformly random and hidden from the simulator, this implies that} 
\[
{\Pr_{i\gets[n]}\left[\textsf{PKE.Dec}\bigl(sk_i,\Dec(\mathsf{fk}_{C_i},\widetilde{\mathsf{CT}})\bigr)=m_i\right]\geq 1-\negl[\lambda].} 
\]
{By the gentle measurement lemma, these computations can be performed and uncomputed sequentially to recover all messages, except with negligible probability.} 

This requirement effectively makes the simulator an \emph{information compression algorithm}.
Specifically, it must take the encryption of $n$ messages and compress it into a state
$\widetilde{\mathsf{CT}}$ {of size strictly smaller than half the total message length,} while ensuring that the
combination of $\widetilde{\mathsf{CT}}$, the functional keys, and the secret keys suffice to {recover all $n$ messages except with negligible probability}.

Crucially, the secret keys and functional keys are generated
independently of the messages. They form input-independent shared advice
in the compression experiment. Only the simulated ciphertext
counts toward the communication bound in
\cref{lem:incompress}. Since such compression is impossible, we reach
a contradiction. Hence, no $1\text{-}\textsf{M}\text{-}\textsf{NA}$ \textsf{SIM}--secure $\QFE$ scheme can exist under the assumption of
$\PKE$ with classical keys.

An interesting direct corollary to our result is that $1\text{-}\textsf{M}\text{-}\textsf{NA}$ \textsf{SIM}--secure public key $\QFE$ (with classical public keys) is \emph{unconditionally} impossible, given that this primitive implies $\PKE$.  
\begin{corollary}
There does not exist $1\text{-}\textsf{M}\text{-}\textsf{NA}$ \textsf{SIM}--secure public key $\QFE$ with classical public keys. 
\end{corollary}

We provide a brief simplified overview of our two constructions in the bounded storage models.

%% file: tech_overview_last_reviewed.tex
\providecommand{\QFEreplace}[2]{}
\renewcommand{\QFEreplace}[2]{%
  \begingroup
  \let\QFEoriginalcref\cref
  \let\QFEoriginalcite\cite
  \renewcommand{\cref}[1]{\mbox{\QFEoriginalcref{##1}}}%
  \renewcommand{\cite}[1]{\mbox{\QFEoriginalcite{##1}}}%
  \sout{#1}\endgroup\ {\color{red}#2}}
\providecommand{\QFEinsert}[1]{{\color{red}\relax#1}}
\providecommand{\QFEmathreplace}[2]{\text{\sout{\ensuremath{#1}}}\;{\color{red}#2}}

\subsection{Functional Encryption with Disappearing Functional Keys}

We now describe our construction of $\FE$ satisfying {$\textsf{DF-SIM}$--security} with disappearing functional keys{, for a single challenge ciphertext and polynomially many adaptive functional key queries}, assuming the existence of: 
\begin{itemize}
\item {A disappearing $\GB$ denoted $\mathcal{O}$.}
\item {A one-time information-theoretically secure symmetric authenticated encryption scheme} $(\textsf{Enc},\textsf{Dec})$, which exists unconditionally.
    \item A puncturable pseudorandom function \textsf{PPRF} $(\textsf{KeyGen},\textsf{Puncture},F)$ (see Definition \ref{def:puncturable}).  
    \item {A pseudorandom generator $\textsf{PRG}:\{0,1\}^{\lambda}\rightarrow\{0,1\}^{2\lambda}$.}
\end{itemize} 

Note that the last two assumptions can both be based on \textsf{OWF}s{, which must be post-quantum in the \textsf{BQSM}}.{ In both constructions below, all streams obey the prescribed memory bounds.}

The scheme proceeds as follows. We sample a key $\nk\leftarrow \textsf{KeyGen}(1^\lambda)$ for the \textsf{PPRF} and set this as the master key. Then, when a user requests a public key, we send $\mathcal{O}(P_{\nk})$, where $P_{\nk}$ is defined as the program that on input $r$, {outputs $(t,F(\nk, t))$ where} $t=\textsf{PRG}(r)$. {The value $F(\nk,t)$ is then used} as a secret one-time key to encrypt a message $m$ using the symmetric encryption scheme i.e. $\ct\leftarrow \textsf{Enc}(F(\nk,t),m)$. The value $t$ is included in the ciphertext to allow decryption and evaluation. 

Specifically, the functional key of a circuit $C$ is set as the obfuscation $\mathcal{O}(F_{\nk,C})$ of the following function:
\begin{align*}
    F_{\nk,C}(t,\ct)\coloneqq 
    \begin{split}
\begin{cases} 
 C(m) & m\leftarrow \textsf{Dec}(F(\nk,t),\ct){,\ m\neq\perp}\\
\perp & \text{otherwise}.\\ 
\end{cases}\end{split}  
\end{align*}

Clearly, an honest user can {evaluate the obfuscation during its transmission on an already available ciphertext} $(t,\ct)$ to obtain the evaluation of the underlying message.

The main barrier to proving security is the tension between $\GB$ security and the security provided by the rest of the assumptions. Because $\GB$ security only guarantees that {an adversary's view of the program streams can be simulated} by a computationally unbounded simulator---{with only polynomially many classical oracle queries during each stream}---we face a significant technical hurdle. {The computational security of the \textsf{PPRF} and \textsf{PRG} cannot be invoked against such a simulator, since it may use brute force to attack these primitives.} 

Previous techniques used to deal with this issue involved strong cryptographic assumptions such as \emph{non-interactive zero knowledge proofs} for \textsf{IND}--secure \textsf{FE} with disappearing ciphertexts or \emph{lossy functions} for public key encryption with disappearing ciphertext \cite{GZ21}. {For disappearing functional keys, we provide an alternative approach whose only additional computational assumption, beyond $\GB$, is \textsf{OWF}s.}

The proof relies on a careful and granular hybrid argument that slowly shifts the security basis from computational hardness to information-theoretic secrecy. {By iteratively applying the security properties of the underlying primitives, we effectively ``delete'' all sensitive information regarding the master key and the encrypted message within the functional keys and ciphertexts.} At the end of these transitions, {the only remaining data is the evaluation of the challenge ciphertext on functional keys queried after the challenge itself}.

{Once this information-theoretic state is reached, the computationally unbounded simulator can be safely invoked to handle the obfuscations. Since the sensitive message components have already been removed from the hybrid, the simulator's computational power is neutralized.} Thus, we demonstrate that an adversary learns nothing from a challenge ciphertext {except for evaluations on functional keys received after the challenge ciphertext}. {Furthermore, because the functional keys use disappearing $\GB$s, keys whose transmissions end before the challenge ciphertext is issued provide no additional information about the challenge message, thereby satisfying disappearing functional key security.}

\subsection{Functional Encryption with Disappearing Ciphertexts}

We build {a $\textsf{DC-SIM}$--secure symmetric-key $\FE$ scheme} with disappearing ciphertexts, {assuming only the existence of disappearing $\GB$ denoted $\mathcal{O}$}. Specifically, our scheme satisfies simulation security against adversaries that are given {at most a single non-adaptive functional key query and polynomially many disappearing challenge ciphertexts, followed by disclosure of the identity functional key}. {Given that disappearing $\GB$ exists unconditionally in the \textsf{BQSM}, our scheme is information-theoretically secure in this setting.}

Notably, disappearing ciphertext security is modeled by revealing the functional key for the identity function at the conclusion of the experiment. {This provides decryption capability for future ciphertext streams, while requiring that the disclosure reveal no additional information about the previously transmitted challenge messages.} {The experiment therefore allows one non-adaptive functional key query followed by this identity-key disclosure.} {The framework does not allow further functional key queries during the challenge ciphertext transmissions.} {Thus, the construction supports keys queried strictly before or strictly after the reception phase, representing a slightly constrained version of fully adaptive security.}

Our construction proceeds as follows. The master key is essentially {the key for a pairwise independent authentication function} i.e. {two independent uniform elements $\mk\coloneqq (a,b)$ of $\mathbb{F}_{2^\ell}$, where $\ell=\poly[\lambda]$ and $\ell\geq\lambda$}. The functional key for a function $F$ is simply given by $\fk_F\coloneqq (F,a\cdot F+b)$, {where $F$ denotes a fixed injective $\ell$-bit encoding of its circuit, identified with an element of $\mathbb{F}_{2^\ell}$. All additions and multiplications are over this field, and malformed encodings are rejected.} Meanwhile, the ciphertext of a message $m$ is an obfuscation $\mathcal{O}(E_{\mk,m})$ of the following function:
\begin{align*}
    E_{\mk,m}(F,y)\coloneqq 
    \begin{split}
\begin{cases} 
 F(m) & y=a\cdot F+b\\
\perp & \text{otherwise}.\\ 
\end{cases}\end{split}  
\end{align*}

The adversary in this setting is permitted only a single functional key prior to receiving the challenge ciphertext. Assume the adversary receives a functional key $\fk_{F^*}$ for some function $F^*$, followed by the challenge ciphertext $\mathcal{O}(E_{\mk,m^*})$ for some message $m^*$.

{By $\GB$ security, the adversary's view can be simulated by a computationally unbounded algorithm with only polynomially many classical queries to each ciphertext program, made during its transmission.} {Except when a valid new tag is guessed, the only message-dependent oracle response available to the simulator from $E_{\mk,m^*}$ is $F^*(m^*)$.} {Given only $\fk_{F^*}$, the tag for any $F\neq F^*$ is uniformly random in $\mathbb{F}_{2^\ell}$. Thus, even a computationally unbounded simulator making a total of $q=\poly[\lambda]$ classical queries across the ciphertext streams guesses any new valid tag with probability at most $q/2^\ell$.} {Furthermore, given that the ciphertext itself is a disappearing $\GB$, it cannot be used after it is received, and the adversary cannot learn anything from functional keys that it receives later.} We use this intuition to show that an adversary can only learn{the evaluation $F^*(m^*)$ from the ciphertext and nothing else}, which demonstrates security.{ } 

{In the \textsf{BQSM}, we can upgrade our scheme to the public-key setting using the same techniques in \cite{BS23} for upgrading symmetric encryption the public-key setting. Similarly, in the $\BCSM$, we can use the techniques of \cite{DQW22} to upgrade security to the public-key setting. Note that both these approaches necessitate using public-keys that can only be announced at the start.} 

%% file: preliminaries_reviewed.tex
\providecommand{\QFEreplace}[2]{}
\renewcommand{\QFEreplace}[2]{%
  \begingroup
  \let\QFEoriginalcref\cref
  \let\QFEoriginalcite\cite
  \renewcommand{\cref}[1]{\mbox{\QFEoriginalcref{##1}}}%
  \renewcommand{\cite}[1]{\mbox{\QFEoriginalcite{##1}}}%
  \sout{#1}\endgroup\ {\color{red}#2}}
\providecommand{\QFEinsert}[1]{{\color{red}\relax#1}}
\providecommand{\QFEmathreplace}[2]{\text{\sout{\ensuremath{#1}}}\;{\color{red}#2}}

\providecommand{\QFEstackreplace}[2]{\begin{gathered}\text{\sout{\ensuremath{#1}}}\\{\color{red}#2}\end{gathered}}

\section{Preliminaries}

\subsection{Notation}

We assume that the reader has preliminary knowledge of standard quantum notations and refer the reader to \cite{NC00} for a detailed exposition. 

Let $\Hcal^d = (\mathbb{C}^2)^{\otimes d}$ denote the Hilbert space of $d$ qubits. Let $D(\Hcal)$ denote the set of density operators on $\Hcal$. A quantum transformation is a completely positive trace preserving (CPTP) map which maps a density operator to a density operator. The trace distance for two density matrices is defined as $\T(\rho,\sigma) :=
\frac{1}{2}\|\rho-\sigma \|_1=
\frac{1}{2}\Tr [\sqrt{(\rho - \sigma)^{\dag}(\rho - \sigma)}]$, where $\Tr$ denotes the trace operator. 
Throughout this paper we use the following notation for quantum states: $\sk, \ct, \msg$, especially when a definition allows for either classical or quantum keys, messages, etc. For specifically classical elements we use $sk,m,ct$. For explicitly quantum states we also use $\rho, \sigma$ or the bra-ket notation $|\phi\rangle$. 

The notation $x\leftarrow X$ means $x$ is chosen from the values in $X$ according to the distribution $X$. If $X$ is a set, then $x$ is chosen uniformly at random from the set. We say that two functions $F,F'$ with the same domain are functionally equivalent, denoted as $F\equiv F'$, if they map every input in their domain to the same output. 

We say that an algorithm $\adv$ is \emph{PPT} if it is a probabilistic polynomial-time classical algorithm and we say it is \emph{QPT} if it is a quantum polynomial-time algorithm. We let $\adv^{qP}$ mean $\adv$ is given $q$ classical queries to an oracle of $P$, and $\adv^P$ mean $\adv$ is given {a polynomial number of classical oracle queries unless} a restriction on the number of queries is otherwise specified. Let $[n]\coloneqq {\{1,\ldots,n\}}$ and let $\negl[x]$ denote any function that is asymptotically smaller than the inverse of any polynomial.

We write $\langle q,q'\rangle$ to denote a transmission where $q$ and $q'$ are sent in sequence; first $q$ and then $q'$. This is relevant in bounded storage models, where the timing of the transmission matters.

\begin{lemma} (Gentle Measurement \cite{W99})
\label{lem:genlemeas}
    Let $\rho$ be a quantum state and let $\{\Pi, \id - \Pi\}$ be a projective measurement  such that $\Tr(\Pi \rho) \geq 1 - \delta${ for $0\leq\delta<1$}. Then the post measurement state post-selected on obtaining the first outcome is $\rho' = \frac{\Pi \rho\Pi}{\Tr(\Pi \rho)}$. It holds that $\T(\rho', \rho) \leq 2\sqrt{\delta}$.
\end{lemma}

\begin{fact}[Trace distance vs.\ Euclidean distance for pure states]
\label{fact:l1l2}
Let $|\psi\rangle,|\phi\rangle$ be unit vectors in $\mathbb{C}^d$. Then
\[
\|\ket{\phi}\bra{\phi} - \ket{\psi}\bra{\psi}\|_1 
\le  2 \|\ket{\psi} - \ket{\phi}\|_2.
\]
\end{fact}

\begin{proof}
The trace distance between two pure states admits the closed form
\[
\tfrac12\|\ket{\phi}\bra{\phi}-\ket{\psi}\bra{\psi}\|_1 
= \sqrt{\,1 - |\langle\psi|\phi\rangle|^2\,}.
\]
On the other hand, the Euclidean distance satisfies
\[
\|\ket{\psi}-\ket{\phi}\|_2^2 
= 2\big(1 - \mathrm{Re}\,\langle\psi|\phi\rangle\big)
\;\ge\; 2\big(1 - |\langle\psi|\phi\rangle|\big).
\]
For $c = |\langle\psi|\phi\rangle| \in [0,1]$ we have 
$1-c^2 = (1-c)(1+c) \le 2(1-c)$. 
Combining these relations yields
\[
\sqrt{\,1-|\langle\psi|\phi\rangle|^2\,} 
\;\le\; \| |\psi \rangle - |\phi\rangle \|_2,
\]
and multiplying both sides by~2 gives the claim. \qed
\end{proof}

\subsection{Pseudorandom States}

\begin{definition}[Pseudorandom States]
An ensemble of quantum states 
$\{ \ket{\psi_k} \}_{k \in \{0,1\}^\lambda}$, indexed by a classical key 
$k \in \bin^{\lambda}$ of length $\lambda \in \N$, is called a family of \emph{pseudorandom  
states} (PRS) if the following conditions hold:

\begin{enumerate}
    \item \textbf{Efficient generation.} There exists a quantum 
    polynomial-time (QPT) algorithm $\Gen$ that, given 
    $k \in \{0,1\}^\lambda$, outputs $\ket{\psi_k}$ on 
    $n = \mathrm{poly}(\lambda)$ qubits.
    $$ \Gen(1^{\lambda}, k) = \ket{\psi_k}$$

    \item \textbf{Pseudorandomness.} For every QPT distinguisher 
    $\mathcal{A}$, every polynomially bounded 
    $m = \mathrm{poly}(\lambda)$, and all sufficiently large $\lambda$, 
    \[
        \left| \Pr\!\left[ \mathcal{A}\!\left( |\psi_k\rangle^{\otimes m} \right) = 1 \right]
        - \Pr\!\left[ \mathcal{A}\!\left( |\varphi\rangle^{\otimes m} \right) = 1 \right] \right| 
        \leq \negl[\lambda] \enspace,
    \]
    where $k \leftarrow \{0,1\}^\lambda$ is chosen uniformly, $\Gen(1^\lambda, k) = |\psi_k\rangle$, 
    $\ket{\varphi}$ is sampled from the Haar measure on $n$-qubit pure states.
\end{enumerate}
\end{definition}

\subsection{Quantum Functional Encryption }

We follow \cite{STOC:BY22} in defining a quantum circuit that allows for a classical description. {Each such quantum circuit defines a CPTP map.} 

\begin{definition}(Classical Description of Quantum Circuits)
\label{def:circuit}
A quantum circuit is a tuple $(\mathcal{P}, \mathcal{G})$ where $\mathcal{P}$ is the topology of the circuit and $\mathcal{G}$ is a set of unitaries.{ The gates are taken from a fixed finite universal gate set of constant arity.} 
The topology of a quantum circuit is a tuple $(\mathcal{B}, \mathcal{I}, \mathcal{O}, \mathcal{W}, \texttt{inwire}, \texttt{outwire}, \mathcal{Z}, \mathcal{T})$.
\begin{enumerate}[align=left, leftmargin=2.8em]
    \item $\mathcal{I}$ is an ordered set of input terminals.
    \item $\mathcal{Z}$ is a subset of $\mathcal{I}$ which indicates ancilla qubits that are to be initialized to the state $|0\rangle$.
    \item $\mathcal{O}$ is an ordered set of output terminals.
    \item {$\mathcal{T}\subseteq\mathcal{O}$ is the set of output terminals to be traced out.} 
    \item $\mathcal{W}$ is the set of wires. 
    \item $\mathcal{B}$ are placeholder gates. For every $g \in \mathcal{B}$ $\texttt{inwire(g)}$ describes an ordering of input wires $w \in \mathcal{W}$ and $\texttt{outwire(g)}$ describes an ordering of output wires $w \in \mathcal{W}$.  For every $g \in \mathcal{B}$ the number of input and output wires is equal.
    \item The disjoint sets $\mathcal{I}, \mathcal{O} ,\mathcal{B}$ form the nodes of the circuit. Together with the set $\mathcal{W}$ as edges they define a directed acyclic graph. 
\end{enumerate}

The gate set $\mathcal{G}$ defines a unitary of the appropriate size for every node in $\mathcal{B}$. The evaluation of a circuit $C= (\mathcal{P}, \mathcal{G})$ {on state $\rho$ of size $|\mathcal{I}|$ is defined as $C(\rho, |0\rangle^{\otimes |\mathcal{Z}|}) = \sigma$ where $\sigma$ resulted from applying the gates in $\mathcal{G}$ according to the topology and tracing out the qubits specified by $\mathcal{T}$.} The size of a quantum circuit is the number of wires in $\mathcal{W}$. 
\end{definition}

Throughout this paper, we expect the $\QFE$ schemes to handle the above form of quantum circuits. {Polynomial-size circuits of this form capture efficiently implementable CPTP maps, up to negligible approximation error, including maps with different numbers of input and output registers.} The definition does not cover quantum circuits that contain a quantum state as part of their description.

Below we give definitions for secret-key Quantum Functional Encryption. $\QFE$ was first defined in \cite{MM24} in the public-key setting. {We prove our impossibility results for secret-key $\QFE$ with classical master secret keys, allowing quantum functional keys and ciphertexts. They therefore also apply to public-key $\QFE$ in this model.}

\begin{definition}[Syntax of a QFE scheme]
\label{def:qfe}
    Let $\lambda \in \mathbb{N}$ be the security parameter and let $\mathcal{C}_\lambda$ be the class of polynomially sized quantum circuits with input space $\mathcal{X}_\lambda$ and output space $\mathcal{Y}_\lambda$. A QFE scheme consists of the QPT algorithms $(\Setup, $ $\Keygen,$ $ \Enc, \Dec)$ which are defined as follows, 
\begin{itemize}[align=left]
    \item[$\Setup(1^\lambda) \rightarrow \mk$] Given the security parameter $\lambda$ {output a classical master secret key $\mk$.} 
    \item[$\Keygen(\mk, C \in \mathcal{C}_\lambda) \rightarrow \fk_C$] Given the master secret key and a quantum circuit $C$  output a function secret key $\fk_C$.
    \item[$\Enc(\mk, \msg \in \mathcal{X}_\lambda) \rightarrow \ct$] Given the master secret key $\mk$ and a message $\msg$ output a ciphertext $\ct$.
    \item[$\Dec(\fk_C, \ct) \rightarrow C(\msg) \in \mathcal{Y}_\lambda$] Given a function secret key $\fk_C$ and ciphertext $\ct$ which is an encryption of $\msg$ output the value $C(\msg)$.
\end{itemize}
\end{definition}

The following definition captures that a $\QFE$ scheme is called succinct if the size of the ciphertexts does not depend on the size of the circuits for which functional keys can be created. The running time of $\Enc$ may depend on the circuit size. 
\begin{definition}[Succinctness]
    Let $\QFE = (\Setup, \Keygen, \Enc, \Dec)$ be a functional encryption scheme for a circuit class $\mathcal{C}_\lambda$ for which the circuit size is bounded by $\ell_C$ and message space $\mathcal{M}$ which contains messages of size $\ell_M$ and which is instantiated with a security parameter $\lambda \in \N$. Then the $\QFE$ scheme is called succinct if the size of ciphertexts generated by $\Enc$ {are bounded by $\poly[\lambda, \ell_M]$ independent of $\ell_C$.}
\end{definition}

\begin{definition}[Correctness of a QFE scheme]
\label{def:qfecorrect}
{For every joint quantum state $(\msg,\mathsf{z})$, where $\mathsf{z}$ is an untouched reference register, and every circuit $C\in\mathcal{C}_\lambda$, it holds that} 
    \begin{align*}
        &\Pr\left[\T((C(\msg),\mathsf{z}), (\Dec(\fk_C,\ct), \mathsf{z})) \leq \negl[\lambda]\right] \geq 1 - \negl[\lambda]
    \end{align*}
    where $\fk_C \leftarrow \Keygen(\mk,C) ,\ct \leftarrow \Enc(\mk, \msg)$ and $\mk\leftarrow\Setup(1^\lambda)$.
\end{definition} 

In \cref{def:simsecurityna} and \cref{def:simsecurityad} we give the definitions for non-adaptive and adaptive simulation security of $\QFE$. {We allow for} \textit{quantum} function keys and ciphertexts, even if the currently known construction \cite{MM24} only needs classical keys.

\begin{definition}[Non-Adaptive Sim-Security for secret-key QFE]
\label{def:simsecurityna}
Let $\lambda$ be the security parameter and let $\mathcal{C}_\lambda$ be the class of polynomially sized quantum circuits. Let $(\Setup, $ $\Keygen,$ $ \Enc, \Dec)$ be a QFE scheme, let $\Adv = (\Adv_1, \Adv_2)$ be a QPT adversary and let $\Sim$ be a QPT simulator.{ The message length is supplied to $\Sim$ even when no key is queried.} 

   \begin{table}[H]
        \centering
        \begin{tabular}{p{6cm}|p{6cm}}
            $\Exp_{\Adv,\textsf{NA}}^{\textnormal{Real}}(1^\lambda)$ & $\Exp_{\Adv,\textsf{NA}}^{\textnormal{Ideal}}(1^\lambda)$ \\
            $\mk \leftarrow \Setup(1^\lambda)$&$\mk \leftarrow \Setup(1^\lambda)$ \\
            $ (\msg, \st) \leftarrow \Adv_1^{\Keygen(\mk, \cdot)}(1^\lambda)$ & $ (\msg,\st) \leftarrow \Adv_1^{\Keygen(\mk,\cdot)}(1^\lambda)$  \\
            ${\ct} \leftarrow \Enc(\mk, \msg)$& ${\ct} \leftarrow {\Sim(1^\lambda,\mk,\mathcal{V},1^{|\msg|})}$ \\
            & \quad where ${\mathcal{V}=(C,C(\msg))}$ if $\Adv$\\ & \quad queried $C$ and $\mathcal{V} = \emptyset$ otherwise.\\
            $\alpha \leftarrow \Adv_2({\ct}, \st)$ & $\alpha \leftarrow \Adv_2(\ct, \st)$\\
            The experiment outputs the state $\alpha$ & The experiment outputs the state $\alpha$\\
        \end{tabular}
    \end{table}

The $\QFE$ scheme is $1\text{-}1\text{-}\textsf{NA}$ \textbf{(single-message single-query non-adaptive)} simulation-secure (\textsf{SIM}--secure) if {for any adversary $\Adv$ that makes a single query to \textsf{KeyGen}, there exists one simulator $\Sim$, independent of the sampled messages,} such that the real and ideal distributions are computationally indistinguishable:
$$ \{ \Exp_{\Adv,\textsf{NA}}^{\textnormal{Real}}(1^\lambda)\}_{\lambda \in \N} \approx_c \{ \Exp_{\Adv,\textsf{NA}}^{\textnormal{Ideal}}(1^\lambda)\}_{\lambda \in \N}. $$

\begin{itemize}[align=left,leftmargin=2.8em]
    \item The scheme is $\textsf{M}\text{-}1\text{-}\textsf{NA}$ \textbf{(many-message single-query non-adaptive)} \textsf{SIM}--secure if the adversary can request $n$ messages $({\msg_1}, \ldots {\msg_n}, \st) \leftarrow \Adv_1^{\Keygen(\mk, \cdot)}(1^\lambda)$ and receives multiple ciphertexts $({\ct_1}, \ldots {\ct_n})$ where $n = \poly[\lambda]$ is not known at setup time.
    \item The scheme is $1\text{-}\textsf{M}\text{-}\textsf{NA}$ \textbf{(single-message many-query non-adaptive )} \textsf{SIM}--secure if $\Adv_1$ can make $q$ key queries to $\Keygen$ where $q = \poly[\lambda]$ is not known at setup time. { For the classical challenge message $m$, the leakage is $\mathcal{V}=((C_j,C_j(m)))_{j\in[q]}$, with one freshly generated evaluation state for each query.}
\end{itemize}

\end{definition}

\begin{definition}[Adaptive Sim-Security for secret-key QFE]
\label{def:simsecurityad}
Let $\lambda$ be the security parameter and let $\mathcal{C}_\lambda$ be the class of polynomially sized quantum circuits. Let $(\Setup, $ $\Keygen,$ $ \Enc, \Dec)$ be a QFE scheme, let $\Adv = (\Adv_1, \Adv_2)$ be a QPT adversary and let $\Sim = (\Sim_1,\Sim_2)$ be a QPT simulator.{ The simulator retains its state between invocations of $\Sim_1$ and $\Sim_2$. } 
   \begin{table}[H]
        \centering
        \begin{tabular}{p{6cm}|p{6cm}}
            $\Exp_{\Adv,\textsf{AD}}^{\textnormal{Real}}(1^\lambda)$ & $\Exp_{\Adv,\textsf{AD}}^{\textnormal{Ideal}}(1^\lambda)$ \\
            $\mk \leftarrow \Setup(1^\lambda)$&$\mk \leftarrow \Setup(1^\lambda)$ \\
            $ (\msg, \st) \leftarrow \Adv_1^{\Keygen(\mk, \cdot)}(1^\lambda)$ & $ (\msg,\st) \leftarrow \Adv_1^{\Keygen(\mk, \cdot)}(1^\lambda)$ \\
            ${\ct} \leftarrow \Enc(\mk, \msg)$& ${\ct} \leftarrow {\Sim_1(1^\lambda,\mk,\mathcal{V},1^{|\msg|})}$ \\
            & \quad where ${\mathcal{V}=(C,C(\msg))}$ if $\Adv_1$\\ & \quad queried  $C$ and $\mathcal{V} = \emptyset$ otherwise.\\
            $\alpha \leftarrow \Adv_2^{\Keygen(\mk, \cdot)}({\ct}, \st)$ & $\alpha \leftarrow \Adv_2^{\Sim_2(1^\lambda, \mk, C, C(\msg), 1^{|\msg|})}(\ct, \st)$\\
            The experiment outputs the state $\alpha$ & The experiment outputs the state $\alpha$\\
        \end{tabular}
    \end{table}
The $\QFE$ scheme is $1\text{-}1\text{-}\textsf{AD}\text{-}\textsf{SIM}$ \textbf{(single-message single-query adaptive)} simulation-secure if {for any adversary $\Adv$ that makes a single query to \textsf{KeyGen}, there exists one stateful simulator $\Sim$, independent of the sampled messages,} such that the real and ideal distributions are computationally indistinguishable:
$$ \{ \Exp_{\Adv,\textsf{AD}}^{\textnormal{Real}}(1^\lambda)\}_{\lambda \in \N} \approx_c \{ \Exp_{\Adv,\textsf{AD}}^{\textnormal{Ideal}}(1^\lambda)\}_{\lambda \in \N} $$.

\begin{itemize}[align=left,leftmargin=2.8em]
    \item The scheme is $\mathsf{M}\text{-}1\text{-}\textsf{AD}$ \textbf{(many-message single-query adaptive)} \textsf{SIM}--secure if the adversary can request $n$ 
    messages $({\msg_1}, \ldots {\msg_n}, \st) \leftarrow \Adv_1^{\Keygen(\mk, \cdot)}(1^\lambda)$ and receives multiple ciphertexts $({\ct_1}, \ldots {\ct_n})$ where $n = \poly[\lambda]$ is not known at setup time.
    \item The scheme is $1\text{-}\mathsf{M}\text{-}\textsf{AD}$ \textbf{(single-message many-query adaptive)} \textsf{SIM}--secure if $\Adv_1$ can make $q_1$ key queries to $\Keygen$ and ${\Adv_2}$ can make $q_2$ key queries to $\Keygen$ ($\Sim$ in the ideal world) where $q_1 = \poly[\lambda], q_2 = \poly[\lambda]$ are not known at setup time.{ For the classical challenge message $m$, $\Sim_1$ receives one fresh evaluation $C_j(m)$ for each pre-challenge query, while $\Sim_2$ receives a fresh evaluation only when the corresponding post-challenge query arrives.} 
     
\end{itemize}

\end{definition}

\subsection{Puncturable Pseudorandom Functions}

Puncturable pseudorandom functions (\textsf{PPRF}), defined in \cite{SW14}, enable the generation of punctured keys which can be used to evaluate the pseudorandom function on all inputs except a designated polynomial set of inputs.  

\begin{definition}[Puncturable Pseudorandom Functions]
\label{def:puncturable}
    Let $\lambda$ be the security parameter. Let $n\coloneqq n(\lambda)$ and $m\coloneqq m(\lambda)$ be polynomials in the security parameter $\lambda$. A post-quantum puncturable pseudorandom function family $(\textsf{PPRF})$ consists of a tuple of {PPT} algorithms $(\textsf{KeyGen},\textsf{Puncture},F)$, where $F:\{0,1\}^\lambda\times \{0,1\}^n\rightarrow \{0,1\}^m$ {is a deterministic polynomial classical algorithm,  satisfying the following conditions:}
    \begin{itemize}
        \item Let $\adv$ be any QPT algorithm that outputs a polynomial-sized set $S\subset \{0,1\}^n$. Then, for any $x\notin S$,
        \begin{align*}            \Pr[F(\textsf{K},x)=F(\textsf{K}_S,x)|\textsf{K}\leftarrow \textsf{KeyGen}(1^\lambda),\textsf{K}_S\leftarrow \textsf{Puncture}(\textsf{K},S)]=1.
        \end{align*}
        \item Let $\adv=(\adv_1,\adv_2)$ be any pair of QPT algorithms such that $\adv_1$ outputs a polynomial-sized set $S\subset \{0,1\}^n$ and a state $\sigma$. Let $\textsf{K}\leftarrow \textsf{KeyGen}(1^\lambda)$ and $\textsf{K}_S\leftarrow \textsf{Puncture}(\textsf{K},S)$. Then, 
        \begin{align*}
           \lvert \Pr[\adv_2(\sigma,S,\textsf{K}_S,F(\textsf{K},S))=1]-\Pr[\adv_2(\sigma,S,\textsf{K}_S,U_{m\cdot \lvert S\rvert})=1]\rvert \leq \negl[\lambda].
        \end{align*}
        where $F(\textsf{K},S)=(F(\textsf{K},x_1),\ldots, F(\textsf{K},x_k))$ and $S=\{x_1,\ldots,x_k\}$ is the enumeration of the elements of $S$ in lexicographical order and $U_\ell$ is the uniform distribution over binary strings of length $\ell$.
    \end{itemize}
\end{definition}

{Post-quantum \textsf{PPRF}s can be built from post-quantum classical \textsf{OWF}s using the standard puncturing constructions \cite{BW13,BGI14}.}

\begin{lemma}
    Assuming post-quantum deterministic \textsf{OWF}s exist, then for any polynomials $m$ and $n>\lambda$ in the security parameter $\lambda$, there exists a post-quantum \textsf{PPRF} family mapping $n$-bits to $m$-bits.
\end{lemma}

\subsection{Disappearing Obfuscation}

In this section, we define online or disappearing virtual grey-box obfuscation, similar to \cite{GZ21}. We let $\xx$ denote the required memory to run the scheme honestly and {let $\s$ denote the memory bound on the adversary.} In the  \textsf{BCSM}, this refers to classical memory and, in the \textsf{BQSM}, this refers to quantum memory. {In the $\textsf{BCSM}$, adversaries are classical and their total stored information is bounded by $\s$ bits. In the $\textsf{BQSM}$, classical memory is unrestricted and the total retained quantum state is bounded by $\s$ qubits at the storage checkpoints prescribed by the stream protocol.}  

\begin{definition}[Grey-Box Obfuscation in Bounded Storage Models]
\label{def:vgb}
An algorithm $\mathcal{O}$ is a \emph{$(\xx,\s)$ grey-box obfuscation $(\GB)$} for a class of classical polynomial circuits $\mathcal{C}_\lambda$ {if its streaming and evaluation algorithms are polynomial-time in the parameters and circuit size, quantum in the $\textsf{BQSM}$ and classical in the $\textsf{BCSM}$, and satisfy the following:}
\begin{enumerate}
\item (functionality) For any circuit $C\in \mathcal{C}_\lambda$, {the stream $\mathcal{O}(C)$ can be evaluated as it is received on any input $x$ chosen before evaluation, producing $C(x)$ with probability at least $1-\negl[\lambda]$.} 
\item For any circuit $C\in {\mathcal{C}_\lambda}$, the sender requires $\xx$ memory to send $\mathcal{O}(C)$ and the receiver requires $\xx$ memory to evaluate it.
\end{enumerate}
\end{definition}


\begin{definition}[Disappearing Grey-Box Security]
A $(\xx,\s)$ $\GB$ $\mathcal{O}$ satisfies \emph{disappearing security}
if for every polynomial $k=k(\lambda)$ and every computationally
unbounded adversary $\mathcal{A}$ satisfying the memory bound $\s$,
there exist a computationally unbounded simulator $\mathcal{S}$ and
a polynomial $q=q(\lambda)$ such that the following holds for every
challenger $\mathsf{ch}$.

The challenger may adaptively choose at most $k$ possibly correlated
circuits $C_i\in\mathcal{C}_\lambda$ and provide auxiliary information
at any time. In the real experiment
$\mathsf{ExpAdv}_{\mathcal{A},\mathsf{ch},\mathcal{O}}(1^\lambda,k)$,
the adversary receives fresh streams $\mathcal{O}(C_i)$, which may
overlap and be interleaved. Each stream follows its prescribed storage
checkpoints, and the memory bound applies jointly across all streams
and interactions. The challenger cannot provide additional storage
for the adversary.

In the ideal experiment
$\mathsf{ExpSim}_{\mathcal{S},\mathsf{ch},\mathcal{O}}(1^\lambda,k,q)$,
one stateful simulator $\mathcal{S}$ replaces $\mathcal{A}$.
For each obfuscation instance, $\mathcal{S}$ may make at most $q$
adaptive classical queries to $C_i$, starting when the corresponding
stream begins and ending when that stream, including its prescribed
memory restrictions, is complete. Queries to different active
programs may be interleaved. The simulator receives the public
parameters, input, output, and circuit sizes, and stream start and
end notifications, but no circuit descriptions except through
explicitly revealed auxiliary information. All other interaction
with $\mathsf{ch}$ is unchanged.

Both experiments output the bit produced by $\mathsf{ch}$, and
\[
\left|
\Pr\!\left[
\mathsf{ExpAdv}_{\mathcal{A},\mathsf{ch},\mathcal{O}}
(1^\lambda,k)=1
\right]
-
\Pr\!\left[
\mathsf{ExpSim}_{\mathcal{S},\mathsf{ch},\mathcal{O}}
(1^\lambda,k,q)=1
\right]
\right|
\leq \negl[\lambda].
\]
\end{definition}

An information-theoretically secure disappearing one-time program was constructed in the \textsf{BQSM} \cite{BS23}. {A disappearing one-time program implies disappearing $\GB$. This is not immediate from the one-time program definition given in \cite{BS23}, which only considers the security of a single one-time program. However, the same proof can easily be adapted to many one-time programs and this property is used in their work. Therefore, we have the following result. }

\begin{theorem}[\cite{BS23}]
\label{thm:otp}
For any polynomial $\s$ in the security parameter $\lambda$, {there exists a $(0,\s)$ one-time compiler in the $\BQSM$ for classical polynomial circuits satisfying information-theoretic security and disappearing security as defined in \cite{BS23}.} 
\end{theorem}

Meanwhile, in the \textsf{BCSM}, \cite{GZ21} provided candidate constructions for $(\xx,O(\xx^2))$ $\GB$ conjectured to satisfy  disappearing and information-theoretic security.

%% file: first_imposs_reviewed.tex
\section{Impossibility of Simulation-Security under Many Ciphertext Queries}

In this section, we demonstrate the impossibility of simulation-secure $\QFE$ in the setting where the adversary receives many ciphertext queries and a single adaptive functional key query. First, we show a natural incompressibility result adapted from \cite{NS06}. 

\begin{definition}[Compression Circuits]
    Let $n=n(\lambda)$ and $k=k(\lambda)$ be nonnegative integer-valued functions on the security parameter $\lambda$. A pair of (non-uniform) QPT algorithms $(\Comp,\Decomp)$ are $(n,k)$-\emph{compression circuits} if, for all sufficiently large $\lambda$, there exists a polynomial-size shared advice state $\rho_\lambda^{AB}$, independent of $m$, whose registers $A$ and $B$ are given to $\Comp$ and $\Decomp$, respectively, such that:
    \begin{itemize}
        \item For any $m\in \{0,1\}^n$, $\lvert \Comp(m;A)\rvert \leq n-k$.
        \item When both algorithms act on their respective registers of the same joint state $\rho_\lambda^{AB}$,
        \begin{align*}
            \Pr_{m\gets \{0,1\}^n}[\Decomp(\Comp(m;A);B)=m]\geq 1 - \negl[\lambda].
        \end{align*}
    \end{itemize}
\end{definition}

\begin{lemma}[Consequence of Theorem 1.3(2) in \cite{NS06}]
\label{thm:compression}
    Suppose Alice wishes to convey $n$ bits to Bob by communicating over an entanglement-assisted quantum channel. For any choice of the shared entangled state, and any protocol using this prior entanglement, such that for uniformly random $x \in \{0,1\}^n$ the probability that Bob correctly recovers $x$ is at least $p \in (0,1]$, the total number of qubits $m_A$ sent by Alice to Bob, over all the rounds of communication, is at least
\[
m_A \;\geq\; \tfrac{1}{2}\Bigl(n - \log \tfrac{1}{p}\Bigr),
\]
independent of the number of qubits sent by Bob to Alice.
\end{lemma}

\begin{lemma}[Incompressibility]
\label{lem:incompress}
    Let $n=n(\lambda)$ and $k=k(\lambda)$ be nonnegative integer-valued functions on the security parameter $\lambda$. There does not exist a pair of $(n,k)$-compression circuits for any $k> n/2$. 
\end{lemma}
\begin{proof}
    
    Assume for contradiction that $(\Comp,\Decomp)$ are $(n,k)$-compression circuits. 

    By definition of compression circuits, for large enough $\lambda$, there exists a shared advice state $\rho_\lambda^{AB}$ such that
\begin{align*}
            \Pr_{{m\gets \{0,1\}^n}}[\Decomp(\Comp(m;A);B)=m]\geq 1- \negl[\lambda],
        \end{align*}
Therefore, $\Decomp$ recovers $m$ with probability $ 1-\negl[\lambda]$  using at most $n-k<n/2$ qubits of communication from $\Comp(m;A)$ along with shared quantum state $\rho_\lambda^{AB}$. By \cref{thm:compression}, 
  with $p=1-\negl[\lambda]$, we have $\log\frac1p \le \negl[\lambda]$, so any such protocol requires $m_A\geq\frac{1}{2}\bigl(n-\negl[\lambda]\bigr)>(n-1)/2$ qubits for sufficiently large $\lambda$. However, $m_A\leq n-k<n/2$ and $m_A$ is an integer, so $m_A\leq(n-1)/2$, a contradiction.
\qed
\end{proof}

\begin{theorem}
\label{thm:impM1AD}
    There does not exist a $\textsf{M}\text{-}1\text{-}\textsf{AD}$ \textsf{SIM}--secure $\QFE$ scheme. 
\end{theorem}

\begin{proof}
Assume for contradiction that there exists a $\textsf{M}\text{-}1\text{-}\textsf{AD}$ \textsf{SIM}--secure $\QFE$ scheme $(\Setup,\Keygen,$ $\Enc,\Dec)$. Modify the algorithm $\Dec$ by deferring any measurements to the end, so that it is modeled as a sequence of reversible gates followed by a final measurement. 
Since $\textsf{KeyGen}$ is a QPT algorithm, there must exist some integer-valued polynomial $p$ such that the size of the state of the functional key $\fk_I$ for the identity function is bounded by $p=p(\lambda)$. By padding if necessary and increasing $p$ to include any length encoding, we assume that this key has a fixed $p$-qubit format. Set $n\coloneqq 2p+\lambda$. 

Consider the behavior of the following QPT adversary $\adv=(\adv_1,\adv_2)$ in the real experiment $\Exp_{\Adv, \textsf{AD}}^{\textnormal{Real}}(1^\lambda)$ given in the security definition of $\QFE$.  

    \begin{itemize}
        \item $\adv_1$ does not request any functional keys. 
        \item $\adv_1$ samples $n$ random messages $m_i\leftarrow \{0,1\}$ for $i\in [n]$.
        \item $\adv_1$ outputs $(m_i)_{i\in [n]}$ and retains a copy in its private state for $\adv_2$.
    \item $\adv_2$ receives ciphertexts $({\ct_i})_{i\in [n]}$ from the challenger, requests the functional key for the identity function and receives a state $\textsf{fk}_I$ from the challenger. It outputs 0 if the response is not in the prescribed $p$-qubit format. 
    \item For each $i\in[n]$ in order, $\adv_2$ applies the coherent implementation of $\Dec$ to $\textsf{fk}_I$ and $\ct_i$, measures the output in the computational basis to obtain $\tilde{m}_i$, and applies the inverse computation. It retains all work registers until this inversion and keeps the measured bit for the final comparison. 
    \item If $\tilde{m}_{i}=m_i$ for all $i\in [n]$, then $\adv_2$ outputs 1 and 0 otherwise.  
    \end{itemize}

The correctness of $\QFE$ guarantees that in the real experiment, $\textsf{Dec}(\textsf{fk}_I,\ct_1)$ returns $m_1$, except with negligible probability. Since the correct output is a fixed classical bit, by the gentle measurement lemma (\cref{lem:genlemeas}), we can recover a state negligibly close to $\fk_I$ after this computation (see Remark 4.2 in \cite{KMY25} for a similar argument). The same bound holds for the joint state of the key and all ciphertexts. A hybrid argument over the polynomially many decryptions shows that this procedure recovers all messages with only negligible total error. Hence, the output of $\adv$ in the real experiment is 1, except with negligible probability. 

By the security of $\QFE$, there exists a stateful QPT simulator $\Sim=(\Sim_1,\Sim_2)$ such that
$$ \{ \Exp_{\Adv, \textsf{AD}}^{\textnormal{Real}}(1^\lambda)\}_{\lambda \in \N} \approx_c \{ \Exp_{\Adv, \textsf{AD}}^{\textnormal{Ideal}}(1^\lambda)\}_{\lambda \in \N}. $$
We now analyze the simulation in $\Exp_{\Adv, \textsf{AD}}^{\textnormal{Ideal}}(1^\lambda)$. $\Sim_1$ receives no functional key queries from $\adv_1$. Therefore, it outputs $n$ ciphertexts $(\tilde{\ct}_i)_{i\in [n]}\gets \Sim_1(1^\lambda,\mk, 1^n)$ for $n$ messages without knowing any evaluations. The simulator may retain a register $R$ entangled with these ciphertexts; we make this retained state explicit below. Next, $\adv_2$ requests the functional key for the identity function. Therefore, $\Sim_2$ receives the messages $m\coloneqq (m_1,\ldots, m_n)$ and, using its retained register $R$, outputs $\tilde{\textsf{fk}}_I\gets \Sim_2(1^\lambda,\mk,I,m,1^n;R)$. Any response outside the prescribed $p$-qubit format causes $\adv_2$ to output 0. 

Set $\tilde{\ct}$ be the collection of all ciphertexts $(\tilde{\ct}_i)_{i\in [n]}$. Consider the operation of $\adv_2$: it takes $\tilde{\fk}_I$ and $\tilde{\ct}$ as input, applies the same coherent decryption, measurement, and inversion procedure to each $i\in[n]$, and obtains $\tilde{m}\coloneqq(\tilde{m}_1,\ldots,\tilde{m}_n)$. This procedure uses $m$ only in the final comparison. 

For the simulation to be indistinguishable, we require that the simulated key has the prescribed format and this procedure returns $\tilde{m}=m$, jointly with probability at least $1-\negl[\lambda]$. 

On the other hand, since $(m_i)_{i\in [n]}$ are sampled at random, there are $2^n$ possible values for these messages. In essence, this means that the $\Sim$ and $\adv$ give a pair of non-uniform compression circuits (with shared advice $\rho_\lambda^{AB}$) as follows. 

Let $\rho_\lambda^{AB}$ be the joint state obtained by running $\Setup$ and $\Sim_1$ before any message evaluations are provided. Register $A$ contains $\mk$ and the simulator's retained register $R$, while $B$ contains $\tilde{\ct}$. This state is independent of $m$. The compression circuit $\Comp(m;A)$ continues $\Sim_2$ on $(1^\lambda,\mk,I,m,1^n;R)$ and outputs its functional key, replacing any response outside the prescribed format by a fixed $p$-qubit state. Meanwhile, $\Decomp(\tilde{\fk}_I;B)$ applies the decryption procedure of $\adv_2$, omitting the final comparison with $m$, and outputs $\tilde{m}$.  

We already deduced that $\Pr[\tilde{m}=m]\geq 1-\negl[\lambda]$ and given that $n>2p$, this means that $(\Comp,\Decomp)$ are $(n,n-p)$-compression circuits. But this contradicts \cref{lem:incompress}, so  $\textsf{M}\text{-}1\text{-}\textsf{AD}$ \textsf{SIM}--secure $\QFE$ schemes do not exist.
\qed
\end{proof}

Interestingly, we obtain another impossibility result for $\QFE$ from \cref{lem:incompress} and \cite{AGV13}. Recall that Agrawal et al.~\cite{AGV13} showed that in the classical setting $1\text{-}\textsf{M}\text{-}\textsf{NA}$ \textsf{SIM}--secure $\FE$ is impossible assuming the existence of $\PRF$. Notably this result only applies to classical schemes as it relies on an incompressibility result for classical information. With \cref{lem:incompress}, we have an incompressibility result for classical information encoded into quantum states. Their argument then extends to quantum ciphertexts and functional keys, provided that the $\PRF$ is secure against QPT distinguishers making classical oracle queries. The adversary queries keys for circuits that evaluate the $\PRF$ at distinct inputs and then encrypts its seed. Since the outputs are classical, the same gentle-measurement argument allows it to reuse the ciphertext. The adversary also rejects ciphertexts exceeding the size bound for real ciphertexts. Replacing the pseudorandom outputs with independent random bits would then compress more than twice as many bits as the number of qubits in the simulated ciphertext, contradicting \cref{lem:incompress}. This gives the following theorem.

\begin{theorem}
    There does not exist a $1\text{-}\textsf{M}\text{-}\textsf{NA}$ \textsf{SIM}--secure $\QFE$ scheme assuming the existence of post-quantum $\PRF$s. 
\end{theorem}

%% file: section5_reviewed.tex
\providecommand{\QFEreplace}[2]{}
\renewcommand{\QFEreplace}[2]{%
  \begingroup
  \let\QFEoriginalcref\cref
  \let\QFEoriginalcite\cite
  \renewcommand{\cref}[1]{\mbox{\QFEoriginalcref{##1}}}%
  \renewcommand{\cite}[1]{\mbox{\QFEoriginalcite{##1}}}%
  \sout{#1}\endgroup\ {\color{red}#2}}
\providecommand{\QFEinsert}[1]{{\color{red}\relax#1}}
\providecommand{\QFEmathreplace}[2]{\text{\sout{\ensuremath{#1}}}\;{\color{red}#2}}

\providecommand{\QFEstackreplace}[2]{\begin{gathered}\text{\sout{\ensuremath{#1}}}\\{\color{red}#2}\end{gathered}}

\section{Impossibility of Simulation-Security for Succinct Schemes}

In this section, we show the impossibility of succinct $1\text{-}1\text{-}\textsf{NA}$ \textsf{SIM}--secure $\QFE$. The proof proceeds by showing that there is a class of incompressible states and that a succinct $\QFE$ scheme would constitute a compression scheme for this class of states. The incompressible states are the outputs of a pseudorandom state generator. To show their incompressibility, we first show that Haar random states, which are indistinguishable from $\PRS$ for a bounded adversary, are incompressible. 

For the proof, we will need the following Lemma, which can for example be found in \cite{watrous,PSW06}. 
\begin{lemma}[Levy's Lemma]
\label{lem:levy}
Given a Lipschitz function $f : \mathbb{S}^r \to \mathbb{R}$ defined on the 
$r$-dimensional hypersphere $\mathbb{S}^r$, and a point 
$\psi \in \mathbb{S}^r$ chosen uniformly at random,
\[
\Pr\big[ |f(\psi) - \Expect f | \geq \epsilon \big] 
   \leq 2 \exp\!\left( \frac{-2C(r+1)\epsilon^2}{\eta^2} \right),
\]
where $\eta$ is the Lipschitz constant of $f$, and $C$ is a positive constant (which can be taken to be $C = (18\pi^3)^{-1}$).
\end{lemma}

Pure states $\ket{\psi}$ sampled Haar randomly from a space of dimension $d$ are in correspondence with points $\psi \in \mathbb{S}^{2d-1}$ sampled uniformly at random. Throughout this section, we use squared fidelity, so that $F(\ket{\psi}\bra{\psi},\rho)=\bra{\psi}\rho\ket{\psi}$.

\begin{lemma}[Incompressibility of Haar Random States]
\label{lem:incomphaar} Let $n,m \in \N$ such that $m<n$ and let $d = 2^n, k = 2^m$. For all CPTP channels $\Comp: D(\Hcal^{n})\rightarrow D(\Hcal^{m}), \Decomp: D(\Hcal^{m}) \rightarrow D(\Hcal^n)$ it holds that 


         $$
  \Pr[F\big(|\phi\rangle\langle\phi|,\ \Decomp \circ \Comp(|\phi\rangle\langle\phi|)\big) \ge \frac{2^m+1}{2^n+1} +
  \epsilon] \le 2e^{-\Theta(2^n\epsilon^2)}
  $$
 for every $\epsilon>0$ where the probability is taken over choosing a Haar random state.
\end{lemma} 

\begin{proof}
    Let $ \Phi = \Decomp \circ \Comp$, then the average fidelity of a  state to itself after passing through this channel is
    $$ \Favg (\Phi) = \int d \psi \langle \psi | \Phi (|\psi\rangle \langle \psi |) |\psi\rangle$$
    where the integral is over the Haar measure on $\C^d$. 

    The entanglement fidelity is defined as 

    $$ \Fe(\Phi) = \langle \phi^+| (\Id \otimes \Phi) (|\phi^+ \rangle\langle \phi^+| ) |\phi^+\rangle$$
    
    where $\phi^+$ is the maximally entangled state in $d$ dimensions $|\phi^+\rangle = \frac{1}{\sqrt{d}} \sum_{i=1}^d |i\rangle \otimes |i\rangle$. 

    The following relationship between the average fidelity and the entanglement fidelity is described by Horodecki et al.~\cite{horodecki99} and Nielsen~\cite{nielsen02}:
    $$ \Favg(\Phi) = \frac{d \; \Fe(\Phi) + 1}{d + 1}\enspace.$$

    Let $\{A_i\}$ be a Kraus representation of $\Comp$ and let $\{B_j\}$ be a Kraus representation of $\Decomp$. Let $\{C_{\alpha}\}$ be the corresponding Kraus representation of $\Phi$, given by $C_{i,j}=B_jA_i$. 

\begin{align*}
     \Fe(\Phi) &= \langle \phi^+|\sum_{\alpha} (\Id \otimes C_\alpha) (|\phi^+ \rangle\langle \phi^+| ) (\Id \otimes C_\alpha^\dagger) |\phi^+\rangle\\
     &= \sum_{\alpha}\langle \phi^+| (\Id \otimes C_\alpha) (|\phi^+ \rangle\langle \phi^+| ) (\Id \otimes C_\alpha^\dagger) |\phi^+\rangle\\
     &= \sum_{\alpha} \left|\langle\phi^+| (\Id \otimes C_\alpha)  |\phi^+\rangle \right|^2\\
     &= \sum_{\alpha} \left| \left(\frac{1}{\sqrt{d}} \sum_{j=1}^d \langle j| \otimes \langle j| \right) (\Id \otimes C_\alpha) \left(\frac{1}{\sqrt{d}} \sum_{i=1}^d |i\rangle \otimes |i\rangle \right) \right|^2\\
      &= \sum_{\alpha} \left| \frac{1}{d}\left( \sum_{j=1}^d \langle j| \otimes \langle j| \right) \left( \sum_{i=1}^d |i\rangle \otimes C_\alpha|i\rangle \right) \right|^2\\
      &= \sum_{\alpha} \left|\frac{1}{d} \sum_{j=1}^d \sum_{i=1}^d \underbrace{\langle j|i\rangle}_{0\text{ if }i\neq j} \otimes \langle j|  C_\alpha|i\rangle\right|^2\\
      &= \sum_{\alpha} \left|\frac{1}{d} \sum_{i=1}^d \langle i|  C_\alpha|i\rangle\right|^2\\
      &= \frac{1}{d^2}\sum_{\alpha} \left| \Tr(C_\alpha)\right|^2 \enspace.
\end{align*}

Due to the compression requirement, $A_i$ outputs states in a $k$-dimensional space, and since $B_j$ acts on a $k$-dimensional space, the range of $B_j$ is contained in a $k$-dimensional subspace of $\C^d$. Therefore, the operator $C_{ij} = B_jA_i$ which is a Kraus operator for the composed channel $\Phi$ has rank$(C_{ij}) \leq k$.

Using the inequality $|\Tr(X)|^2 \leq \text{rank}(X)\Tr(X^\dagger X)$ and the fact that $\sum_{ij} C^\dagger_{ij}C_{ij} = \Id_d$ we get

\begin{align*}
    \Fe(\Phi) &= \frac{1}{d^2}\sum_{i,j} \left| \Tr(C_{ij})\right|^2\\
    &\leq \frac{1}{d^2}\sum_{i,j} k \Tr(C_{ij}^\dagger C_{ij})\\
    &= \frac{k}{d^2} \Tr(\Id_d)\\
    &= \frac{k}{d} \enspace.
\end{align*}

Therefore we get that the average fidelity of $\Phi$ is $$ \Favg(\Phi) \leq \frac{d \frac{k}{d}+1}{d+1} = \frac{k+1}{d+1} \enspace.$$

Plugging in $d = 2^n, k = 2^{m}$ we get 
\begin{align*}
    \Favg(\Phi) &\leq \frac{2^m + 1}{2^n + 1} \\
    &= \frac{2^m}{2^n} + \frac{1}{2^n+1} - \frac{\frac{2^m}{2^n}}{2^n+1}\\
    &\leq \frac{1}{2^{n-m}} + O\left(\frac{1}{2^{n}}\right)\enspace.
\end{align*}

In particular, if we decide to compress by only one qubit, which corresponds to setting $m = n-1$, we get 
$$ \Favg(\Phi) \leq \frac{1}{2} + O\left(\frac{1}{2^{n}}\right).$$

For the final step define the function $f: \mathbb{S}^{2d-1} \rightarrow \mathbb{R}$ as $f(\psi) = \bra{\psi} \Phi(\ket{\psi}\bra{\psi})\ket{\psi}$. We then see that $\mathbb{E}f= \mathsf{F}_{avg}$. Now, we show that $f$ is 4-Lipschitz.
\begin{alignat*}{2}
    | f(\phi) - f(\psi)| &= &| &\Tr [ \Phi(\ket{\psi}\bra{\psi}) \ket{\psi}\bra{\psi} ] -\Tr [ \Phi(\ket{\phi}\bra{\phi}) \ket{\phi}\bra{\phi} ] |\\
    &=& | &\Tr\left[ \Phi(\ket{\psi}\bra{\psi}) \ket{\psi}\bra{\psi}\right] - \Tr\left[ \Phi(\ket{\phi}\bra{\phi}) \ket{\psi}\bra{\psi}\right]\\
    & &&+ \Tr\left[ \Phi(\ket{\phi}\bra{\phi}) \ket{\psi}\bra{\psi}\right] -\Tr\left[ \Phi(\ket{\phi}\bra{\phi}) \ket{\phi}\bra{\phi}\right] |\\
    &=&|& \Tr [ (\Phi(\ket{\psi}\bra{\psi}) - \Phi(\ket{\phi}\bra{\phi}) )\ket{\psi}\bra{\psi} ] +\Tr [ \Phi(\ket{\phi}\bra{\phi}) ({\ket{\psi}\bra{\psi} - \ket{\phi}\bra{\phi}}) ] |\enspace.\\
    &\leq&|&\Tr [ (\Phi(\ket{\psi}\bra{\psi}) - \Phi(\ket{\phi}\bra{\phi}) )\ket{\psi}\bra{\psi} ]| + |\Tr [ \Phi(\ket{\phi}\bra{\phi}) (\ket{\phi}\bra{\phi} - \ket{\psi}\bra{\psi}) ] |\enspace.\\
\end{alignat*}
We can compute the first term by applying the Hölder inequality $|\Tr(AB)| \leq \lVert A \rVert_1 \lVert B \rVert_\infty$.
We then get
\begin{align*}
 {\lvert}\Tr [ (\Phi(\ket{\psi}\bra{\psi}) - \Phi(\ket{\phi}\bra{\phi}) )\ket{\psi}\bra{\psi} ]{\rvert} &=
 {\lvert}\Tr [ (\Phi(\ket{\psi}\bra{\psi} - \ket{\phi}\bra{\phi}) )\ket{\psi}\bra{\psi} ]{\rvert}  \\
 &\leq \lVert \Phi(\ket{\psi}\bra{\psi} - \ket{\phi}\bra{\phi})\rVert_1 \lVert\ket{\psi}\bra{\psi}\rVert_\infty \enspace.
\end{align*}

For pure states it holds that ${\lVert |\psi\rangle \langle \psi|\rVert_\infty} = 1$ and since CPTP maps are contractive we have $\lVert \Phi(\ket{\psi}\bra{\psi} - \ket{\phi}\bra{\phi})\rVert_1 \leq \lVert \ket{\psi}\bra{\psi} - \ket{\phi}\bra{\phi} \rVert_1$.

We can compute the second term by making use of trace cyclicity and also applying the Hölder inequality:
$${\lvert}\Tr [ \Phi(\ket{\phi}\bra{\phi}) (\ket{\phi}\bra{\phi} - \ket{\psi}\bra{\psi}) ]{\rvert} \leq \lVert\ket{\phi}\bra{\phi} - \ket{\psi}\bra{\psi} \rVert_1 \lVert\Phi(\ket{\phi}\bra{\phi})\rVert_\infty \enspace.$$

For any density operator such as $\Phi(\ket{\phi}\bra{\phi}) $ it holds that ${\lVert\Phi(\ket{\phi}\bra{\phi})\rVert_\infty} \leq 1$.
Therefore, using \cref{fact:l1l2},
\begin{align*}
    | f(\phi) - f(\psi)| &\leq 2\lVert\ket{\phi}\bra{\phi} - \ket{\psi}\bra{\psi} \rVert_1\\
    & \leq 4 \lVert \ket{\phi} - \ket{\psi} \rVert_2 \enspace,
\end{align*}
which shows that $f$ is 4-Lipschitz.

Now, we can apply Levy's Lemma 
\[
\Pr\big[ |f(\psi) - \Expect f | \geq \epsilon \big] 
   \leq 2 \exp\!\left( \frac{-4C d\epsilon^2}{4^2} \right) \enspace,
\]
where $d = 2^n$ and $C$ is a constant as specified in \cref{lem:levy} .

It follows that 

$${\begin{aligned}\Pr\left[f(\psi)\geq\frac{2^m+1}{2^n+1}+\epsilon\right]\
\leq\Pr[f(\psi)\geq\Expect f+\epsilon]\end{aligned}} \leq 2\exp\left( -\Theta (2^n \epsilon^2) \right)$$


In particular, even for a single qubit of compression $m = n-1$, recovering the state with fidelity at least $1-\negl[n]$ has negligible success probability. Set $\epsilon=1-\negl[n]-\frac{2^{n-1}+1}{2^n+1}$, which is bounded below by a positive constant for sufficiently large $n$.
\begin{align*}
   {\Pr[f(\psi)\geq 1-\negl[n]]} 
   &\leq 2 \exp\!\left( \frac{-C 2^n {\epsilon^2}}{4} \right)\\
   & \leq 2 \exp\left(-\Theta(2^n)\right) \enspace.
\end{align*}

\qed
\end{proof}

Now we can relate the incompressibility of Haar random states to pseudorandom states.

\begin{corollary}[Incompressibility of Pseudorandom States]
\label{lem:incomprPRS}
     {Let $\lambda \in \N$ and let $\Gen$ be a $\PRS$ generator with key length $\lambda$ and output length $n = \poly[\lambda]\geq 1$.} There do not exist QPT algorithms $\Comp: D(\Hcal^n)\rightarrow D(\Hcal^{n-1}), \Decomp: D(\Hcal^{n-1}) \rightarrow D(\Hcal^n)$ such that 
    $$ \Pr_{k \leftarrow {\{0,1\}^{\lambda}}} \left[ F(|\psi_k\rangle \langle \psi_k|, \Decomp \circ \Comp (|\psi_k \rangle \langle \psi_k |)) \geq 1- \negl[\lambda] \right] \geq 1 - \negl[\lambda]$$
    where  $\Gen(1^\lambda,k) = |\psi_k\rangle $.
\end{corollary}

\begin{proof}
    Assume towards contradiction that there exist algorithms $\Comp, \Decomp$ that can compress and decompress pseudorandom states and achieve fidelity $F \geq 1 - \negl[\lambda]$ on all keys but a negligible fraction. These algorithms can then be used to break the pseudo-randomness property of $\PRS$ as follows.
    The adversary receives two copies of the $\PRS$ state or Haar random state in the $\PRS$ security game. Then the adversary against the pseudorandomness of $\PRS$ can distinguish $\PRS$ from Haar random states by applying $\Comp$ and $\Decomp$ to one copy and comparing the result with the untouched copy by applying a Swap test.
    The Swap test accepts with probability $\frac{1}{2}(1 + \bra{\psi}\Phi(\ket{\psi}\bra{\psi})\ket{\psi})$ where $|\psi\rangle$ is the original state and $\Phi(\ket{\psi}\bra{\psi})$ is the state after the compression and decompression algorithm was applied. Per assumption the fidelity is $F \geq 1 - \negl[\lambda]$ for every good key, so the Swap test accepts with probability at least $\frac{1}{2}(1 + 1 - \negl[\lambda])$ conditioned on such a key. On Haar random states, the average acceptance probability is $\frac{1}{2}(1+\Favg(\Phi))$, which we bound using the proof of \cref{lem:incomphaar}.
  Therefore, the Swap test accepts with probability 
    $$\Pr[\mathsf{SWAP} \text{ accepts}| |\psi\rangle \text{ is PRS}] \geq \frac{1}{2}(1 + 1 - \negl[\lambda]) \cdot (1-\negl) \geq 1 - \negl$$

    where the additional factor of $(1- \negl)$ comes from picking a good key.
    
   On Haar random states, the average-fidelity bound in the proof of \cref{lem:incomphaar}, with $m=n-1$, gives ${\Pr[\mathsf{SWAP} \text{ accepts}\mid\psi\text{ is Haar}]\leq\frac{1}{2}\left(1+\frac{2^{n-1}+1}{2^n+1}\right)\leq\frac{5}{6}}$.

    Overall, the probability that the adversary in the $\PRS$ security game wins is 
    \begin{align*}
        \Pr[\adv \text{ wins}] &= \frac{1}{2}\Pr[\mathsf{SWAP} \text{ accepts}{\mid |}\psi\rangle \text{ is PRS}] + \frac{1}{2} \Pr[\mathsf{SWAP} \text{ rejects}{\mid |}\psi\rangle \text{ is Haar}] \\
        &= \frac{1}{2}\Pr[\mathsf{SWAP} \text{ accepts}{\mid |}\psi\rangle \text{ is PRS}] + \frac{1}{2} (1 - \Pr[\mathsf{SWAP} \text{ accepts}{\mid |}\psi\rangle \text{ is Haar}]) \\
        &\geq \frac{1}{2}(1 - \negl) + \frac{1}{2}(1 - {\frac{5}{6}})\\
        &= \frac{1}{2}+ \frac{1}{2} -  {\frac{5}{12}} - \negl\\
        &= \frac{1}{2} + {\frac{1}{12}} - \negl \enspace.
    \end{align*}

Therefore, the adversary has non-negligible advantage in distinguishing $\PRS$ from Haar random states and we conclude that there cannot exist compression and decompression algorithms that recover a $\PRS$ with fidelity $1 - \negl[\lambda]$ with high probability. The same argument rules out expected fidelity $1-\negl[\lambda]$, and also applies when $\Comp$ and $\Decomp$ share efficiently sampled classical randomness independent of the input: the Haar average-fidelity bound holds for each fixed value of that randomness.
\qed
\end{proof}

Now we can show the main theorem of this section, the impossibility of succinct single-message, single-query \textsf{SIM}--secure $\QFE$.

\begin{theorem}
\label{thm:imp11NA}
Assuming the existence of $\PRS$ generators with $\lambda$-bit keys and any polynomial number of output qubits, there does not exist a succinct $1\text{-}1\text{-}\textsf{NA}$ \textsf{SIM}--secure $\QFE$ scheme, $\QFE = (\Setup, \Keygen, \Enc, \Dec)$,
for all polynomial-sized quantum circuits, where the master secret key is classical and the messages, functional keys and ciphertexts may be quantum, under the security definition in \cref{def:simsecurityna}.
\end{theorem}
\begin{proof}
    Assume $\PRS$ with the stated output lengths exists.  Assume towards contradiction that there also exists a succinct $1\text{-}1\text{-}\textsf{NA}$
    \textsf{SIM}--secure $\QFE$ scheme. Let $\lambda \in \N$ be the security parameter and $\mathcal{M}$ be the message space of the scheme with messages of size $\ell_M = \lambda$. Let $\ell_{CT} = \poly[\lambda,\ell_M]$ be the ciphertext size of the scheme. By succinctness, this bound is independent of the queried circuit size; we use a fixed $\ell_{CT}$-qubit ciphertext format, including any padding and length information. Such a scheme implies a compression and decompression function.
    
    First fix $C = \PRS.\Gen(1^\lambda,\cdot): \bin^{\lambda} \rightarrow D(\Hcal^{s})$, for some $s=\poly[\lambda]$ with $s>\ell_{CT}$. Consider the following QPT adversary. It requests the functional key for $C$, samples $k\leftarrow\bin^\lambda$, retains a reference copy of $|\psi_k\rangle=C(k)$, and submits $k$ as its challenge message. It rejects a ciphertext outside the prescribed format; otherwise, it decrypts using its functional key and applies the SWAP test to the result and the reference. Fix the simulator $\Sim$ guaranteed by simulation security for this adversary. For a classical master secret key $\mk$ sampled from $\QFE.\Setup$ and a given state $\rho$, define $\Comp(\rho;\mk)$ as follows:
     The simulator $\Sim$ obtains the master secret key $\mk$, a description of the circuit $\PRS.\Gen$ and the state $\rho \in D(\Hcal^s)$ and creates the ciphertext $\ct$.  The ciphertext $\ct$ is the output of the compression. More formally,
    $$\Comp(\rho;\mk)=\Sim(1^{\lambda}, \mk, (C,\rho),1^\lambda) \enspace.$$ If the simulator outputs a ciphertext outside the prescribed format, $\Comp$ replaces it by a fixed $\ell_{CT}$-qubit state. Thus, the compression output always has $\ell_{CT}$ qubits. The simulator receives the evaluation state $\rho$, not the seed $k$.

    
    

    The decompression $\Decomp$ proceeds as follows. Given an input $\ct$ and the same classical master key $\mk$, it samples an honest functional key $\fk_C\leftarrow\QFE.\Keygen(\mk,C)$ and runs the decryption algorithm of the $\QFE$ scheme to obtain
    $$\Decomp(\ct;\mk) = \QFE.\Dec(\fk_C, \ct) \enspace.$$ Conditioned on $\mk$, the honest functional key is independent of the input and the simulator workspace, so generating it during decompression gives the same joint distribution as generating it before the challenge. The maps therefore share only classical randomness independent of the input. By correctness, the adversary accepts with probability $1-\negl[\lambda]$ in the real world, and hence also in the ideal world by simulation security. Replacing malformed simulator outputs by a fixed state cannot decrease this lower bound, since the adversary rejects such outputs. The SWAP acceptance formula therefore implies that $\Decomp(\Comp(\rho;\mk);\mk)$ has expected fidelity $1-\negl[\lambda]$ with $\rho=|\psi_k\rangle\langle\psi_k|$, where the expectation is over $k$ and $\mk$. A QPT adversary in the $\PRS$ game can sample $\mk$ itself, apply these maps to one of its two input copies, and apply the SWAP test against the other copy. It accepts $\PRS$ inputs with probability $1-\negl[\lambda]$. For Haar inputs, conditioning on $\mk$ and applying the average-fidelity bound from \cref{lem:incomphaar} gives acceptance probability at most $\frac{1}{2}\left(1+\frac{2^{\ell_{CT}}+1}{2^s+1}\right)\leq\frac{5}{6}$. This gives a non-negligible distinguishing advantage, as in the proof of \cref{lem:incomprPRS}, and contradicts pseudorandomness. Therefore succinct $1\text{-}1\text{-}\textsf{NA}$ \textsf{SIM}--secure $\QFE$ cannot exist under the stated assumptions.
\qed
\end{proof}

\begin{remark}
    In the above proof, we use the fact that in \cref{def:simsecurityna} in the ideal world $\QFE.\Setup$ and $\QFE.\Keygen$ are run honestly by the experiment and not controlled by the simulator. If we allow the simulator to control $\Keygen$, it may retain a state entangled with the functional key. The compression and decompression procedures then share prior entanglement, so \cref{lem:incomphaar} does not directly apply. This is a limitation of the proof, rather than a demonstrated way to circumvent the ciphertext bound: standard teleportation of an arbitrary $s$-qubit state requires $2s$ classical bits of communication. In the real world, $\Keygen$ cannot depend on the message. 
  Furthermore, if the functional keys are classical, as in \cite{MM24}, this entanglement issue does not arise, and the same argument extends to simulator-generated classical keys under the stated assumptions. 
\end{remark}

%% file: section6_reviewed.tex
\providecommand{\QFEreplace}[2]{}
\renewcommand{\QFEreplace}[2]{%
  \begingroup
  \let\QFEoriginalcref\cref
  \let\QFEoriginalcite\cite
  \renewcommand{\cref}[1]{\mbox{\QFEoriginalcref{##1}}}%
  \renewcommand{\cite}[1]{\mbox{\QFEoriginalcite{##1}}}%
  \sout{#1}\endgroup\ {\color{red}#2}}
\providecommand{\QFEinsert}[1]{{\color{red}\relax#1}}
\providecommand{\QFEmathreplace}[2]{\text{\sout{\ensuremath{#1}}}\;{\color{red}#2}}

\section{Impossibility of Simulation-Security under Many Functional-Key Queries}

In this section, we establish the impossibility of $1\text{-}\mathsf{M}\text{-}\mathsf{NA}$ \textsf{SIM}--secure $\QFE$ under the assumption that $\PKE$ exists. We refer the reader to the technical overview (\cref{sec:tech}) for an intuitive description of the proof to aid readability. 

Note that the earlier impossibility result (\cref{thm:imp11NA}) relied on $\PRS$ and applied only to succinct schemes. By contrast, our new impossibility extends to the broader class of non-succinct schemes, though it requires a different assumption—namely, $\PKE$. Importantly, {the relation between $\PKE$ with classical keys and $\PRS$ is not fully understood} \cite{sattah25,morimae24}: {a quantum oracle separation rules out relativizing constructions of $\PKE$ with classical keys from $\PRS$ \cite{K21}, while the reverse implication remains open.} Thus, the two results represent {complementary approaches to establishing barriers} for $\QFE$.

\begin{theorem}
\label{thm:imp1MNA}
 Assuming the existence of public-key encryption with classical keys, there does not exist a (secret-key) $1\text{-}\textsf{M}\text{-}\textsf{NA}$ \textsf{SIM}--secure $\QFE$ scheme.
 \end{theorem}

\begin{proof}
    Assume for contradiction that there exists a (secret-key) $1\text{-}\textsf{M}\text{-}\textsf{NA}$ \textsf{SIM}--secure $\QFE$ scheme $\Pi_{\QFE}=(\Setup, \Keygen,\Enc,\Dec)$. 
    
    Also, assume that there exists a $\textsf{CPA}$--secure $\PKE$ scheme. This implies the existence of a \textsf{MK-CPA}--secure $\PKE$ scheme $\Pi_{\PKE}=(\PKKeygen, \PKEnc,\PKDec)$ by \cref{lem:many-key}. 
    
 Since $\Enc$ is a QPT algorithm, by padding ciphertexts and including
any required length encoding, we may fix an integer-valued polynomial
$q$ such that ciphertexts of $\lambda$-bit messages have a fixed
$q(\lambda)$-qubit format. Decryption removes this padding.

Consider the following fixed $\QFE$ adversary
$\mathcal{A}=(\mathcal{A}_1,\mathcal{A}_2)$.
It sets $n\coloneqq 2q(\lambda)+\lambda$, samples independent pairs
$(\mathsf{sk}_i,\mathsf{pk}_i)\leftarrow\textsf{PKE.Gen}(1^\lambda)$
for $i\in[n]$, and independently samples
$M\leftarrow\{0,1\}^{\lambda}$ and $j\leftarrow[n]$.
It requests the functional keys for
$C_i\coloneqq\textsf{PKE.Enc}(\mathsf{pk}_i,\cdot)$, submits $M$
as its challenge message, and retains
$\mathsf{st}=(j,M,\mathsf{sk}_j,(\fk_{C_i})_{i\in[n]})$.
On receiving $\ct$, $\mathcal{A}_2$ sets $\ct'_j=\perp$ if the
ciphertext format is invalid. Otherwise, it computes
$\ct'_j\leftarrow\Dec(\fk_{C_j},\ct)$.
It outputs $(\ct'_j,M,\mathsf{sk}_j)$.

Let $\Sim$ be the simulator guaranteed by simulation security for
this fixed adversary. In particular, $\Sim$ receives neither $j$
nor the private state $\mathsf{st}$.


    We construct a QPT adversary $\adv$ that breaks \textsf{MK-CPA} security of $\Pi_{\textsf{PKE}}$. For convenience, we recall the \textsf{MK-CPA} security experiment in \cref{fig:recall CPA}.

\begin{figure}[!htb]
   \begin{center} 
   \begin{tabular}{|p{12cm}|}
    \hline 
\begin{center}
\underline{$\PKE^{\textnormal{IND-MK-CPA}}_{\Pi_{\PKE},\adv}({\lambda})$}: 
\end{center}
\begin{enumerate}
    \item $\adv$ sends a number $n$ that is polynomial in $\lambda$.
    \item Challenger $\ch$ generates $(sk_i,pk_i)\gets \PKKeygen(1^\lambda)$ for every $i\in [n]$, and sends $(pk_i)_{i\in[n]}$ to $\adv$.
    \item $\adv$ outputs a message pair $(m_{0,i},m_{1,i})\in \{0,1\}^\lambda {\times} \{0,1\}^\lambda$ for every $i\in [n]$.
    \item $\ch$ samples a bit $b\leftarrow  \{0,1\}$ and generates $\ct_i\gets \PKEnc(pk_i,m_{b,i})$ for every $i\in [n]$. 
    \item $\ch$ sends $(\ct_i)_{i\in [n]}$ to $\adv$.
    \item $\adv$ submits a bit $b'$.
    \item The output of the experiment is $1$ if $b=b'$, and $0$ otherwise. 
\end{enumerate}
\ \\ 
\hline
\end{tabular}
    \caption{\textsf{MK-CPA} security experiment.}
    \label{fig:recall CPA}
    \end{center}
\end{figure}

We then describe the algorithm $\adv$ (in \cref{attack-against-cpa})  which is used to break \textsf{MK-CPA} security.{ The reduction uses the public keys and message already sampled in Steps 1--9 to implement the fixed $\QFE$ adversary. When $b=0$, their joint distribution, together with the retained secret key $sk_j$, is exactly the one specified above. When $b=1$, the reduction runs the same computation with the supplied evaluations, which need not arise from a common message.}

\begin{figure}[!htb]
   \begin{center} 
   \begin{tabular}{|p{12cm}|}
    \hline 
\begin{center}
{\underline{Algorithm of $\adv$}}: 
\end{center}
\begin{enumerate}
    \item $\adv$ sets $n\coloneqq {2q(\lambda)+\lambda}$.
    \item $\adv$ receives $(pk_i)_{i\in[n]}$ from the challenger $\ch$.
    \item $\adv$ samples a key pair $(\tilde{sk}_{i},\tilde{pk}_{i})\leftarrow \PKKeygen(1^\lambda)$ for each $i\in [n]$. 
    \item $\adv$ samples $M\leftarrow \{0,1\}^\lambda$ and sets $m_{0,i}=M$ for all $i\in [n]$. 
      \item $\adv$ samples $m_{1,i}\leftarrow \{0,1\}^\lambda$ for each $i\in [n]$. 
    \item $\adv$ outputs $(m_{0,i},m_{1,i})_{i\in [n]}$ and sends it to $\mathcal{C}$.
    \item $\adv$ receives  $(\ct_i)_{i\in [n]}$ from $\ch$.
    \item $\adv$ randomly picks $j\gets [n]$ and computes $\tilde{\ct}_{j}\leftarrow \PKEnc(\tilde{pk}_{j},M)$.
    \item $\adv$ updates and replaces: $\ct_j\gets \tilde{\ct}_j$, $sk_j\gets \tilde{sk}_j$, $pk_j\gets \tilde{pk}_j$, and $m_{1,j}\gets M$. 
    \item {$\adv$ runs the ideal-experiment computation with the supplied evaluations:}
    \begin{itemize}
        \item $\mk \leftarrow \Setup(1^\lambda)$.
        \item {$\adv_1$, using the public keys from Step 9 and the already sampled $M$ and $j$, requests} the functional keys for $C_i\coloneqq \PKEnc(pk_{i},\cdot)$ for each $i\in [n]$. Let the responses be $\fk_{C_i}\leftarrow \Keygen(\mk,C_i)$ for each $i\in [n]$.{ It outputs $M$ and retains $\textsf{st}$.}
        \item $\ct \leftarrow \Sim(1^\lambda, \mk, \mathcal{V})$, where $\mathcal{V} = (C_i, \ct_{i}, 1^{\lambda})_{i\in [n]}$. 
        \item ${\adv_2(\ct,\textsf{st})}$ {sets $\ct'_{j}=\bot$ if} $\lvert \ct\rvert {\neq q(\lambda)}$. {Otherwise, it computes} {${\ct}'_{j}\leftarrow \Dec(\fk_{C_{j}},\ct)$. It outputs $(\ct'_j,M,sk_j)$.} 
    \end{itemize}
    \item $\adv$ computes ${M}'\leftarrow \PKDec(sk_{j},{\ct}'_{j})$. 
    \item $\adv$ outputs $b'=0$ if $M'=M$ and $b'=1$ otherwise.
\end{enumerate}
\ \\ 
\hline
\end{tabular}
    \caption{Attack against \textsf{MK-CPA} security.}
    \label{attack-against-cpa}
    \end{center}
\end{figure}
\begin{Claim}
    If $b=0$ in ${\PKE^{\textnormal{IND-MK-CPA}}_{\Pi_{\PKE},\adv}}({\lambda})$, then $\adv$ outputs $b'=0$ with $1-\negl[\lambda]$ probability. 
\end{Claim}

\begin{proof}
    In the case $b=0$, $\adv$ samples a random message $M\leftarrow \{0,1\}^\lambda$ and ciphertexts are generated $\ct_i\gets \PKEnc(pk_i,M)$ for every $i\in [n]$. 

Consider an algorithm $\tilde{\adv}$ that is the same as $\adv$ except it runs ${\textsf{Exp}_{\adv, \textsf{NA}}^{\textnormal{Real}}}(1^\lambda)$ instead of ${\textsf{Exp}_{\adv, \textsf{NA}}^{\textnormal{Ideal}}}(1^\lambda)$ in Step 10 of the algorithm. We describe the algorithm of $\tilde{\adv}$ in \cref{fig:A}. 

\begin{figure}[!htb]
   \begin{center} 
   \begin{tabular}{|p{12cm}|}
    \hline 
\begin{center}
{\underline{Algorithm of $\tilde{\adv}$}}: 
\end{center}
\begin{enumerate}
    \item Steps 1-9 are the same as $\adv$. 
        \setcounter{enumi}{9} 
    \item $\tilde{\adv}$ runs the experiment $\textsf{Exp}_{({\adv}_1,{\adv}_2)}^{Real}(1^\lambda)$ as follows:
\begin{itemize}
       \item $\mk \leftarrow \Setup(1^\lambda)$.
        \item {$\adv_1$, using the public keys from Step 9 and the already sampled $M$ and $j$, queries the functional key of $C_i$,} where $C_i\coloneqq \PKEnc(pk_{i},\cdot)$ for each $i\in [n]$.
        \item Let $\fk_{C_i}\leftarrow \Keygen(\mk,C_i)$ denote the response for each $i\in [n]$.{ It outputs $M$ and retains $\textsf{st}$.}
        \item $\ct \leftarrow \Enc(\mk, M)$. 
        \item ${{\adv}_2(\ct,\textsf{st})}$ {sets $\ct'_{j}=\bot$ if} $\lvert \ct\rvert {\neq q(\lambda)}$. {Otherwise, it computes}  {${\ct}'_{j}\leftarrow \Dec(\fk_{C_{j}},\ct)$. It outputs $(\ct'_j,M,sk_j)$.} 
        \end{itemize}
    \item $\tilde{\adv}$ computes ${M}'\leftarrow \PKDec(sk_{j},{\ct}'_{j})$. 
    \item $\tilde{\adv}$ outputs $b'=0$ if $M'=M$ and $b'=1$ otherwise.
\end{enumerate}
\ \\ 
\hline
\end{tabular}
\caption{Modified algorithm of $\adv$ with Real experiment.}
    \label{fig:A}
    \end{center}
\end{figure}

It is easy to check that, {by the correctness of $\QFE$ and $\PKE$,} when $b=0$, the output of $\tilde{\adv}$ is 0, except with negligible probability. 

Given that 
$$ \{ {\Exp_{\adv, \textsf{NA}}^{\textnormal{Real}}}(1^\lambda)\}_{\lambda \in \N} \approx_c \{ {\Exp_{\adv, \textsf{NA}}^{\textnormal{Ideal}}}(1^\lambda)\}_{\lambda \in \N} $$
we must have that the output of $\tilde{\adv}$ is indistinguishable from $\adv$. Therefore, when $b=0$, $\adv$ outputs $b'=0$, except with negligible probability. 
\qed
\end{proof}
\begin{remark}
{Note that in the case $b=0$, the preceding claim and the ciphertext-format check imply that the ideal ciphertext }{{has the prescribed format, except with negligible probability}}. This means that {in the ideal experiment we must have $\lvert \ct\rvert = q(\lambda)$, except with negligible probability.} This point will be useful in showing {the following claim.}      
\end{remark}

\begin{Claim}
    If $b=1$ in ${\PKE^{\textnormal{IND-MK-CPA}}_{\Pi,\adv}}({\lambda})$, then $\adv$ outputs $b'=1$ with non-negligible probability. 
\end{Claim}

\begin{proof}


    
    Assume for contradiction that the probability that $\adv$ outputs 1 when $b=1$ is negligible. {This means that $\lvert \ct\rvert = q(\lambda)$ and $M'=M$, except with negligible probability.} {After Step 9, the tuples $(pk_i,sk_i,m_{1,i},\ct_i)$ are independently distributed as fresh PKE key pairs, uniform messages, and their encryptions; their joint distribution is independent of $j$. Here the full secret-key vector is used only for analysis: the reduction knows only $sk_j$. The simulator receives neither $j$ nor the private state containing $M=m_{1,j}$. Consequently, success at index $j$ equals the average recovery probability over all indices.} More formally,
    \begin{align}
        \Pr_{i\gets [n]}[m_{1,i}=\PKDec(sk_{i},\Dec(\fk_{C_i},\ct))]\geq 1-\negl[\lambda] \enspace.
    \end{align}
    
    {Each individual recovery procedure has failure probability at most $n$ times the average failure probability, which is negligible. For each $i\in[n]$ in order, we coherently compute the composite operation $\PKDec(sk_i,\Dec(\fk_{C_i},\ct))$, measure only its final classical output, and apply the inverse computation, retaining all work registers until this inversion. The gentle measurement lemma (see \cref{lem:genlemeas}) bounds the disturbance to the entire joint state, including $\ct$ and the functional keys. A hybrid argument over the polynomially many indices therefore recovers all messages except with negligible probability.} Such a recursive argument was also used {in the proof of \cref{thm:impM1AD}} and is described in more detail there. 

    Notice that the messages $m_1\coloneqq (m_{1,i})_{i\in [n]}$ are sampled uniformly at random and independently of the public keys $pk\coloneqq (pk_i)_{i\in [n]}$, functional keys $\fk_C\coloneqq (\fk_{C_i})_{i\in [n]}$, and secret keys $sk\coloneqq (sk_i)_{i\in [n]}$. {Before receiving the messages, prepare a joint advice state by running $\Setup$, generating all PKE key pairs, and obtaining all functional keys. Give $(\mk,pk)$ to $\Comp$ and $(\fk_C,sk)$ to $\Decomp$, preserving their joint distribution. This advice is independent of $m_1$. }This gives compression circuits as follows. 
    
The compression circuit $\Comp (m_1;{\mk,pk})$, with advice ${(\mk,pk)}$, computes $\ct_i\gets \PKEnc(pk_i, m_{1,i})$ for each $i\in [n]$. Then, it computes $\ct \leftarrow \Sim(1^\lambda, \mk, \mathcal{V})$, where $\mathcal{V} = (C_i, \ct_{i}, {1^\lambda})_{i\in [n]}$,  {and outputs the result $\ct$, replacing any output outside the prescribed format by a fixed $q(\lambda)$-qubit state. Thus, the size bound holds on every input, and this replacement affects the recovery probability by at most a negligible amount.} Meanwhile, the decompression circuit $\Decomp(\ct; \fk_C, sk)$, with advice $(\fk_C, sk)$, {uses the coherent computation, measurement, and inversion procedure above to obtain, for every $i\in[n]$,}
\[
    \tilde{m}_{1,i}\gets \PKDec(sk_{i},\Dec(\fk_{C_i},\ct)) \enspace.
\]
Then, it outputs $\tilde{m}_{1}\coloneqq (\tilde{m}_{1,i})_{i\in [n]}$. 

We already deduced that $\Pr[\tilde{m}_{1}=m_{1}]\geq 1-\negl[\lambda]$ and given that $\lvert \ct\rvert {= q(\lambda)< \frac{n\lambda}{2}}$, this means that $(\Comp,\Decomp)$ are {$(n\lambda,n\lambda-q(\lambda))$-compression circuits, since $m_1$ consists of $n\lambda$ uniformly random bits.} But this contradicts \cref{lem:incompress}, so we conclude that $\adv$ outputs $b'=1$ when $b=1$ with non-negligible probability. 
    \qed
\end{proof}

{All in all, by the previous two claims, for a suitable polynomial and infinitely many $\lambda$, we have} 
 \begin{align*}
         \Pr[\PKE^{\textnormal{IND-MK-CPA}}_{\Pi_{\PKE},\adv}({\lambda})=1] > \frac{1}{2}+\frac{1}{\poly[\lambda]} \enspace.
    \end{align*}
which contradicts \textsf{MK-CPA} security of ${\Pi}_{\PKE}$. Therefore, there does not exist a $1\text{-}\textsf{M}\text{-}\textsf{NA}$
\textsf{SIM}--secure $\QFE$ scheme. 
\qed
\end{proof}

\begin{remark}
\cref{def:simsecurityna} for non-adaptive simulation security allows the simulator in the ideal experiment to control only the response to the ciphertext queries. However, we can consider a more general and weaker security notion that also allows the simulator to control the response to the non-adaptive functional key queries. The proof of \cref{thm:imp1MNA} easily generalizes to show impossibility of this weaker notion of simulation security. Critically, it is not necessary for the functional keys to be honestly generated in the  proof.{ It suffices that the simulated functional keys are issued before the simulator receives the evaluations $C_i(M)$. These keys and the simulator\textquoteright s retained state then form a joint state independent of the messages in the compression experiment. We give the keys to $\Decomp$ and the retained state to $\Comp$ as shared advice, and apply the entanglement-assisted bound as before.} 
\end{remark}

A direct corollary to our result is that $1\text{-}\textsf{M}\text{-}\textsf{NA}$ \textsf{SIM}--secure public-key $\QFE$ (with classical public-keys) is \emph{unconditionally} impossible, given that this primitive implies $\PKE$.

\begin{corollary}
There does not exist a $1\text{-}\textsf{M}\text{-}\textsf{NA}$ \textsf{SIM}--secure public-key $\QFE$ with classical public keys. 
\end{corollary}

%% file: section7_reviewed.tex
\providecommand{\QFEreplace}[2]{}
\renewcommand{\QFEreplace}[2]{%
  \begingroup
  \let\QFEoriginalcref\cref
  \let\QFEoriginalcite\cite
  \renewcommand{\cref}[1]{\mbox{\QFEoriginalcref{##1}}}%
  \renewcommand{\cite}[1]{\mbox{\QFEoriginalcite{##1}}}%
  \sout{#1}\endgroup\ {\color{red}#2}}
\providecommand{\QFEinsert}[1]{{\color{red}\relax#1}}
\providecommand{\QFEmathreplace}[2]{\text{\sout{\ensuremath{#1}}}\;{\color{red}#2}}

\section{Functional Encryption with Disappearing Functional Keys}
\label{sec:CBQS-FE}

We now turn to the bounded storage models. In this setting, we define and build a ${1\text{-}\textsf{M}\text{-}\textsf{AD}}$ \textsf{SIM}--secure $\FE$ scheme with disappearing functional keys. Specifically, we build a \textsf{SIM}--secure \textsf{FE} scheme that is secure under many adaptive functional key queries and a single challenge ciphertext. The functional keys in our scheme are disappearing, which means that {they cannot be used to evaluate ciphertexts after they are received.}



\subsection{Definitions}
\label{sec:sec1}
We first define \textsf{FE} with disappearing functional keys.{ We use the same memory bound $s=s(\lambda)$ and streaming convention as for disappearing $\GB$. The bound applies to quantum memory in the $\textsf{BQSM}$ and classical memory in the $\textsf{BCSM}$. } 

\begin{definition}[\textsf{FE} with Disappearing Functional Keys]
    Let $\lambda $ be the security parameter and $\mathcal{C}_\lambda$ be a class of circuits with input space $\mathcal{X}_\lambda$ and output space $\mathcal{Y}_\lambda$. A \emph{functional encryption} scheme on the class $\mathcal{C}_\lambda$ consists of the following algorithms: 
\begin{itemize}
    \item $\textsf{Setup}(1^\lambda):$ Outputs a classical master key $\mk$. 
   \item $\textsf{PkSend}(\mk):$ Sends a public key transmission ${\pk}$ (may be quantum).
    \item $\textsf{PkReceive}({\pk}):$ {Outputs a one-time classical encryption key $\ek$ from the public key $\pk$.}
    \item $\textsf{KeyGen}(\mk,C):$ Takes as input the master key $\mk$ and circuit $C\in \mathcal{C}_\lambda$ and outputs a functional key stream ${\fk}_C$.
    \item $\textsf{Enc}(\ek,\mu):$ Takes as input an encryption key $\ek$ and a message $\mu \in \mathcal{X}_\lambda$ and outputs a ciphertext $\ct$.
    \item $\textsf{Dec}({\fk}_C, \ct)$: {Takes a functional key ${\fk}_C$ and a ciphertext $\ct$ and outputs a value $y\in \mathcal{Y}_\lambda$.}
\end{itemize}

\emph{Correctness} requires that for any circuit $C\in  \mathcal{C}_\lambda$ and message $\mu \in \mathcal{X}_\lambda$,
\begin{align*} \Pr{\left[
\begin{tabular}{c|c}
 \multirow{5}{*}{$\textsf{Dec}({\fk}_C,\ct)=C(\mu)\ $} &   $\mk\ \leftarrow \textsf{Setup}(1^\lambda)$\\
 & ${\pk} \leftarrow \textsf{PkSend}(\mk)$ \\ 
    & $\ek \leftarrow \textsf{PkReceive}({\pk})$ \\ 
 & ${\ct\ \leftarrow \textsf{Enc}(\ek,\mu)}$ \\ 
 & ${{\fk}_C \leftarrow \textsf{KeyGen}(\mk,C)}$\\
 \end{tabular}\right]} \geq 1-\negl[\lambda] .
\end{align*}
\end{definition}

We define adaptive simulation-based security through the use of two experiments: \textsf{Real Experiment} and \textsf{Ideal Experiment}. We require these experiments to be indistinguishable. 

Our security model closely follows the standard definition of adaptive simulation-based security \cite{BSW11}, with one crucial distinction: the simulator is granted access to the evaluation of a message under function $f$ only if the adversary requests the functional key $\textsf{fk}_f$ \emph{after} the challenge ciphertext has been issued. This restriction is designed to capture the core property of \textbf{disappearing functional keys}. Specifically, it implies that an adversary cannot utilize functional keys obtained prior to the challenge phase to decrypt or analyze the challenge ciphertext itself.

We now present the two experiments where the simulator consists of a tuple of algorithms $\Sim=(\textsf{Setup}^*,\textsf{PkSend}^*,\textsf{EncKey}^*,\textsf{KeyGen}^*_1,\textsf{Enc}^*,\textsf{KeyGen}^*_2)$ and adversary $\adv=(\adv_1,\adv_2)$. In both experiments, every functional-key stream requested by
$\mathcal{A}_1$ must finish, including all prescribed storage
checkpoints, before the challenge ciphertext is issued.
Streams may overlap within either phase, subject to the joint
memory bound $\s$.

\smallskip \noindent\fbox{%
    \parbox{\textwidth}{%
\textbf{Real Experiment} $\textsf{DF-FE}^{\textsf{Real}}_{\Pi,\adv,\Sim}({\lambda})$:
\begin{enumerate}
    \item Sample $\mk \leftarrow \textsf{Setup}(1^\lambda)$.
      \item Generate $\pk\leftarrow \textsf{PkSend}(\mk).$
       \item Generate $\ek\leftarrow \textsf{PkReceive}(\pk)$.
   \item $(m,\tau)\leftarrow \adv_1^{\textsf{KeyGen}(\mk,\cdot),\textsf{PkSend}(\mk)}$.
       \item Compute $\ct\leftarrow \textsf{Enc}(\ek,m)$.
 \item Output $\alpha\leftarrow \adv_2^{\textsf{KeyGen}(\mk,\cdot),\textsf{PkSend}(\mk)}(\ct,\tau)$.
    \end{enumerate}}}
    \smallskip  

In the experiment above and the rest, $\textsf{PkSend}(\mk)$ is the oracle that outputs a public key copy $\pk$ every time it is queried. This is relevant in the bounded memory models where a public key transmission may not be storable, so the adversary  may gain more information from multiple copies. It is also relevant in the quantum setting, {where the public key may be a quantum state} and the adversary may learn more from multiple copies. 

\smallskip \noindent\fbox{%
    \parbox{\textwidth}{%
\textbf{Ideal Experiment} $\textsf{DF-FE}^{\textsf{Ideal}}_{\Pi,\adv,\Sim}({\lambda})$:
\begin{enumerate}
    \item Sample $\mk^* \leftarrow \textsf{Setup}^*(1^\lambda)$.
       \item Generate $\ek^*\leftarrow \textsf{EncKey}^*(\mk^*)$.
   \item $(m,\tau)\leftarrow \adv_1^{\textsf{KeyGen}_1^*(\mk^*,\ek^*,\cdot),\textsf{PkSend}^*(\mk^*)}$.
   \item $\ct^*\leftarrow \textsf{Enc}^*(\ek^*,1^{\lvert m\rvert}).$
 \item {$\alpha\leftarrow \adv_2^{\textsf{KeyGen}^*_2(\mk^*,\ct^*, \mathcal{V},\cdot),\textsf{PkSend}^*(\mk^*)}(\ct^*,\tau)$. Let $\mathcal{Q}=(C_i)_{i\in [p]}$ be the queries made to $\textsf{KeyGen}^*_2$ by $\adv_2$ and let $\mathcal{V}=(C_i,y_i=C_i(m))_{i\in [p]}$.} 
    \end{enumerate}}}
    \smallskip

\begin{definition}[$\SIM$ Security]
    A $\FE$ scheme $\Pi$ satisfies \emph{$\SIM$ security} if there exists a QPT simulator $\Sim$ {such that for any QPT adversary $\adv$ with memory bound $\s$, the following distributions} are computationally indistinguishable:
    \begin{align*}
         \textsf{DF-FE}^{\textsf{Real}}_{\Pi,\adv,\Sim}({\lambda}) \approx_c \textsf{DF-FE}^{\textsf{Ideal}}_{\Pi,\adv,\Sim}({\lambda}).
    \end{align*}
\end{definition}

Note that the definition considers QPT algorithms. {We instead consider PPT algorithms satisfying the classical memory bound when discussing security in the \textsf{BCSM}.} 

\subsection{Construction}

\begin{construct}[\textsf{FE} with Disappearing Functional Keys]
\label{con:fe}
{\small Let $\mathcal{C}_\lambda$ be a class of deterministic classical polynomial circuits. {Let $\mathcal{O}$ be a disappearing $\GB$ over classical polynomial circuits, secure against adversaries with memory bound $s$.} Let $(\textsf{SKE.KeyGen},\textsf{SKE.Enc},\textsf{SKE.Dec})$ be the algorithms for a one-time information-theoretic authenticated symmetric encryption scheme on $n$-bit messages, where $\textsf{SKE.KeyGen}(1^\lambda)$ outputs a random $\ell$-bit string for some $\ell \in \poly[\lambda]$.{ Authentication means one-time ciphertext integrity: even given an encryption of a chosen message, an unbounded adversary cannot produce a different valid ciphertext except with negligible probability.} Let $\textsf{PRG}:\{0,1\}^{\lambda}\rightarrow \{0,1\}^{2\lambda}$ be a (post-quantum) pseudorandom generator and let $(\textsf{PPRF.KeyGen},\textsf{PPRF.Puncture},F)$ be the algorithms of a (post-quantum) \textsf{PPRF} family mapping $2\lambda$-bits to $\ell$-bits. The construction for \textsf{FE} on $n$-bit messages over $\mathcal{C}$ is as follows:
\begin{itemize}
    \item $\textsf{Setup}(1^\lambda):$ Sample $\nk \leftarrow \textsf{PPRF.KeyGen}(1^\lambda)$ and let this be the master key $\mk\coloneqq \nk$.

    \item $\textsf{PkSend}(\mk):$ Define the function $P_{\mk}$ as follows:
    \begin{enumerate}
    \item Takes as input $x$.
        \item Computes $t\coloneqq \textsf{PRG}(x)$.
        \item Computes $y \coloneqq  F(\nk, t).$
        \item Outputs $(t,y)$.
    \end{enumerate}
    Output $\mathcal{O}(P_\mk).$

    \item $\textsf{PkReceive}(\mathcal{O}(P_\mk)):$
    \begin{enumerate}
        \item Sample $r\leftarrow \{0,1\}^\lambda$.
        \item Evaluate $\mathcal{O}(P_\mk)$ on $r$ and let $(t,y)$ denote the result. 
        \item Output $\ek\coloneqq (t,y)$.
    \end{enumerate}

\item $\textsf{KeyGen}(\mk,C):$ Define the program $F_{\mk,C}$ as follows: 
\begin{enumerate}
    \item Takes as input $(t, \textsf{e})$. 
    \item Computes $y = F(\nk, t).$
    \item Compute $m'=\textsf{SKE.Dec}(y, \textsf{e})$
    \item Outputs $C(m')$.
\end{enumerate}
Output $\fk_C\coloneqq \mathcal{O}(F_{\mk,C})$.

\item $\textsf{Enc}({\ek},m):$ 
\begin{enumerate}
    \item Interpret $\ek$ as $(t,y)$.
    \item Compute $\textsf{e}\leftarrow \textsf{SKE.Enc}(y,m).$
    \item Output $\ct\coloneqq (t,\textsf{e})$.
\end{enumerate}

\item $\textsf{Dec}({\fk}_C,\ct):$ 
Evaluate the obfuscation $\fk_C$ on $\ct$ and output the result.

\end{itemize}}\end{construct}

All public-key programs in the construction and proof are padded
to a common public size. For each circuit $C$, all functional-key
programs are likewise padded to a common public size depending
only on $\lambda$, $n$, and the input, output, and circuit sizes
of $C$. These bounds accommodate every hybrid below, and
corresponding programs use the same input and output formats.
Throughout the construction and proof, we set $C(\perp)=\perp$
and reject malformed inputs. In particular, failed symmetric
decryption always produces $\perp$.

\begin{theorem}
    Construction \ref{con:fe} is an $\SIM$ secure \textsf{FE} scheme over $\mathcal{C}$ {assuming the existence of disappearing $\GB$ over classical polynomial circuits and (post-quantum) \textsf{PPRF}s.}
\end{theorem}

\begin{proof}
To show security, we need to construct a QPT simulator $\Sim$ such that for any QPT $\adv$,
    \begin{align*}
         \textsf{DF-FE}^{\textsf{Real}}_{\Pi,\adv,\Sim}({\lambda}) \approx_c \textsf{DF-FE}^{\textsf{Ideal}}_{\Pi,\adv,\Sim}({\lambda}).
    \end{align*}

    We prove security through a sequence of indistinguishable hybrids, starting with the $\textsf{Real Experiment}$ and ending with the $\textsf{Ideal Experiment}$.

\begin{itemize}
    \item \textbf{Hybrid} $\Hy_0:$ {This is the $\textsf{DF-FE}^{{\textsf{Real}}}_{\Pi,\adv,\Sim}({\lambda})$ experiment.}

\smallskip \noindent\fbox{%
    \parbox{\textwidth}{%
\textbf{Experiment} $\textsf{DF-FE}^{\textsf{Real}}_{\Pi,\adv,\Sim}({\lambda})$:
\begin{enumerate}
    \item Sample $\nk=\mk\ \leftarrow \textsf{Setup}(1^\lambda)$.
       \item Extract an encryption key as follows:
   \begin{enumerate}
       \item Sample random $r^*\leftarrow \{0,1\}^\lambda.$
       \item Compute $t^*= \textsf{PRG}(r^*)$.
       \item Compute $y^*= F(\nk, t^*).$
       \item Set $\ek^*=(t^*,y^*)$.
       \end{enumerate}
   \item $(m,\tau)\leftarrow \adv_1^{\textsf{KeyGen}(\mk,\cdot),\textsf{PkSend}(\mk)}$. 
   \item Generate ciphertext $\ct^*=(t^*,\textsf{e}^*)\leftarrow \textsf{Enc}(\ek^*,m)$.
 \item Run $\alpha \leftarrow \adv_2^{\textsf{KeyGen}(\mk,\cdot),\textsf{PkSend}(\mk)}(\ct^*,\tau)$. Let $\mathcal{Q}=(C_j)_{j\in [p]}$ be the set of queries to the $\textsf{KeyGen}$ oracle by $\adv_2$ and let $\mathcal{V}\coloneqq (C_j,y_j=C_j(m))_{j\in [p]}$. 
    \end{enumerate}}}
    \smallskip 

    \item \textbf{Hybrid} $\Hy_1:$ Same as hybrid $\Hy_0$, except the experiment generates the encryption key $\ek^*$ by sampling $t^*$ randomly.

    \item \textbf{Hybrid} $\Hy_2:$ Same as $\Hy_1$ except with the following modification. The experiment computes $\nk^*\leftarrow \textsf{PPRF.Puncture}(\nk, t^*)$ and replaces $\textsf{PkSend}(\nk)$ with $\textsf{PkSend}(\nk^*)$ which produces $\mathcal{O}(P_{\nk^*})$.
    
   \item \textbf{Hybrid} $\Hy_{3}$: Same as $\Hy_{2}$ except both $\adv_1$ and $\adv_2$'s oracle access to $\textsf{KeyGen}(\nk,\cdot)$ is replaced with oracle access to $\textsf{KeyGen}'(\nk^*,\ek^*,\cdot)$, which on input $C$, outputs $\mathcal{O}(F'_{\nk^*,\ek^*,C})$ (see below). 

    \smallskip \noindent\fbox{%
    \parbox{\textwidth}{%
    $F'_{\nk^*,\ek^*,C}(t, \textsf{e}):$
\begin{enumerate}
    \item If $t=t^*$, output $C(\textsf{SKE.Dec}(y^*, \textsf{e}))$.
    \item Otherwise, compute $y \leftarrow F(\nk^*, t).$
    \item Output $C(\textsf{SKE.Dec}(y, \textsf{e}))$.
\end{enumerate}
}}
    \smallskip

    \item \textbf{Hybrid} $\Hy_4:$ We modify $\Hy_{3}$ to instead sample ${y}^*$ randomly. 

    \item \textbf{Hybrid} $\Hy_5:$  Same as $\Hy_{4}$ except $\adv_1$'s oracle access to $\textsf{KeyGen}'$ is replaced with oracle access to $\textsf{KeyGen}^1(\nk^*,{t^*},\cdot)$, which on input $C$, outputs $\mathcal{O}({F^1_{\nk^*,t^*,C}})$ (see below). 

    \smallskip \noindent\fbox{%
    \parbox{\textwidth}{%
    ${F^1_{\nk^*,t^*,C}}(t, \textsf{e}):$
\begin{enumerate}
    \item If $t=t^*$, output $\bot$.
    \item Otherwise, compute $y \leftarrow F(\nk^*, t).$
    \item Output $C(\textsf{SKE.Dec}(y, \textsf{e}))$.
\end{enumerate}
}}
    \smallskip

        \item \textbf{Hybrid} $\Hy_{6}$: Same as $\Hy_{5}$ except $\adv_2$'s oracle access to $\textsf{KeyGen}'$ is replaced with oracle access to $\textsf{KeyGen}^2(\nk^*,\ct^*,\mathcal{V},\cdot)$, which on input $C_j\in\mathcal{Q}$, outputs $\mathcal{O}({F^2_{\nk^*,\ct^*,C_j,y_j}})$ (see below). 

    \smallskip \noindent\fbox{%
    \parbox{\textwidth}{%
    ${F^2_{\nk^*,\ct^*,C_j,y_j}}(t, \textsf{e}):$
\begin{enumerate}
    \item If $t=t^*$ and $\textsf{e}=\textsf{e}^*$, then output $y_j$.
    \item If $t=t^*$ and $\textsf{e}\neq \textsf{e}^*$, then output $\perp$.
    \item Otherwise, compute $y \leftarrow F(\nk^*, t).$
    \item Output $C_j(\textsf{SKE.Dec}(y, \textsf{e}))$.
\end{enumerate}
}}
    \smallskip 
    
    \item \textbf{Hybrid} $\Hy_7:$ We modify $\Hy_{6}$ to let $\ct^*$ be an encryption of $1^n$, i.e. $\ct^*\leftarrow \textsf{Enc}(\ek^*,1^n)$.

    \item \textbf{Hybrid} $\Hy_8:$ We swap some steps in $\Hy_{7}$. Specifically, we sample $t^*$ as the first step and generate the key $\nk^*$ as the third step. The final experiment is as follows:
    
    \smallskip \noindent\fbox{%
    \parbox{\textwidth}{%
 $\Hy_8$:
\begin{enumerate}
    \item Sample $t^*\leftarrow \{0,1\}^{2\lambda}$.
        \item Sample $\mk\ \leftarrow \textsf{Setup}(1^\lambda)$.
        \item Generate $\nk^*\leftarrow \textsf{PPRF.Puncture}(\mk,t^*)$.
    \item Sample $y^*\leftarrow \{0,1\}^\ell$.
    \item Set $\ek^*=(t^*,y^*).$
   \item $(m,\tau)\leftarrow \adv_1^{{\textsf{KeyGen}^1(\nk^*,t^*,\cdot)},\textsf{PkSend}(\nk^*)}$. 
      \item Generate ciphertext $\ct^*=(t^*,\textsf{e}^*)\leftarrow \textsf{Enc}(\ek^*,1^n)$.
 \item Run $\alpha \leftarrow {\adv_2}^{\textsf{KeyGen}^2(\nk^*,\ct^*,\mathcal{V},\cdot),\textsf{PkSend}(\nk^*)}(\ct^*,\tau)$. 
    \end{enumerate}}}
    \smallskip     
\end{itemize}

We are now ready to define the algorithms of the simulator $\mathcal{S}$:

    \smallskip \noindent\fbox{%
    \parbox{\textwidth}{%
    \text{Algorithms of $\mathcal{S}$:}
\begin{itemize}
\item $\textsf{Setup}^*(1^\lambda):$ 
\begin{enumerate}
    \item Sample $t^*\leftarrow \{0,1\}^{2\lambda}$.
    \item Sample $\nk\leftarrow \textsf{PPRF.KeyGen}(1^\lambda)$.
    \item Generate $\nk^*\leftarrow \textsf{PPRF.Puncture}(\nk,t^*)$.
    \item Output $\mk^*\coloneqq (t^*,\nk^*)$.
\end{enumerate} 
\item $\textsf{EncKey}^*(\mk^*):$ 
\begin{enumerate}
    \item Sample $y^*\leftarrow \{0,1\}^\ell$.
    \item Output $\ek^*=(t^*,y^*).$
    \end{enumerate}
    \item $\textsf{KeyGen}^*_1(\mk^*,\ek^*,\cdot):$ Output ${\textsf{KeyGen}^1(\nk^*,t^*,\cdot)}$.
    \item $\textsf{PkSend}^*(\mk^*):$ Output $\textsf{PkSend}(\nk^*)$.
    \item $\textsf{Enc}^*(\ek^*,1^n):$ Output $\ct^*\leftarrow \textsf{Enc}(\ek^*,1^n)$. 
    \item $\textsf{KeyGen}^*_2(\mk^*,\ct^*,\mathcal{V},\cdot):$ Output $\textsf{KeyGen}^2(\nk^*,\ct^*,\mathcal{V},\cdot)$.
\end{itemize}}}
    \smallskip     

With this definition, it is easy to check that $\Hy_8$ is the same as $\textsf{DF-FE}^{{\textsf{Ideal}}}_{\Pi,\adv,\Sim}({\lambda})$. Hence, it is sufficient to show that hybrids $\Hy_0$ is indistinguishable from $\Hy_8$ to prove security.  

We show this by proving that no QPT adversary can distinguish between any two consecutive hybrids and then apply the triangle inequality.{ In every hybrid, $\mathcal{V}$ is initialized as empty and is extended with $(C_j,C_j(m))$ only when the corresponding post-challenge query arrives. Each experiment outputs $\alpha$.} 

\begin{Claim}
No QPT adversary can distinguish between hybrids $\Hy_0$ and $\Hy_1$ with non-negligible probability. 
\end{Claim}
\begin{proof}
The only difference between these hybrids is that in hybrid $\Hy_0$, the experiment samples $r^*$ randomly and sets $t^*=\textsf{PRG}(r^*)$, while in hybrid $\Hy_1$, $t^*$ is sampled randomly. Note that $r^*$ is sampled randomly and never revealed. Hence, if a QPT adversary can distinguish between these hybrids, then this can easily be converted to a QPT distinguisher that breaks the security of the pseudorandom generator $\textsf{PRG}$.
    \qed
\end{proof}

\begin{Claim}
No QPT adversary can distinguish between hybrids $\Hy_1$ and $\Hy_2$ with non-negligible probability. 
\end{Claim}
\begin{proof}
The only difference between these hybrids is that queries to the public key are answered with $\mathcal{O}(P_{\nk^*})$ instead of $\mathcal{O}(P_\nk).$ 

Note that if $t^*$ is not in the image of $\textsf{PRG}$, then $P_{\nk}\equiv P_{\nk^*}$. This occurs with at least $1-2^{-\lambda}$ probability, since the image size of $\textsf{PRG}$ is at most $2^\lambda$, while $t^*$ is sampled randomly from a set of size $2^{2\lambda}$. 

Assume there exists a QPT algorithm $\adv$ that can distinguish between these hybrids with non-negligible probability. 
Since $P_{\nk}\equiv P_{\nk^*}$ occurs with high probability, $\adv$ must be able to distinguish between these hybrids in this case as well. {By the security of $\GB$, $\adv$ can be simulated with a simulator given polynomial queries to $P_{\nk}$ or to $ P_{\nk^*}$ depending on the hybrid.} If $\adv$ can distinguish these hybrids, then the simulator can distinguish these hybrids, {but all corresponding program oracles are identical on this event, yielding a contradiction.}  
    \qed
\end{proof}

\begin{Claim}
No QPT adversary can distinguish between hybrids $\Hy_{3}$ and $\Hy_{2}$ with non-negligible probability.
\end{Claim}
\begin{proof}
The only difference between $\Hy_{3}$ and $\Hy_{2}$ is that when $\adv$ queries the oracle on $C$, it receives $\mathcal{O}(F'_{\nk^*,\ek^*,C})$ instead of $\mathcal{O}(F_{\mk,C})$. Note that these functions have the same functionality, i.e. $F'_{\nk^*,\ek^*,C}\equiv F_{\mk,C}$. By the same arguments as in the previous claim, no QPT adversary can distinguish these hybrids with non-negligible probability. 
    \qed
\end{proof}

\begin{Claim}
No QPT adversary can distinguish between hybrids $\Hy_{3}$ and $\Hy_{4}$ with non-negligible probability. 
\end{Claim}
\begin{proof}
Assume a QPT adversary $\adv$ can distinguish between these two hybrids with non-negligible probability, i.e. $\adv$ outputs a guess $b\in \{0,1\}$ of whether it is in hybrid $\Hy_{3}$ or $\Hy_{4}$. Then, $\adv$ can be used to construct a QPT algorithm $\mathcal{B}$ that breaks \textsf{PPRF} security as follows. 

In the \textsf{PPRF} security definition, first a key is sampled $K\leftarrow \textsf{PPRF.KeyGen}(1^\lambda)$. $\mathcal{B}$ samples a random string $s\leftarrow \{0,1\}^{2\lambda}$ and sets the puncturing set $S=\{s\}$. Then, $Y$ is set to either the evaluation $F(K,s)$ or a uniform random string. $\mathcal{B}$ is given $(K_S,Y)$ where $K_S=\textsf{PPRF.Puncture}(K,s)$ and must distinguish whether $Y$ is a random string or $F(K,s)$. 

 \smallskip \noindent\fbox{%
    \parbox{\textwidth}{%
    $\mathcal{B}$ commences as follows:
\begin{enumerate}
\item Set $\nk^*=K_S$.
       \item Set the encryption key as follows:
   \begin{enumerate}
    \item Set $t^*=s$.
    \item Set $y^*=Y$.
    \item Set $\ek^*=(t^*,y^*).$
    \end{enumerate}
   \item $(m,\tau)\leftarrow \adv_1^{\textsf{KeyGen}'(\nk^*,\ek^*,\cdot),\textsf{PkSend}(\nk^*)}$.
      \item Generate ciphertext $\ct^*=(t^*,\textsf{e}^*)\leftarrow \textsf{Enc}(\ek^*,m)$.
 \item Run $b\leftarrow \adv_2^{\textsf{KeyGen}'(\nk^*,\ek^*,\cdot),\textsf{PkSend}(\nk^*)}(\ct^*,\tau)$. 
 \item $\mathcal{B}$ outputs $b$.
    \end{enumerate}}}
    \smallskip     

If $Y=F(K,s)$, then this experiment is the same as $\Hy_{3}$, while if $Y$ is sampled randomly, then this experiment is the same as $\Hy_4$. Since $\adv$ can distinguish the two hybrids with non-negligible probability, $\mathcal{B}$ can distinguish whether $Y=F(K,s)$ or is sampled randomly with non-negligible probability which contradicts the security of \textsf{PPRF}s.  
    \qed
\end{proof}

\begin{Claim}
No QPT adversary can distinguish between hybrids $\Hy_{5}$ and ${\Hy_{6}}$ with non-negligible probability. 
\end{Claim}

\begin{proof}
{The only difference between $\Hy_{5}$ and $\Hy_{6}$ is that $\adv_2$'s query $C_j$ is answered with $\mathcal{O}(F^2_{\nk^*,\ct^*,C_j,y_j})$ instead of $\mathcal{O}(F'_{\nk^*,\ek^*,C_j})$.}

{By $\GB$ security, the streams can be simulated using polynomially many classical oracle queries. The pre-challenge programs do not use $y^*$, and all other oracle answers, apart from decryption at $t^*$, are independent of $y^*$.} 

{By correctness of the symmetric encryption scheme, the two programs agree at $(t^*,\textsf{e}^*)$, except with negligible probability. Otherwise, they can differ only on an input $(t^*,\textsf{e})$ satisfying $\textsf{e}\neq\textsf{e}^*$ and $\textsf{SKE.Dec}(y^*,\textsf{e})\neq\bot$.} 

{Note that $y^*$ is sampled uniformly at random in these hybrids. Hence, if the simulator can guess such an input, then it is easy to show that we can break the information-theoretic unforgeability of the symmetric encryption scheme giving a contradiction.}

\textit{(Note that for this argument, we needed $\GB$ rather than the weaker notion of \textsf{iO}.)}
    \qed
\end{proof}

\begin{Claim}
No QPT adversary can distinguish between hybrids $\Hy_{5}$ and ${\Hy_{4}}$ with non-negligible probability.
\end{Claim}
\begin{proof}
{The only difference between $\Hy_{5}$ and $\Hy_{4}$ is that when $\adv_1$ queries the oracle on $C$, it receives $\mathcal{O}(F^1_{\nk^*,t^*,C})$ instead of $\mathcal{O}(F'_{\nk^*,\ek^*,C})$. These functions only disagree on inputs starting with $t^*$.} 

{By disappearing $\GB$ security, we can simulate all streams using at most $Q=\poly[\lambda]$ classical oracle queries before the challenge. Except with probability $2^{-\lambda}$, $t^*$ lies outside the image of $\textsf{PRG}$, so the public-key programs agree with $P_{\nk}$. Before a query whose first component is $t^*$, every oracle answer can be computed using $\nk$ alone and is independent of $t^*$. The probability of such a query is therefore at most $Q/(2^{2\lambda}-2^\lambda)$. Once the pre-challenge streams end, the simulator has no further oracle access to their programs. Revealing $\ct^*$ later thus leaves the two simulations indistinguishable, except with negligible probability. Disappearing $\GB$ security proves the claim.}
    \qed
\end{proof}

\begin{Claim}
No QPT adversary can distinguish between hybrids $\Hy_{6}$ and $\Hy_{7}$ with non-negligible probability. 
\end{Claim}
\begin{proof}
{The only difference between these hybrids is that $\ct^*$ is an encryption of $m$ in $\Hy_{6}$ and an encryption of $1^n$ in hybrid $\Hy_{7}$. Assume a QPT adversary $\adv$ can distinguish these hybrids with non-negligible probability. We construct a QPT algorithm $\mathcal{D}$ that breaks one-time indistinguishability of the symmetric encryption scheme as follows.} 

{In the one-time indistinguishability game, the challenger samples a secret key $K$, which is not given to $\mathcal{D}$. The reduction samples $t^*$ and $\nk$, computes $\nk^*=\textsf{PPRF.Puncture}(\nk,t^*)$, and runs $\adv_1$ with oracles $\textsf{KeyGen}^1(\nk^*,t^*,\cdot)$ and $\textsf{PkSend}(\nk^*)$ to obtain $(m,\tau)$. These oracles do not use $y^*$. It submits $(m_0=m,m_1=1^n)$ and receives the symmetric ciphertext $\textsf{e}^*$. It sets $\ct^*=(t^*,\textsf{e}^*)$ and runs $\adv_2$ with oracles $\textsf{KeyGen}^2(\nk^*,\ct^*,\mathcal{V},\cdot)$ and $\textsf{PkSend}(\nk^*)$, computing $C_j(m)$ as each adaptive query $C_j$ arrives. Finally, $\mathcal{D}$ outputs the bit returned by $\adv_2$.} 

{If $\textsf{e}^*$ encrypts $m_0$, this is hybrid $\Hy_6$; if it encrypts $m_1$, this is hybrid $\Hy_7$. Thus, any non-negligible distinguishing advantage of $\adv$ contradicts one-time indistinguishability of the symmetric encryption scheme.}  
    \qed
\end{proof}

\begin{Claim}
No QPT adversary can distinguish between hybrids $\Hy_{7}$ and $\Hy_{8}$ with non-negligible probability. 
\end{Claim}
\begin{proof}
    This is clear as moving the steps for sampling $t^*$ and $\nk^*$ to the start of the experiment has no effect on the outcome. 
    \qed
\end{proof}

All in all, we have shown, by the triangle inequality, that no QPT adversary can distinguish between hybrid $\Hy_0$, which is $\textsf{DF-FE}^{{\textsf{Real}}}_{\Pi,\adv,\Sim}({\lambda})$, and hybrid $\Hy_8$, which is $\textsf{DF-FE}^{{\textsf{Ideal}}}_{\Pi,\adv,\Sim}({\lambda})$. Therefore,
   \begin{align*}
         \textsf{DF-FE}^{\textsf{Real}}_{\Pi,\adv,\Sim}({\lambda}) \approx_c \textsf{DF-FE}^{\textsf{Ideal}}_{\Pi,\adv,\Sim}({\lambda}).
    \end{align*}
\qed    
\end{proof}

%% file: section8_audit_reviewed.tex
\providecommand{\QFEreplace}[2]{}
\renewcommand{\QFEreplace}[2]{%
  \begingroup
  \let\QFEoriginalcref\cref
  \let\QFEoriginalcite\cite
  \renewcommand{\cref}[1]{\mbox{\QFEoriginalcref{##1}}}%
  \renewcommand{\cite}[1]{\mbox{\QFEoriginalcite{##1}}}%
  \sout{#1}\endgroup\ {\color{red}#2}}
\providecommand{\QFEinsert}[1]{{\color{red}\relax#1}}
\providecommand{\QFEmathreplace}[2]{\text{\sout{\ensuremath{#1}}}\;{\color{red}#2}}

\section{Functional Encryption with Disappearing Ciphertexts}

We now build a {{$\textsf{DC-SIM}$--secure}} symmetric-key \textsf{FE} with disappearing ciphertexts for polynomial circuits in the bounded storage models. In other words, we build a scheme where the adversary in the security experiment can request only a single non-adaptive functional key and can request {an arbitrary polynomial number} of challenge ciphertexts. 

{Notably, after all challenge ciphertext transmissions have ended, the adversary receives the functional key of the identity function. This grants full decryption capability for subsequent ciphertext streams, while security requires that it reveal no additional information about the previously transmitted messages. The experiment below allows one non-adaptive functional key query followed by this identity-key disclosure.}

{{Note that it is easy to generalize our scheme to  allow for any (polynomial) bounded number of functional key queries, although our ciphertext size scales with this polynomial bound.}}

\subsection{Definitions}
\label{sec:sec2}
We define a symmetric $\FE$ scheme with disappearing ciphertexts.{ We use the same memory bound $s=s(\lambda)$ and streaming convention as for disappearing $\GB$. The bound applies to quantum memory in the $\textsf{BQSM}$ and classical memory in the $\textsf{BCSM}$. } 

\begin{definition}[\textsf{FE} with Disappearing Ciphertexts]
    Let $\lambda $ be the security parameter and $\mathcal{C}_\lambda$ {be a class of circuits with input space $\mathcal{X}_\lambda$ and output space $\mathcal{Y}_\lambda$, containing the identity circuit $I$.} A \emph{functional encryption scheme with disappearing ciphertexts} on the class $\mathcal{C}_\lambda$ consists of the following algorithms: 
\begin{itemize}
    \item $\textsf{Setup}(1^\lambda):$ Outputs a classical master key $\mk$. 
    \item $\textsf{KeyGen}(\mk,C):$ Takes as input the master key $\mk$ and circuit $C\in \mathcal{C}_\lambda$ and outputs a functional key ${\fk}_C$.
    \item $\textsf{Enc}(\mk,\mu):$ Takes as input {the master key} $\mk$ and a message $\mu \in \mathcal{X}_\lambda$ and outputs a ciphertext stream $\ct$.
    \item $\textsf{Dec}({\fk}_C, \ct)$: {{Takes a functional key ${\fk}_C$ available before the ciphertext stream $\ct$ begins, evaluates the stream during transmission, and outputs a value $y\in\mathcal{Y}_\lambda\cup\{\bot\}$.}}
\end{itemize}

\emph{Correctness} requires that for any circuit $C\in  \mathcal{C}_\lambda$ and message $\mu \in \mathcal{X}_\lambda$,
\begin{align*} \Pr{\left[
\begin{tabular}{c|c}
 \multirow[c]{3}{*}{$\textsf{Dec}({\fk}_C,\ct)=C(\mu)\ $} &   $\mk\ \leftarrow \textsf{Setup}(1^\lambda)$\\
 & ${\fk}_C \leftarrow \textsf{KeyGen}(\mk,C)$ \\ 
 & $\ct\ \leftarrow \textsf{Enc}(\mk,\mu)$\\
 \end{tabular}\right]} \geq 1-\negl[\lambda] .
\end{align*}
\end{definition}

We define adaptive simulation-based security through the use of two experiments: \textsf{Real Experiment} and \textsf{Ideal Experiment}. We require these experiments to be indistinguishable. 

Our definition closely follows the standard adaptive simulation-based framework \cite{BSW11}, with a fundamental departure: in our \textsf{Real} experiment, the adversary is granted the functional key for the identity function {{after all challenge ciphertext streams and their prescribed memory restrictions have ended.}} Crucially, unlike the standard definition where the simulator would also receive the message to maintain simulation feasibility, our simulator {receives no additional evaluations from this later identity-key disclosure.} This gap is specifically designed to capture the property of \textbf{disappearing ciphertexts}. It implies that the adversary can no longer apply the newly received identity functional key {to recover additional information from the previously issued challenge ciphertexts.} 

We now present the two experiments where the simulator consists of a tuple of algorithms $\Sim=(\textsf{Setup}^*,\textsf{KeyGen}_1^*,\textsf{Enc}^*,\textsf{KeyGen}_2^*)$ and adversary $\adv=(\adv_1,\adv_2)$.{ Here $\adv_1$ chooses a batch of $q=\poly[\lambda]$ messages. The notation $\adv_2\langle\tau,(\ct_i)_{i\in[q]},\fk_I\rangle$ denotes a stateful execution: $\adv_2$ starts with $\tau$ and processes each ciphertext stream as it is generated; only after every stream and its prescribed memory restrictions have ended does it receive $\fk_I$. No further ciphertexts are issued afterward. Thus, the tuples below specify temporal delivery, not simultaneous access. Each experiment outputs $\alpha$. If $\adv_1$ makes no functional key query, set $C^*=\bot$ and all evaluation records $\mathcal{V}_i,\mathcal{V}$ to $\emptyset$.{ In this case, use a dummy $y^*=0^\ell$ in the hybrid descriptions below and interpret every point program there as the constant-$\bot$ program.}}

\smallskip \noindent\fbox{%
    \parbox{\textwidth}{%
\textbf{Real Experiment} $\textsf{DC-FE}^{\textsf{Real}}_{\Pi,\adv,\Sim}({\lambda})$:
\begin{enumerate}
    \item Sample $\mk \leftarrow \textsf{Setup}(1^\lambda)$.
   \item ${((m_i)_{i\in[q]},\tau)\leftarrow} \adv_1^{1\textsf{KeyGen}(\mk,\cdot)}$. 
       \item {For each $i\in[q]$, generate and transmit $\ct_i\leftarrow\textsf{Enc}(\mk,m_i)$.}
                                \item Compute $\fk_I\gets  \textsf{KeyGen}(\mk, I)$.
 \item {$\alpha\leftarrow \adv_2\langle\tau,(\ct_i)_{i\in[q]},\fk_I\rangle$.} 
    \end{enumerate}}}
    \smallskip

\smallskip \noindent\fbox{%
    \parbox{\textwidth}{%
\textbf{Ideal Experiment} $\textsf{DC-FE}^{\textsf{Ideal}}_{\Pi,\adv,\Sim}({\lambda})$:
\begin{enumerate}
    \item Sample $\mk^* \leftarrow \textsf{Setup}^*(1^\lambda)$.
   \item ${((m_i)_{i\in[q]},\tau)\leftarrow} \adv_1^{1\textsf{KeyGen}_1^*(\mk^*,\cdot)}$. Let $C^*$ be the query made to $\textsf{KeyGen}_1^*$ by $\adv_1$ and {let $\mathcal{V}_i\coloneqq(C^*,C^*(m_i))$ and $\mathcal{V}\coloneqq(C^*,(C^*(m_i))_{i\in[q]})$.} 
   \item {For each $i\in[q]$, generate and transmit $\ct_i^*\leftarrow\textsf{Enc}^*(\mk^*,\mathcal{V}_i,1^{\lvert m_i\rvert})$.} 
                            \item Compute $\fk_I^*\gets  \textsf{KeyGen}_2^*({\mk^*}, \mathcal{V})$.
 \item {$\alpha\leftarrow \adv_2\langle\tau,(\ct_i^*)_{i\in[q]},\fk_I^*\rangle$.} 
    \end{enumerate}}}
    \smallskip

\begin{definition}[${\textsf{DC-SIM}}$ Security]
    A $\FE$ scheme $\Pi$ satisfies $\textsf{DC-SIM}$ security if there exists a QPT simulator $\Sim$ such that {for any QPT adversary $\adv$ satisfying the memory bound, the following distributions} are computationally indistinguishable:
    \begin{align*}
         \textsf{DC-FE}^{\textsf{Real}}_{\Pi,\adv,\Sim}({\lambda}) \approx_c \textsf{DC-FE}^{\textsf{Ideal}}_{\Pi,\adv,\Sim}({\lambda}).
    \end{align*}
\end{definition}

Note that the definition considers QPT algorithms. {We instead consider PPT algorithms satisfying the classical memory bound when discussing security in the \textsf{BCSM}.} 

\subsection{Construction}

\begin{construct}[\textsf{FE} with Disappearing Ciphertexts]
\label{con:fe 2}
{\small Let $\mathcal{C}_\lambda$ be a class of classical deterministic polynomial circuits {with fixed injective encodings of length $\ell=\poly[\lambda]$, where $\ell\geq\lambda$, padding if necessary. Identify each encoding, including that of $I$, with an element of $\mathbb{F}_{2^\ell}$. All additions and multiplications below are over this field, and malformed inputs are rejected.} {{Let $\mathcal{O}$ be a disappearing $\GB$ over classical polynomial circuits, secure against adversaries with memory bound $\s$.}} 
The construction for \textsf{FE} on $n$-bit messages over $\mathcal{C}_\lambda$ is as follows:

\begin{itemize}
    \item $\textsf{Setup}(1^\lambda):$ Sample $a,b\gets {\{0,1\}^\ell}$ and let this be the master key $\mk\coloneqq (a,b)$.

\item $\textsf{KeyGen}(\mk,C):$ Output $\fk_C\coloneqq (C,y \coloneqq a\cdot C+b).$

\item $\textsf{Enc}({\mk},m):$ Output $\mathcal{O} (E_{\mk,m})$, where $E_{\mk,m}$ is a program that on input $(C,y')$, does as follows:
\begin{enumerate}
    \item Compute $y\coloneqq a\cdot C+b$. 
    \item If $y=y'$, output $C(m)$. Otherwise, output $\bot$. 
\end{enumerate}

\item $\textsf{Dec}({\fk}_C,\ct):$ 
{Evaluate the obfuscation stream $\ct$ on the previously obtained $\fk_C$ and output the result.}

\end{itemize}}\end{construct}

For each message length, the real ciphertext programs, the point
programs used in the hybrids, and the constant-$\perp$ programs
use identical input and output formats and are padded to a common
public circuit size. This size depends only on the public parameters
and the message length, and not on the message, master key, queried
circuit, or evaluation value.

\begin{theorem}
    Construction \ref{con:fe 2} {{is an $\textsf{DC-SIM}$ secure \textsf{FE} scheme over $\mathcal{C}$ assuming the existence of disappearing $\GB$ over classical polynomial circuits.}}
\end{theorem}

\begin{proof}
To show security, we need to construct a QPT simulator $\Sim$ such that for any QPT $\adv$,
    \begin{align*}
         \textsf{DC-FE}^{\textsf{Real}}_{\Pi,\adv,\Sim}({\lambda}) \approx_c \textsf{DC-FE}^{\textsf{Ideal}}_{\Pi,\adv,\Sim}({\lambda}).
    \end{align*}

    We prove security through a sequence of indistinguishable hybrids, starting with the $\textsf{Real Experiment}$ and ending with the $\textsf{Ideal Experiment}$.

\begin{itemize}
    \item \textbf{Hybrid} $\Hy_0:$ This is the $\textsf{DC-FE}^{{\textsf{Real}}}_{\Pi,\adv,\Sim}({\lambda})$ experiment. 

\smallskip \noindent\fbox{%
    \parbox{0.95\textwidth}{%
\textbf{Experiment} $\textsf{DC-FE}^{\textsf{Real}}_{\Pi,\adv,\Sim}({\lambda})$:
\begin{enumerate}
    \item Sample $\mk \leftarrow \textsf{Setup}(1^\lambda)$.
   \item ${((m_i)_{i\in[q]},\tau)\leftarrow} \adv_1^{1\textsf{KeyGen}(\mk,\cdot)}$.  Let $C^*$ be the query made to ${\textsf{KeyGen}}$ by $\adv_1$ and {let $\mathcal{V}_i\coloneqq(C^*,C^*(m_i))$ and $\mathcal{V}\coloneqq(C^*,(C^*(m_i))_{i\in[q]})$.} 
       \item {For each $i\in[q]$, generate and transmit $\ct_i\leftarrow\textsf{Enc}(\mk,m_i)$.}
                         \item Compute $\fk_I\gets  \textsf{KeyGen}(\mk, I)$.
 \item {$\alpha\leftarrow \adv_2\langle\tau,(\ct_i)_{i\in[q]},\fk_I\rangle$.} 
    \end{enumerate}}}
    \smallskip 

     \item \textbf{Hybrid} $\Hy_1:$ {Same as hybrid $\Hy_0$, except for each $i\in[q]$ we switch $\mathcal{O}(E_{\mk,m_i})$ with $\mathcal{O}(P_{(C^*,y^*),C^*(m_i)})$, where $P_{(C^*,y^*),C^*(m_i)}$ is the point function that sends $(C^*,y^*)$ to $C^*(m_i)$ and everything else to $\bot$.}  

        \smallskip \noindent\fbox{%
    \parbox{0.95\textwidth}{%
\textbf{Hybrid} $\Hy_1({\lambda})$:
\begin{enumerate}
    \item Sample $\mk \leftarrow \textsf{Setup}(1^\lambda)$.
\item ${((m_i)_{i\in[q]},\tau)\leftarrow} \adv_1^{1\textsf{KeyGen}(\mk,\cdot)}$. Let $C^*$ be $\adv_1$'s query and let $(C^*,y^*)$ be the response. 
       \item {For each $i\in[q]$, generate and transmit 
}       $${\ct_i\gets\mathcal{O}(P_{(C^*,y^*),C^*(m_i)})}.$$
                  \item Compute $\fk_I\gets \textsf{KeyGen}(\mk, I)$.
 \item {$\alpha\leftarrow \adv_2\langle\tau,(\ct_i)_{i\in[q]},\fk_I\rangle$.} 
    \end{enumerate}}}
    \smallskip 

        \item \textbf{Hybrid} $\Hy_2:$ Same as hybrid $\Hy_1$, except we generate the functional keys in the experiment by sampling random strings.

        \smallskip \noindent\fbox{%
    \parbox{0.95\textwidth}{%
\textbf{Hybrid} $\Hy_2({\lambda})$:
\begin{enumerate}
    \item Sample $\mk \leftarrow \textsf{Setup}(1^\lambda)$.
   \item ${((m_i)_{i\in[q]},\tau)\leftarrow} \adv_1^{1\textsf{KeyGen}'(\mk,\cdot)}$, where $\textsf{KeyGen}'$ is defined as follows:
  \smallskip \noindent\fbox{%
    \parbox{0.8\textwidth}{%
$\textsf{KeyGen}'(\mk,C^*)$:
   \begin{enumerate}
       \item Sample $y^*\gets {\{0,1\}^\ell}$.
       \item Output $(C^*,y^*)$.
       \end{enumerate}}}
       \item {For each $i\in[q]$, generate and transmit 
}      $${\ct_i\gets\mathcal{O}(P_{(C^*,y^*),C^*(m_i)})}.$$
           \item If $C^*=I$, set $\fk_I\coloneqq (I,y^*)$. Otherwise, sample $y_I\gets {\{0,1\}^\ell}$, {and set $\fk_I\coloneqq(I,y_I)$.}
 \item {$\alpha\leftarrow \adv_2\langle\tau,(\ct_i)_{i\in[q]},\fk_I\rangle$.} 
    \end{enumerate}}}
    \smallskip 
\end{itemize}

We are now ready to define the algorithms of the simulator $\mathcal{S}$:

    \smallskip \noindent\fbox{%
    \parbox{\textwidth}{%
    \text{Algorithms of $\mathcal{S}$:}
\begin{itemize}
\item $\textsf{Setup}^*(1^\lambda):$ 
\begin{enumerate}
    \item Sample $y^*,y_I\gets {\{0,1\}^\ell}$. 
    \item Output $\mk^*\coloneqq (y^*,y_I)$.
\end{enumerate} 
    \item $\textsf{KeyGen}_1^*(\mk^*, C^*):$ Output $(C^*,y^*)$.
    \item  {$\textsf{Enc}^*(\mk^*,\mathcal{V}_i,1^{\lvert m_i\rvert}):$ If $\mathcal{V}_i=\emptyset$, output an obfuscation of the constant-$\bot$ program on inputs $(C,y)$. Otherwise, parse $\mathcal{V}_i=(C^*,v_i)$ and output $\ct_i\gets\mathcal{O}(P_{(C^*,y^*),v_i})$. The program sizes are padded as above using the public parameters and supplied message length.}
        \item {$\textsf{KeyGen}_2^*(\mk^*,\mathcal{V}):$ If $\mathcal{V}=\emptyset$, output $\fk_I\coloneqq(I,y_I)$. Otherwise, parse $\mathcal{V}=(C^*,(v_i)_{i\in[q]})$. If $C^*=I$, output $\fk_I\coloneqq(I,y^*)$; otherwise, output $\fk_I\coloneqq(I,y_I)$.}
\end{itemize}}}
    \smallskip     

With this definition, it is easy to check that $\Hy_2$ is the same as $\textsf{DC-FE}^{{\textsf{Ideal}}}_{\Pi,\adv,\Sim}({\lambda})$. Hence, it is sufficient to show that hybrids $\Hy_0$ is indistinguishable from $\Hy_2$ to prove security.  

We show this by proving that no QPT adversary can distinguish between any two consecutive hybrids and then apply the triangle inequality.

\begin{Claim}
No QPT adversary can distinguish between hybrids $\Hy_{0}$ and $\Hy_{1}$ with non-negligible probability.
\end{Claim}
\begin{proof}
{{If an initial key was requested, the only difference between $\Hy_{0}$ and $\Hy_{1}$ is that each ciphertext $\ct_i$ is an obfuscation of $P_{(C^*,y^*),C^*(m_i)}$, where $C^*$ is the query of $\adv_1$ and $y^*=a\cdot C^*+b$, instead of an obfuscation of $E_{\mk,m_i}$. If no initial key was requested, each $E_{\mk,m_i}$ is instead replaced by the constant-$\bot$ program.}} 

{{By the joint disappearing $\GB$ guarantee, a common stateful simulator can simulate all ciphertext streams in either hybrid using at most $Q=\poly[\lambda]$ classical oracle queries in total. Each query must be made during the corresponding stream, before the later disclosure of $\fk_I$. We compare the two oracle simulations using the same initial key, messages, and later identity key.}}

{{When an initial key was requested, for each $i$ the two programs agree at $(C^*,y^*)$ and can differ only at a valid input $(C,a\cdot C+b)$ with $C\neq C^*$. Conditioned on $(C^*,y^*)$, the tag of each other circuit is uniform in $\mathbb{F}_{2^\ell}$. To handle adaptive queries, run the simulator against the point programs. Before the identity key is disclosed, this query sequence depends on the master key only through $(C^*,y^*)$, and not on the remaining unknown field element $a$. Each candidate valid input with $C\neq C^*$ determines a single value of $a$. A union bound over all $Q$ queries therefore bounds the probability of any differing answer by $Q/2^\ell$. If no initial key was requested, the comparison uses constant-$\bot$ programs; the query sequence is independent of $(a,b)$, and each candidate has a valid tag with probability $2^{-\ell}$, giving the same bound. Until a differing answer occurs, the two oracle simulations agree. The identity key is disclosed only after all oracle access ends and can be coupled identically in both simulations. Since $\ell\geq\lambda$, the difference is negligible, and  $\GB$ security proves the claim.}}
    \qed
\end{proof}

\begin{Claim}
No QPT adversary can distinguish between hybrids $\Hy_1$ and $\Hy_2$. 
\end{Claim}
\begin{proof}
The only difference between these hybrids is that the functional keys of $C^*, I$ in hybrid $\Hy_1$ are computed using the master key as follows: $\fk_{C^*}=(C^*,a\cdot C^*+b),\fk_I=(I,a\cdot I+b)$, while {their tags in $\Hy_2$ are sampled uniformly from} ${\{0,1\}^\ell}$. 

{If an initial key was requested and $C^*\neq I$, the pair $(a\cdot C^*+b,a\cdot I+b)$ is uniform in $\mathbb{F}_{2^\ell}^2$, since the two circuit encodings are distinct. If $C^*=I$, the two keys use the same tag $y^*$, as prescribed in $\Hy_2$. The point programs depend on the master key only through $y^*$. If no initial key was requested, they are constant-$\bot$ programs and the identity tag is uniform. Therefore, the full distributions of the two hybrids are identical.}  
    \qed
\end{proof}

All in all, we have shown, by the triangle inequality, that no QPT adversary can distinguish between hybrid $\Hy_0$, which is $\textsf{DC-FE}^{{\textsf{Real}}}_{\Pi,\adv,\Sim}({\lambda})$, and hybrid $\Hy_2$ which is the same as $\textsf{DC-FE}^{{\textsf{Ideal}}}_{\Pi,\adv,\Sim}({\lambda})$. Therefore,
   \begin{align*}
         \textsf{DC-FE}^{\textsf{Real}}_{\Pi,\adv,\Sim}({\lambda}) \approx_c \textsf{DC-FE}^{\textsf{Ideal}}_{\Pi,\adv,\Sim}({\lambda}).
    \end{align*}
\qed    
\end{proof}

%% file: appendix_audit_reviewed.tex
\providecommand{\QFEreplace}[2]{}
\renewcommand{\QFEreplace}[2]{%
  \begingroup
  \let\QFEoriginalcref\cref
  \let\QFEoriginalcite\cite
  \renewcommand{\cref}[1]{\mbox{\QFEoriginalcref{##1}}}%
  \renewcommand{\cite}[1]{\mbox{\QFEoriginalcite{##1}}}%
  \sout{#1}\endgroup\ {\color{red}#2}}
\providecommand{\QFEinsert}[1]{{\color{red}\relax#1}}
\providecommand{\QFEmathreplace}[2]{\text{\sout{\ensuremath{#1}}}\;{\color{red}#2}}

\providecommand{\QFEstackreplace}[2]{\begin{gathered}\text{\sout{\ensuremath{#1}}}\\{\color{red}#2}\end{gathered}}

\section{Public-Key Encryption with Many-Key Security}

In this section, we {define public-key encryption with classical keys and quantum ciphertexts, adapting the encryption syntax of \cite{BL19}} and show that {any simulation-secure public-key $\QFE$ scheme with classical public and master secret keys, supporting the identity circuit, implies} the existence of a $\PKE$ scheme satisfying a variant of \textsf{CPA}--security that we call many-key \textsf{CPA}--security.

\begin{definition}
    Let $\lambda$ be the security parameter and {$\ell=\ell(\lambda)$ be a polynomial} in $\lambda$. A public-key  encryption $(\PKE)$ scheme consists of a tuple of QPT algorithms $(\Keygen, \Enc, \Dec)$ such that,  
\begin{itemize}[align=left]
    \item[$\Keygen(1^\lambda) \rightarrow (sk,pk)$] Samples a classical secret key $sk$ and classical public key $pk$.
    \item[$\Enc(pk, m) \rightarrow \ct$] Given the public key $pk$ and a message $m$ output a (quantum) ciphertext $\ct$.
    \item[$\Dec(sk, \ct) \rightarrow m$] Given a secret key $sk$ and ciphertext $\ct$, output a message $m$.
\end{itemize}
\end{definition}

\begin{definition}(Correctness of a $\PKE$)
    For all $\lambda \in \N$ and $m \in {\bin^{\ell}}$
    $$ {\Pr\left[\Dec(sk,\Enc(pk,m))=m\right]} \geq 1 - \negl[\lambda]$$
    where $(sk,pk) \leftarrow \Keygen(1^{\lambda}) $.
\end{definition}

\begin{definition}[\textsf{CPA} Security]
    A $\PKE$ scheme $\Pi$ satisfies {\textsf{CPA} indistinguishability} if for any QPT $\adv$ in security experiment $\PKE^{\textnormal{IND-CPA}}_{\Pi,\adv}$ (see \Cref{fig:pkt-cpa}), there exists a negligible function $\negl[\lambda]$ such that
    \begin{align*}
         \Pr[\PKE^{\textnormal{IND-CPA}}_{\Pi,\adv}({\lambda})=1] \leq \frac{1}{2} + \negl[\lambda].
    \end{align*}
\end{definition}

\begin{figure}[!htb]
   \begin{center} 
   \begin{tabular}{|p{12cm}|}
    \hline 
\begin{center}
\underline{$\PKE^{\textnormal{IND-CPA}}_{\Pi,\adv}({\lambda})$}: 
\end{center}
\begin{enumerate}
    \item Challenger $\ch$ generates $(sk,pk)\gets \Keygen(1^\lambda)$ and sends $pk$ to $\adv$.
    \item $\adv$ outputs a message pair $(m_{0},m_{1})\in {\{0,1\}^{\ell}\times\{0,1\}^{\ell}}$.
    \item $\ch$ samples a bit $b\leftarrow  \{0,1\}$ and generates $\ct\gets \Enc(pk,m_{b})$. 
    \item $\ch$ sends $\ct$ to $\adv$.
    \item $\adv$ submits a bit $b'$.
    \item The output of the experiment is $1$ if $b=b'$, and $0$ otherwise. 
\end{enumerate}
\ \\ 
\hline
\end{tabular}
    \caption{\textsf{CPA} security experiment.}
    \label{fig:pkt-cpa}
    \end{center}
\end{figure}

{A simulation-secure public-key $\QFE$ scheme with classical public and master secret keys yields a \textsf{CPA}--secure $\PKE$ scheme by restricting to classical messages and using the master secret key to generate an identity functional key during decryption. We will use the following equivalent formulation of \textsf{CPA}--security, which we term} \emph{many-key \textsf{CPA} (\textsf{MK-CPA}) security}. {Intuitively, this encrypts one message from each adversarially chosen pair under independently generated public keys, using a common challenge bit.} 

\begin{definition}[\textsf{MK-CPA} Security]
    A $\PKE$ scheme $\Pi$ satisfies \emph{many-key \textsf{CPA} indistinguishability} if for any QPT $\adv$ in security experiment $\PKE^{\textnormal{IND-MK-CPA}}_{\Pi,\adv}$ (see \Cref{fig:mk-cpa}), there exists a negligible function $\negl[\lambda]$ such that
    \begin{align*}
         \Pr[\PKE^{\textnormal{IND-MK-CPA}}_{\Pi,\adv}({\lambda})=1] \leq \frac{1}{2} + \negl[\lambda].
    \end{align*}\label{def:qpke}
\end{definition}

\begin{figure}[!htb]
   \begin{center} 
   \begin{tabular}{|p{12cm}|}
    \hline 
\begin{center}
\underline{$\PKE^{\textnormal{IND-MK-CPA}}_{\Pi,\adv}({\lambda})$}: 
\end{center}
\begin{enumerate}
    \item {$\adv$ sends a positive integer $n$ that is polynomial in $\lambda$.}
    \item Challenger $\ch$ generates $(sk_i,pk_i)\gets \Keygen(1^\lambda)$ for every $i\in [n]$, and sends $(pk_i)_{i\in[n]}$ to $\adv$. 
    \item $\adv$ outputs a message pair $(m_{0,i},m_{1,i})\in {\{0,1\}^{\ell}\times\{0,1\}^{\ell}}$ for every $i\in [n]$.
    \item $\ch$ samples a bit $b\leftarrow  \{0,1\}$ and generates $\ct_i\gets \Enc(pk_i,m_{b,i})$ for every $i\in [n]$. 
    \item $\ch$ sends $(\ct_i)_{i\in [n]}$ to $\adv$.
    \item $\adv$ submits a bit $b'$.
    \item The output of the experiment is $1$ if $b=b'$, and $0$ otherwise. 
\end{enumerate}
\ \\ 
\hline
\end{tabular}
    \caption{\textsf{MK-CPA} security experiment.}
    \label{fig:mk-cpa}
    \end{center}
\end{figure}

\begin{lemma}
\label{lem:many-key}
    If a $\PKE$ scheme is \textsf{CPA}--secure then it is \textsf{MK-CPA}--secure. 
\end{lemma}
\begin{proof}
Let $\Pi=(\Keygen,\Enc,\Dec)$ be a $\PKE$ scheme. Assume the scheme is \textsf{CPA}--secure but not \textsf{MK-CPA} secure for the sake of obtaining a contradiction. Then, there exists a QPT adversary $\adv$ such that 
\begin{align*}
         \Pr[\PKE^{\textnormal{IND-MK-CPA}}_{\Pi,\adv}({\lambda})=1] > \frac{1}{2}+\frac{1}{\poly[\lambda]}
    \end{align*}
for infinitely many $\lambda$.
This implies that in the security experiment $\PKE^{\textnormal{IND-MK-CPA}}_{\Pi,\adv}({\lambda})$,
\begin{equation}
\begin{split}
 \left| \Pr[\adv(\right.\hspace{-0.1cm}&\left. \{\ct_i\}_{i\in [n]} )=1  :\ct_i\gets   \Enc(pk_i,m_{0,i})\,\, \forall i\in [n]] \right. \\
 &\quad \left. -\Pr[\adv(\{\ct_i\}_{i\in [n]})=1:\ct_i\gets \Enc(pk_i,m_{1,i})\ \forall  i\in [n]]\right| >\delta
\end{split}
\end{equation}
for some function $\delta=\delta(\lambda) \in \frac{1}{\poly[\lambda]} $ and infinitely many $\lambda$.

{Fix a polynomial $N=N(\lambda)$ such that $1\leq n\leq N$ in every execution of $\adv$. This bound is fixed for $\adv$ and does not depend on its random choices.}

{Define $\Hy_j$ for $j\in\{0,\ldots,N\}$} to be a hybrid of the security experiment $\PKE^{\textnormal{IND-MK-CPA}}_{\Pi,\adv}({\lambda})$, but where the challenge ciphertexts are generated instead in the following way. {If $i>j$, then} $\ct_i\gets \Enc(pk_i,m_{0,i})$ {and if $i\leq j$, then} $\ct_i\gets \Enc(pk_i,m_{1,i})$. Let the output  of the hybrid be the output of $\adv$ when run on these ciphertexts. 

{For $j\in[N]$, let}
\begin{align}
    \epsilon_j\coloneqq \left| \Pr[\Hy_{j}=1]-\Pr[\Hy_{j-1}=1] \right|.
\end{align}
Set  $j^*=\textsf{arg\ max}_j\epsilon_j$. We have
\begin{align}
    \delta&<\left| \Pr[{\Hy_{N}}=1]-\Pr[{\Hy_{0}}=1] \right|\\
         &\leq {\sum_{j=1}^{N}}\left| \Pr[\Hy_{j}=1]-\Pr[\Hy_{j-1}=1] \right|\\
         &\leq {N\cdot \epsilon_{j^*}}.
\end{align}

{We can now design a QPT algorithm $\D$ that breaks \textsf{CPA}--security without finding $j^*$: it samples $j$ uniformly from $[N]$ and uses the signed average below.} 


\begin{figure}[!htb]
   \begin{center} 
   \begin{tabular}{|p{12cm}|}
    \hline 
\begin{center}
{Algorithm of $\D$ in $\PKE^{\textnormal{IND-CPA}}_{\Pi,\D}({\lambda})$}: 
\end{center}
\begin{itemize}
        \item Challenger $\ch$ generates $(sk,pk)\gets \Keygen(1^\lambda)$ and samples $b\gets \{0,1\}$. $\ch$ sends $pk$ to $\D$.
    \item $\D$ runs $\adv$, which first outputs an integer $n$.
        \item {$\D$ chooses $j\gets[N]$ uniformly at random. If $j>n$, it submits $(0^\ell,0^\ell)$ to $\ch$, ignores the challenge ciphertext, outputs a uniform bit, and halts. Otherwise, it sets $pk_j\coloneqq pk$.} 
            \item $\D$ samples $(sk_i,pk_i)\gets \Keygen(1^\lambda)$ for every $i\in {[n]\setminus\{j\}}$, and sends $(pk_i)_{i\in [n]}$ to $\adv$.
            \item $\adv$ outputs $(m_{0,i},m_{1,i})\in {\{0,1\}^{\ell}\times\{0,1\}^{\ell}}$ for every $i\in [n]$.
    \item $\D$ generates $\ct_i\gets \Enc(pk_i,m_{0,i})$ for every $i>j$ and $\ct_i\gets \Enc(pk_i,m_{1,i})$ for every $i< j$. 
    \item $\D$ sends $(m_{0,j},m_{1,j})$ to $\ch$. 
    \item $\ch$ generates $\ct\gets \Enc(pk,m_{b,j})$ and sends the result to $\D$. 
    \item $\D$ sets $\ct_j\coloneqq \ct$ and sends $(\ct_i)_{i\in [n]}$ to $\adv$.
    \item $\adv$ submits a bit $b'$.
    \item $\D$ outputs $b'$.
        \end{itemize}
\ \\ 
\hline
\end{tabular}
    \end{center}
    \caption{Attack against \textsf{CPA}--security}
\end{figure}

{{For a fixed $j$, on executions with $j\leq n$, $\D$ simulates $\Hy_{j-1}$ if $b=0$ and $\Hy_j$ if $b=1$. When $j>n$, these adjacent hybrids coincide, and the uniform output of $\D$ also contributes zero to the signed difference. All probabilities include the random choice of $n$ by $\adv$. The endpoint difference is positive by the assumed success probability of $\adv$. Averaging uniformly over $j\in[N]$ therefore gives}}
\[
{\begin{aligned}
&\Pr[\D=1\mid b=1]-\Pr[\D=1\mid b=0]\\
&\quad=\frac{1}{N}\sum_{j=1}^{N}\left(\Pr[\Hy_j=1]-\Pr[\Hy_{j-1}=1]\right)\\
&\quad=\frac{\Pr[\Hy_N=1]-\Pr[\Hy_0=1]}{N}>\frac{\delta}{N}.
\end{aligned}}
\]
{for infinitely many $\lambda$. Hence, $\D$ wins the \textsf{CPA} experiment with probability greater than $1/2+\delta/(2N)$, contradicting \textsf{CPA}--security.} Therefore, $\Pi$ must satisfy \textsf{MK-CPA}--security. 
\qed
\end{proof}

As a result of \cref{lem:many-key}, we obtain the following corollary. 

\begin{corollary}
\label{cor:many-key-from-QFE}
    {If a simulation-secure public-key $\QFE$ scheme with classical public and master secret keys exists for a circuit class containing the identity, then a \textsf{MK-CPA}--secure $\PKE$ scheme exists.} 
\end{corollary}{

To see this, keep the public key and encryption algorithm of $\QFE$ and use its classical master secret key as the $\PKE$ secret key. To decrypt, generate a fresh functional key for the identity circuit and apply $\QFE$ decryption. Correctness follows from $\QFE$ correctness. With no functional key queries, simulation security hides the classical challenge message, yielding \textsf{CPA}--security; \cref{lem:many-key} then gives \textsf{MK-CPA}--security.}

